\documentclass[12pt, oneside]{amsart}
\pdfoutput=1
\usepackage{geometry}
\usepackage{lmodern}
\usepackage{amsmath}
\usepackage{amssymb}
\usepackage{amsthm}
\usepackage{color}
\usepackage[usenames,dvipsnames]{xcolor}
\usepackage{graphicx}
\usepackage{subfig}
\usepackage{soul}
\usepackage[countmax]{subfloat}
\usepackage{float}
\usepackage{rotfloat}
\usepackage{comment}
\usepackage{setspace}
\usepackage{esint}
\usepackage[authoryear]{natbib}
\definecolor{MyDarkBlue}{rgb}{0,0.08,0.45}
\usepackage[unicode=true,pdfusetitle, bookmarks=true,bookmarksnumbered=false,bookmarksopen=false, breaklinks=false,pdfborder={0 0 1},backref=section,colorlinks=true, linkcolor = MyDarkBlue, citecolor = MyDarkBlue]{hyperref}
\usepackage{breakurl}

\PassOptionsToPackage{normalem}{ulem}
\usepackage{ulem}

\makeatletter

\numberwithin{equation}{section}

\newtheorem{theorem}{{\bf\sc Theorem}}
\newtheorem{lemma}{{\bf\sc Lemma}}

\newtheorem{corollary}{{\bf\sc Corollary}}

\newtheorem{thm}{{\bf \sc Theorem}}[section]
\newtheorem{lem}{{\bf \sc Lemma}}[section]

\newtheorem{prop}{{\bf\sc Proposition}}[section]

\DeclareMathOperator*{\plim}{plim}
\DeclareMathOperator*{\argmin}{argmin}
\DeclareMathOperator*{\argmax}{argmax}
\DeclareMathOperator*{\cvgto}{\stackrel{\mathrm{P}}{\longrightarrow}}
\DeclareMathOperator*{\wkto}{\rightsquigarrow}
\providecommand{\E}{\mathrm{E}}

\providecommand{\Prob}{\mathrm{P}}

\providecommand{\gammahat}{\widehat{\gamma}}
\providecommand{\betahat}{\widehat{\beta}}
\providecommand{\phat}{\widehat{p}}
\providecommand{\Stil}{\widetilde{S}}
\renewcommand{\Pr}{\Prob}

\renewenvironment{proof}[1][Proof]{\noindent\text{#1.} }{\ \rule{0.5em}{0.5em}}

\makeatother

\begin{document}

\title{Tractable and Consistent Random Graph Models}

\author{Arun G. Chandrasekhar$^{\ddagger}$}
\author{Matthew O. Jackson$^{\star}$ }
\date{December 2011, Revision: June 2014}

\thanks{We thank Isaiah Andrews, Larry Blume,
Gabriel Carroll, Victor Chernozhukov, Esther Duflo, Marcel Fafchamps, Ben Golub, Bryan Graham, Bo Honore,  Randall Lewis,
Angelo Mele, Eduardo Morales, Stephen Nei, Elie Tamer, Juan Pablo Xandri and Yiqing Xing for helpful discussions and/or comments on earlier drafts, and especially Andres Drenik for valuable research assistance. We thank participants at the Simons Workshop at Berkeley on Unifying Theory and Experiments for Large-Scale Networks, the Princeton Econometrics Seminar, and the Second Economics of Networks Conference (Essex).
Chandrasekhar thanks the NSF Graduate Research Fellowship Program. Jackson gratefully acknowledges financial support from the NSF under grants SES-0961481 and SES-1155302 and from grant FA9550-12-1-0411
from the AFOSR and DARPA, and ARO MURI award No. W911NF-12-1-0509. }
\thanks{$^{\ddagger}$Department of Economics, Stanford University; NBER}
\thanks{$^{\star}$Department of Economics, Stanford University; Santa Fe Institute; and CIFAR}

\begin{abstract}
We define a general class of network formation models, Statistical Exponential Random Graph Models (SERGMs), that nest standard exponential random graph models (ERGMs) as a special case.
We provide the first general results on when these models' (including ERGMs) parameters estimated from the observation of a single network are consistent (i.e., become accurate as the number of nodes grows).
Next, addressing the problem that standard techniques of estimating ERGMs have been shown to have exponentially slow mixing times for many specifications, we show that by reformulating network formation as a distribution over the space of sufficient statistics instead of the space of networks, the size of the space of estimation can be greatly reduced, making estimation practical and easy.
We also develop a related, but distinct, class of models that we call subgraph generation models (SUGMs) that are useful for modeling sparse networks and whose parameter estimates are also directly and easily estimable, consistent, and asymptotically normally distributed.
Finally, we show how choice-based (strategic) network formation models can be written as SERGMs and SUGMs, and apply our models and techniques to network data from rural Indian villages.

\textsc{JEL Classification Codes:} D85, C51, C01, Z13.

\textsc{Keywords:} Random Networks, Random Graphs, Exponential Random Graph Models, Exponential Family, Social Networks, Network Formation, Consistency, Sparse Networks
\end{abstract}

\thispagestyle{empty}
\setcounter{page}{0}

\maketitle

\newpage

\section{Introduction}

\

\begin{quote}
...[A] pertinent form of statistical treatment would be one which deals with social configurations as wholes, and not with single series of facts, more or less artificially separated from the total picture.

{\emph {Jacob Levy Moreno and Helen Hall Jennings, 1938.}}
\end{quote}

\medskip

For a researcher interested in an economic or social interaction,
endogeneity of the interactions often makes network estimation essential.  That estimation is challenging both because relationships are generally not independent and because the researcher usually only observes one network.
A literature spanning several disciplines (computer science, economics, sociology and statistics) has turned to exponential random graph models (ERGMs) to meet these challenges.\footnote{See, e.g., \cite{frank1986markov,wasserman1996logit,mele2011}.} However, recently these models have come under fire as the maximum likelihood estimator of the parameters may not be computationally feasible nor consistent, and so the software being used may provide inaccurate parameter estimates.
In this paper, we develop two new classes of models that provide for (i) rich interdependencies, (ii) interdependencies that have economic and social micro-foundations, (ii) computationally feasible parameter estimates, (iii) and  consistent and asymptotically normal parameter estimates.

To begin with some background about why such models are needed, let us begin with an illustrative question.
To what extent is someone's proclivity to form relationships influenced by whether those relationships are in public or private?  For example, are people of different castes or races more reluctant to form relationships across types when they have a friend in common than when they do not?  This has implications for communication, learning, inequality, diffusion of innovations, and many other behaviors that are network-influenced.  Being able to statistically test whether people's tendencies to interact across groups depends on social context requires allowing for correlation in relationships within a network.
Beyond this illustrative question, correlations in relationships are important in many other social and economic settings:  from informal favor exchange where the presence of friends in common can facilitate robust favor exchange (e.g., \citet*{jackson2012quilts}),
to international trade agreements where the presence of one trade agreement can influence the formation of another (e.g., \citet{furusawa2007trade}).
Similarly, in forming a network of contacts in the context of a labor market, an individual benefits from relationships with others who are better-connected and hence relationships are not independently distributed (e.g., \citet{calvo2004job,calvo2005job}); nor are they in a setting of risk-sharing (e.g., \citet{bramoulle2007risk}).

Once such interdependencies exist, estimation of network formation cannot take place at the level of pairs of nodes, but must encompass the network as a whole.  ERGMs incorporate such interdependencies and thus have become the workhorse models for estimating network formation.\footnote{These grew from work on what were known as
Markov models (e.g., \citet{frank1986markov}) or $p*$ models (e.g., \citet{wasserman1996logit}).  An alternative approach is to work with regression models at the link (dyadic) level, but to allow for dependent error terms, as in the ``MRQAP'' approach (e.g., see \citet{krackhardt1988mrqap}).  That approach, however, is not well-suited for identifying the incidence of particular patterns of network relationships that may be implied by various social or economic theories of the type that we wish to address here.
There are also a set of growing random network models where one explicitly models a meeting process and a link formation algorithm (e.g., \citet*{barabasi1999emergence,jackson2007meeting,currarini2009economic},\citet{bramoulle-etal}), which can be estimated in some cases.  However, those are specific models with a couple of tunable parameters, and they are not designed or intended for the statistical testing of a wide variety of network formation models and hypotheses, which is the intention of the exponential formulation.
}
Indeed, as originally shown via a powerful theorem by  \citet{hammersley1971markov}, the exponential form can nest {\sl any} random graph model and can incorporate arbitrary interdependencies in connections.\footnote{Their theorem applies to undirected and unweighted networks. See the discussion in  \citet{jackson2008social}. Of course, the representation can become fairly complicated; but the point is that the ERGM model class is broadly encompassing.}
Moreover, ERGMs admit a variety of strategic (choice-based) network formation models, as we show below and others have shown in related contexts (e.g., \citet{mele2011}).

Let us be more explicit about the issues
that ERGMs face.
In an ERGM,  the probability
of observing  a network $g$ depends on an associated vector of statistics $S(g)$, that might include, for example, the density of links, number of cliques of given sizes, the average distance between nodes with various characteristics, counts of nodes with various degrees, and so forth.  The probability of the network is assumed to be proportional to
\[
\exp\left(\beta \cdot S\left(g\right) \right)
\]
where $\beta $ is a vector of model parameters.
Turning the above expression into a probability of observing network $g$
requires
normalizing this expression by summing across all possible networks, and so the probability of observing $g$ is
\begin{equation}
\label{intro}
\mathrm{P}_{\beta}\left(g\right)
= \frac{\exp\left(\beta \cdot S\left(g\right)\right)}{\sum_{g' } \exp\left(\beta \cdot S\left(g'\right)\right)}.
\end{equation}

ERGMs have become widely used because they provide an intuitive formulation focusing on key structural aspects that researchers believe are important in network formation and that can encode rich types of interdependencies. Recent work providing utility based micro-foundations has made these models even more desirable.
However, there are three critical challenges faced in working with ERGMs.

First, computing parameter estimates for an ERGM and drawing simulations from the distribution may be infeasible to do accurately in any nontrivial cases \citep{bhamidi2008mixing,chatterjee2010random}.
This is because estimating the likelihood of a given network requires having some estimate of relative likelihood of other networks that could have appeared instead, which involves explicitly or implicitly estimating the denominator of (\ref{intro}). 
Directly estimating the denominator is impossible: the number of possible networks on a given number of nodes is an exponential function of the number of nodes.
The adaptation of Markov Chain Monte Carlo (MCMC) sampling techniques to draw networks and estimate ERGMs, by \citet{snijders} and \citet{handcock2003assessing}, provided a seeming breakthrough and the subsequent development of computer programs based on those techniques led
to their widespread use.\footnote{See \citet{snijders2006new} for more discussion.}
However, it was clear to the developers and practitioners that the programs had convergence problems for many specifications of ERGMs.
Given the huge set of networks $g'$ to sample, any MCMC procedure can visit only an infinitesimal portion of the set, and until recently it was unclear whether such a technique would lead to an accurate estimate in any practical amount of time.  Unfortunately, important recent papers have shown that for broad classes of ERGMs standard MCMC procedures will
take exponential time to mix unless the links in the network are approximately independent
(e.g., see the discussions in \citet{bhamidi2008mixing} and \citet{chatterjee2010random}).   Of course, if links are approximately independent then there is no real need for an ERGM specification to begin with, and so in cases where ERGMs are really needed they cannot be accurately estimated by such MCMC techniques.
Such difficulties were well-known in practice to users of software programs that perform such estimations, as rerunning even simple models can lead to very different parameter and standard error estimates,
but now these difficulties have been proven to be more than an anomaly.

Second, setting aside the feasibility of estimation, there is also little that is known about the consistency of parameter estimates of ERGMs: would estimated parameters converge to the true parameters as the size of the network grows if those estimates were exactly computed? Given that data in many settings consist of a single network or a handful of networks,
we are interested in asymptotics where the number of nodes in a network grows. However, it may be the case that increasing the number of nodes does not increase the information in the system.  In fact, for some sequences of network statistics and parameters it is obvious that the parameters of the associated ERGM are {\sl not} consistent. For example, suppose that $S(g)$ includes a count of the number of components in the network and the parameters are such that the network consists of a single or a few components.  The limited number of components would not permit consistent estimation of the generative model.   Thus, there are models where consistent estimation is precluded.   On the other extreme where links are all independent, we know that consistent estimation holds. Thus, the question is for which models is it that consistent estimation can be obtained.  
With nontrivial interdependencies between links, standard asymptotic results do not apply.   This does not mean that consistency is precluded, (just as it is not precluded in time series or spatial settings) as there is still a lot of information that can be discerned from the observation of a single large network. Nonetheless, it does mean that asymptotic analyses must account for potentially complex interdependencies in link formation.

The third gap is providing microfoundations for estimable economic models of network formation.
While there are many theoretical models of strategic
 network formation (see \citet{jackson2008social} for references), there are only a handful of econometric  models that have been built from such foundations ( \citet*{currarini2009economic,currarini2010pnas,christakis2010empirical,goldsmith-pinkham2013,mele2011})
and those are rather specific to particular estimation exercises.

In this paper we make five contributions:
\begin{itemize}
\item First, we propose a generalization of the class of ERGMs that we call SERGMs: \emph{Statistical ERGMs}.   Note that in any ERGM the probability of forming a network is determined by its statistics: e.g., having a given link density, a given clustering coefficient, specific path lengths, etc.  Every network exhibiting the same statistics is equally likely.\footnote{This is related to the well-known property of sufficient statistics of the exponential family.  As an analogy, a binomial distribution defines the probability of seeing $x$ heads but does not care about the exact sequence under which the $x$ heads arrive.}  SERGMs nest the usual ERGM models by noting that: (i) we can define the model as one in which statistics are generated rather than graphs and thus greatly reduce the dimensionality of the space, and (ii) we can weight the distribution over the space of statistics in many ways other than simply by how many networks exhibit the same statistics.
Changing to the space of statistics and the reference distribution  allows us to provide computationally practical techniques for estimation of SERGMs.

\item
Second, we examine sufficient conditions as well as some necessary conditions for consistent estimation of SERGM parameters (nesting ERGMs as a special case) and identify a class of SERGMs for which it is both easy to check consistency and estimate parameters.  Models in this class are based on ``count'' statistics: for instance, how many links exist between nodes with certain characteristics, how many triangles
nodes exist, how many nodes have a given degree, etc.

\item
Third, we identify a related class of models that are based on the formation of subgraphs that we call SUGMs (\emph{Subgraph Generated Models}).\footnote{Although some particular examples of random networks have previously been built up from randomly generated subgraphs (\citet{bollobas2011sparse}), our general specification and analysis of SUGMs is new.}
Such a network is constructed from building up subgraphs of various types: links, triangles, larger cliques, stars, etc., layered upon each other, all of which can depend on characteristics of the nodes involved. We show if such models are sufficiently \emph{sparse}, parameter estimates are consistent and asymptotically normally distributed.  Such sparse networks appear in many if not most applications as they have realistic features (e.g., average degree that grows at a rate less than $n$, but still allow for high clustering, homophily, rich degree distributions, and so forth).

\item
Fourth, we provide a set of strategic network formation models that combine utility-based choices of subgraph formation by agents with randomness in meeting opportunities.  We describe two basic approaches: one based on consent in link and subgraph formation and another based on strategic search intensity choices.  We show how these provide foundations for classes of SERGMs and SUGMs, and illustrate them in our applications.

\item
Our fifth and final contribution is to provide illustrations of the techniques developed here by applying them to data on social networks from Indian villages. We show that many patterns of empirical networks are replicable by a parsimonious SUGM with very few parameters. We also answer the question that we began with above, of whether individuals tend to form cross-caste relationships more frequently when there are no friends in common than when there are.  We find that cross caste relationships occur with significantly higher frequency when in isolation than when embedded in triads.
\end{itemize}

The only work to date on consistency in ERGMs is by \citet{shalizi2012}.
They examine sequences of models (here, random networks indexed by the number of nodes $n$)
that satisfy a certain `projective' condition.  In that context, they show that this implies an independence of statistics
across increments of the model, and show that this is sufficient for consistency.
That might be thought of as a somewhat pessimistic result, given the required independence condition, which rules out
many of the most interesting ERGM models - essentially any models that involve more than link counts.
Our results are not implied by theirs and, in particular, their assumptions rule out counting any subgraphs that involve more than links, as those
are not projective, while we are directly interested in counting such subgraphs as these are basic to many models of social networks.\footnote{To be specific, note that
if one is counting subgraphs such as triangles, then generating those on the larger graph can lead to {\sl new} triangles
on a smaller graph.  For instance, suppose that triangles between nodes 1,2, and 3,  as well as 3, 4, and 5, are formed on the first five nodes.
Then if a triangle between nodes 2, 4, and 6 is formed when we go to the sixth node, this introduces a new triangle on the first five nodes, as now there are
links between 2, 3 and 4.  Thus, the model on the larger graph results in a different distribution of triangles on the first five nodes than what was originally there, and so
the marginal distribution working on six nodes is not the same as the distribution one started with on five nodes.
This is ruled out under the \citet{shalizi2012} projective assumption.  Similar points hold for richer cliques or other subgraphs that can generate incidental instances, which
are most of the cases of interest here.  Thus, our results cover large classes of models that are ruled out under their projective/independence conditions.}

The connection between ERGMs, SERGMs, and SUGMs is as follows.  SERGMs not only provide an alternative way of representing ERGMs by working directly with statistics rather than graphs, but also
substantially generalize the class by allowing for alternative reference distributions.   SUGMs then allow for an additional change relative to SERGMs in terms the way the graph is generated.  A SERGM -- in order to maintain the nesting of ERGMs -- has the likelihood of a network depend on the {\sl observed} counts of various statistics, including subgraphs.  A SUGM can be thought of as generating subgraphs, but allowing them to overlap: it is not clear whether a given triangle was generated directly as a triangle or as three separate links.  Thus, one needs to infer the true statistics in estimating the parameters of the model.  This subtle change allows for a more direct estimation in the case of sparse networks.  Nonetheless, there is a close relationship, and we provide an exact relationship between SUGMs and SERGMs below.   SERGMs, in addition to nesting ERGMs, provide the intuitive link between SUGMs and exponential-style representations.

\section{Preliminaries and Examples}\label{prelim}

Let $\mathcal{G}^n$ be a set of possible graphs on a finite number of nodes $n$.  The class can consist of undirected or directed graphs
with a generic element denoted by $g\in \{0,1\}^{n\times n}$. We often omit notation $\mathcal{G}^n$, and for instance, $\sum_g$  is understood to mean $\sum_{g\in \mathcal{G}^n}$.

We observe a single (large) graph from which to estimate a network formation model, which is a estimation problem faced by researchers. A family of models is indexed by a vector of parameters  $\beta$, and can be represented by corresponding probability distributions over graphs $\Prob_\beta \left(g \right)$, which depends on parameters $\beta$.

Some of our results concern asymptotic properties of such models, and so at times we consider a sequence of random graphs $g^n$, $n \in \mathbb{N}$, drawn from a sequence of probability distributions $\Prob^n_{\beta^n} (\cdot )$.  Since everything then carries an $n$ index we suppress it except when we want to highlight dependence.

A vector of {\sl statistics} of a network $g\in  \mathcal{G}^n$, $S(g) = (S_1(g), \ldots, S_k(g) )$, is a $k$-dimensional vector
where $S_\ell: \mathcal{G}^n \rightarrow \mathbb{R}$ for each $\ell \in \{1,\ldots k\}$.  For example, a statistic might be the number of links in a network, the average path length, the number of cliques of a given size, the number of isolated nodes, the number of links that go between two specific types of groups, and so forth.

\subsection{Links across social boundaries}\label{crosscaste}

\

To motivate the various models that we introduce and analyze below, let us reconsider the question that we mentioned in
the introduction.

Individuals are associated with groups and identities that can lead to strong social norms about interactions across groups.  For instance, in much of India there are strong forces that influence if and when individuals form relationships across castes.
Are people significantly more likely to form cross-caste relationships when those links are unsupported (without any friends in common) compared to when those links are supported with at least one friend in common?
To answer this we need models that account for link dependencies, as cliques of three or more may dictate greater adherence to a group norm prohibiting certain inter-caste relationships, while the norm may be circumvented in isolated bilateral relationships.

To analyze this, we examine data from 75 Indian villages (from our study  \citet{banerjee2014gossip}  that we discuss in more detail below). We link two households if members of either engaged in favor exchange with each other: that is, they borrowed or lent goods such as kerosene, rice or oil in times of need.
We work with two caste categories:  the first consists of people in scheduled castes and scheduled tribes and the second consists of those people in any other caste \citep{munshi2006traditional}.  Scheduled castes and scheduled tribes are those defined by the Indian government as being disadvantaged.  This is a fundamental  distinction over which the strongest cultural forces are likely to focus.  Additional norms are at work with finer caste (\emph{jati}) distinctions, but those norms are more varied depending on the particular castes in question while this provides for a clear barrier.

As a simple model to address this issue, consider a process in which
individuals may meet in pairs or triples and then decide whether to form a given link or triangle.
The link is formed if and only if both individuals prefer to form the link, and a triangle is formed if and only if all three individuals prefer to form it.
This minimally complicates an independent-link model enough to require modeling link interdependencies.

In particular, there are probabilities, denoted $\pi_L(diff), \pi_L(same)$, that a given link has an {\sl opportunity} to form (i.e., the pair meets and can choose to form the relationship) that depend on the pair of individuals being of different castes or of the same caste, respectively.
Similarly, there are probabilities, denoted $\pi_T(diff), \pi_T(same)$, that a given triangle has an {\sl opportunity} to form (that the three people involved meet and can choose to form the relationship) that depend on the triple of individuals being of all the same castes or two of the same and one of a different caste.

Preferences are similarly described in a random utility framework \citep{mcfadden1973conditional}. Individual $i$'s utility of having a relationship with $j$ can by influenced by whether they share caste and is given by
\[
u_i(ij) = \beta_{0,L} + \beta_{1,L} SameCaste_{ij} + \delta_L' X_{ij}- \epsilon_{L,ij},
\]
where $SameCaste_{ij}$ is a dummy for whether both individuals are members of the same caste,  $X_{ij}$ is a vector of covariates depending on $X_i$ and $X_j$. For expositional simplicity here set $\delta_L = 0$. The outside option is zero, so $p_L(same)$ is the probability that an individual will desire to form a link with an individual of the same caste group, and $p_L(diff)$ is the probability that an
individual will desire to form a link with an individual of a different caste group.

The crucial point is that $i$ can have returns that depend on being in a multilateral relationship with $j$ and $k$ -- that is conceptually distinct from having these two bilateral relationships -- and this can be given by
\[
u_i(ijk) = \beta_{0,T} + \beta_{1,T} SameCaste_{ijk} + \delta_T' X_{ijk}-  \epsilon_{T,i,jk},
\]
where $SameCaste_{ijk}$ is a dummy for whether all three individuals are members of the same caste, $X_{ijk}$ is a vector of covariates depending on $X_i$, $X_j$, and $X_k$. Again for expositional simplicity $\delta_T = 0$. Correspondingly, $p_T(same)$ is the probability that an individual will desire to form a triangle when all individuals are of the same caste group, and $p_T(diff)$ is the probability that an
individual will desire to form a triangle when it consists of people from both caste groups.\footnote{This is a simplified model for
illustration, but one can clearly consider preferences conditional on any string of covariates.  This extends a model such
as that of \cite*{currarini2009economic,currarini2010pnas} to allow for additional link dependencies. We could also be interested in higher order relationships.}

The hypothesis that we explore is that $p_T(diff)/p_T(same) < p_L(diff)/p_L(same)$ so that people are more reluctant to involve themselves in cross-caste relationships when those are ``public'' in the sense that other individuals observe those relationships; with a null hypothesis that they are equal $p_T(diff)/p_T(same) = p_L(diff)/p_L(same)$.

\subsection{ERGMs}

\

The standard (and to date essentially only) model for dealing with this sort of formulation in which we want to test hypotheses
about the formation of triangles and links is an ERGM.

In order to work with the data, which also
contains non-trivial numbers of isolated nodes (asocial individuals who do not form relationships with others),
we also allow for isolates.

Before incorporating the distinction between links and triangles of various types (e.g., same, different)
let us show that ERGMs are ill-equipped even to handle a non-type based model.
So, suppose that the probability of the formation of a network $g$ can be expressed as a function of the network's number of isolated nodes $S_I(g)$,  number of links $S_L(g)$, and  number of triangles $S_T(g)$.
In its exponential random graph model (ERGM) form, the probability of a network $g$ being formed is
\begin{equation}
\label{ergm-LT}
\Prob_\beta \left(g\right)=\frac{\exp\left(\beta_I S_I\left(g\right) +\beta_L S_L\left(g\right)+\beta_T S_T\left(g\right)\right)}{\sum_{g'}\exp\left(\beta_I S_I\left(g'\right) +\beta_L S_L\left(g'\right)+\beta_T S_T\left(g'\right)\right)}.
\end{equation}

If $\beta_I=\beta_T=0$ then this reduces to a standard Erd\H{o}s-R\'{e}nyi random graph. The more interesting case is where at least one of
 $\beta_I\neq 0$ or $\beta_T\neq 0$, so that networks become more ($\beta_T>0$) or less ($\beta_T<0$) likely based on the number of triangles they contain - or, similarly, of isolates they contain.

\subsubsection{ERGM Estimation}\label{ergm-examplebad}

The difficulty with estimating such a model is that the number of such networks in the calculation of the denominator's $\sum_{g'}$ is
$2^{\binom{n}{2}}$.\footnote{In the undirected case, even with a tiny society of just 30 nodes this is $2^{435}$, while estimates of the number of atoms in the universe are less than $2^{258}$ \citep{schutz2003gravity}.} Thus, the fraction of networks that can be sampled is necessarily negligible, and unless careful knowledge of the model is used in guiding the sampling, the estimation of the denominator can be inaccurate.

Given that estimating the parameters of an ERGM are thus forced to circumvent direct calculation of  the denominator,
approximation methods such as MCMC techniques have been used.\footnote{See \citet{snijders}, \citet{handcock2003assessing}, and discussions in  \citet{snijders2006new} and \citet{jackson2008social,jackson2010handbook}.}
The rough intuition is that such methods sample some networks (picking a few $g'$s ) to estimate the relative sizes of $\exp\left(\beta_I S_I\left(g'\right) +\beta_L S_L\left(g'\right)+\beta_T S_T\left(g'\right)\right)$
from which to extrapolate the $\sum_{g'}$ in the denominator of (\ref{ergm-LT}).
Even with this approach, the space of all possible networks is difficult to sample in a representative fashion.  For instance, if one samples say 10000 networks, then one samples on the order of $2^{16}$ networks out of the possible $2^{1225}$ on 50 nodes, which is about one out of every $2^{1209}$ networks.
Thus, unless one is very knowledgeable in choosing which networks to sample and how many to sample of different types, or one is very lucky, the sample is unlikely to be even remotely representative of the possible configurations that might occur.  Formally, draws generated by the sampling need to be well-mixed in a practical amount of time.

Indeed, the time before which an MCMC run has a chance to sample enough networks to gain a representative sample is generally {\sl exponential} in the number of links and so is prohibitively large even with a few nodes.\footnote{This does not even include difficulties of sampling.  For example, as discussed by  \citet{snijders2006new}, a technique of randomly changing links based on conditional probabilities of links existing for given parameters can get stuck at complete, empty, or other extreme networks.}
In particular, an important recent result of \citet{bhamidi2008mixing} shows that MCMC techniques using Glauber
dynamics for estimating many classes of ERGMs mix in less than exponential time {\sl only if} any finite group of edges are asymptotically independent. So, the only time those models are practically estimable is when the links are approximately independent, which precludes the whole reason for using ERGMs!

To illustrate the computational challenges, we estimate a version of the simple model from (\ref{ergm-LT}) on $n=50$ nodes.
In particular, we randomly generate networks that have {\sl exactly } 20 isolates, 45 links and 10 triangles on 50 nodes (with 15 links not in triangles).
Thus, the statistics of all of the networks are identical, and only the location of the links and triangles changes.   Any two networks with exactly the same statistics should lead to exactly
the same parameter estimates as they have exactly the same likelihood under all parameter values.
There is a unique, well-defined maximum likelihood estimated ERGM parameters for this set of statistics (as detailed in Section \ref{mle}).
Thus, the only variation in estimated parameter values comes from imperfections in the software and estimation procedure given the
computational challenges.

Using standard ERGM estimation software (\texttt{statnet} via \texttt{R}, \cite{handcock:statnet}) we estimate the parameters of an ERGM with isolates, links and triangles
for each of these randomly drawn networks that should all lead to exactly the same parameter estimates.  We present the estimates in Figure \ref{fig:ERGM-ILT2}.

\begin{figure}[h!]
\centering
\subfloat[Isolate Parameter Estimates]{
\label{fig:isolates2ERGM}
\includegraphics[width=0.33\textwidth]{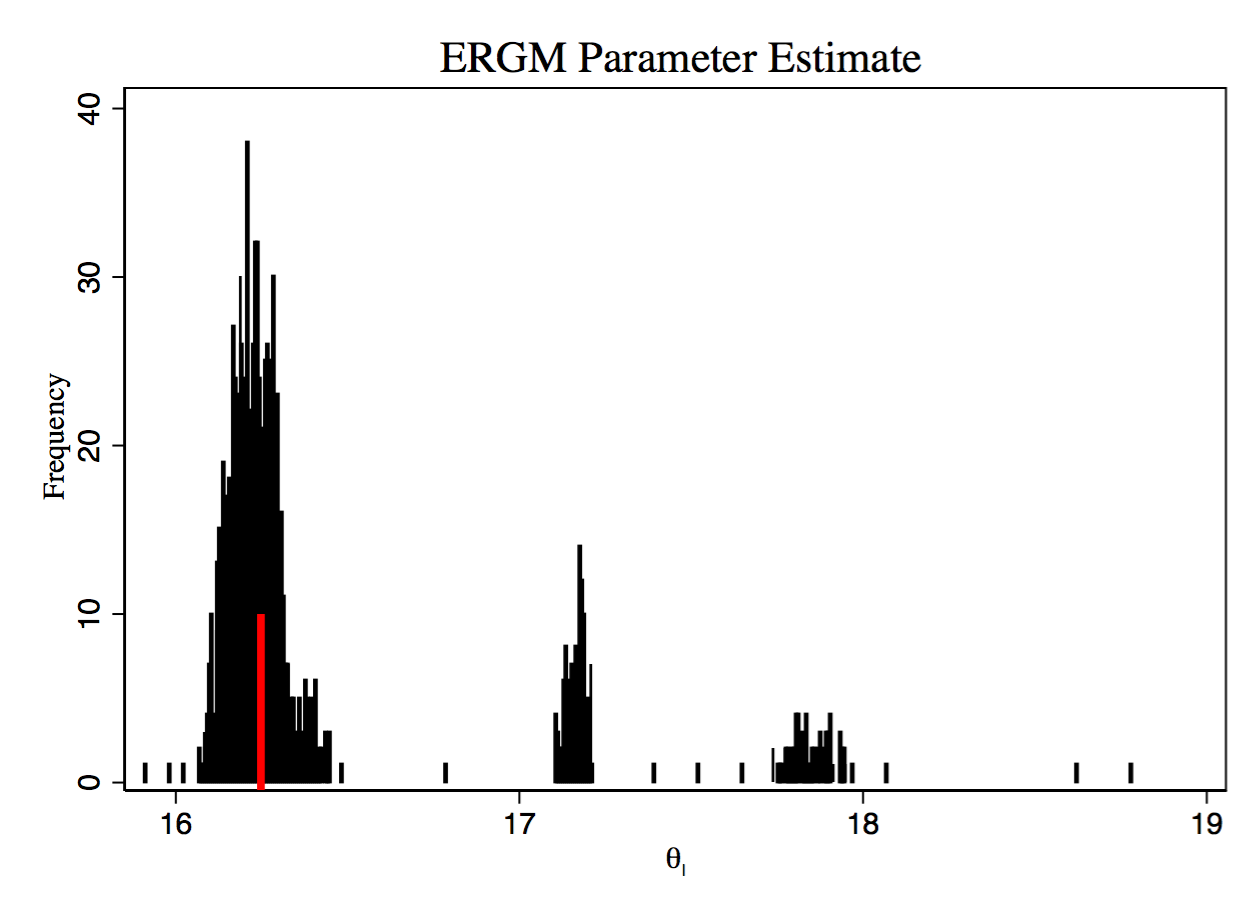}
}
\subfloat[Link Parameter Estimates]{
\label{fig:links2ERGM}
\includegraphics[width=0.33\textwidth]{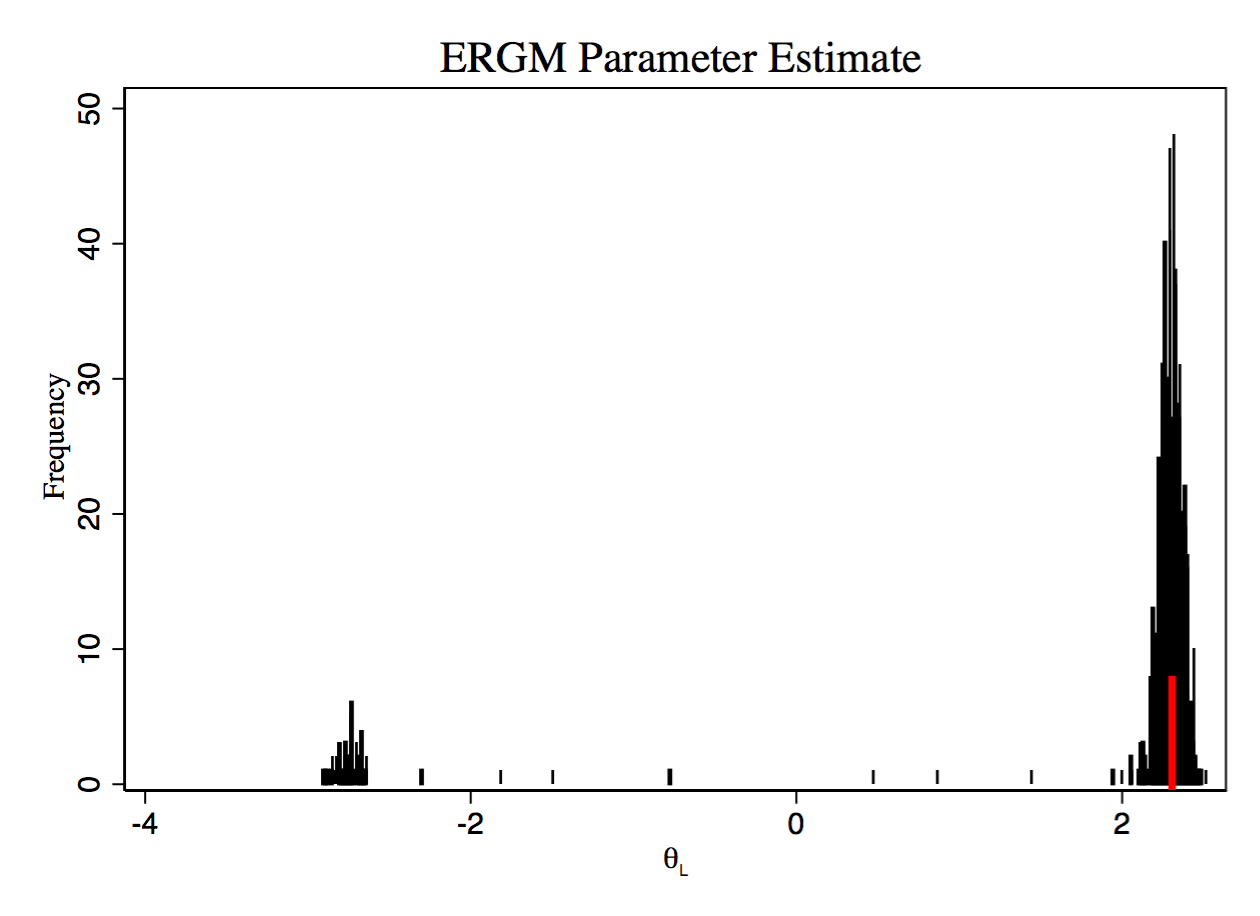}
}
\subfloat[Triangle Parameter Estimates]{
\label{fig:triangles2ERGM}
\includegraphics[width=0.33\textwidth]{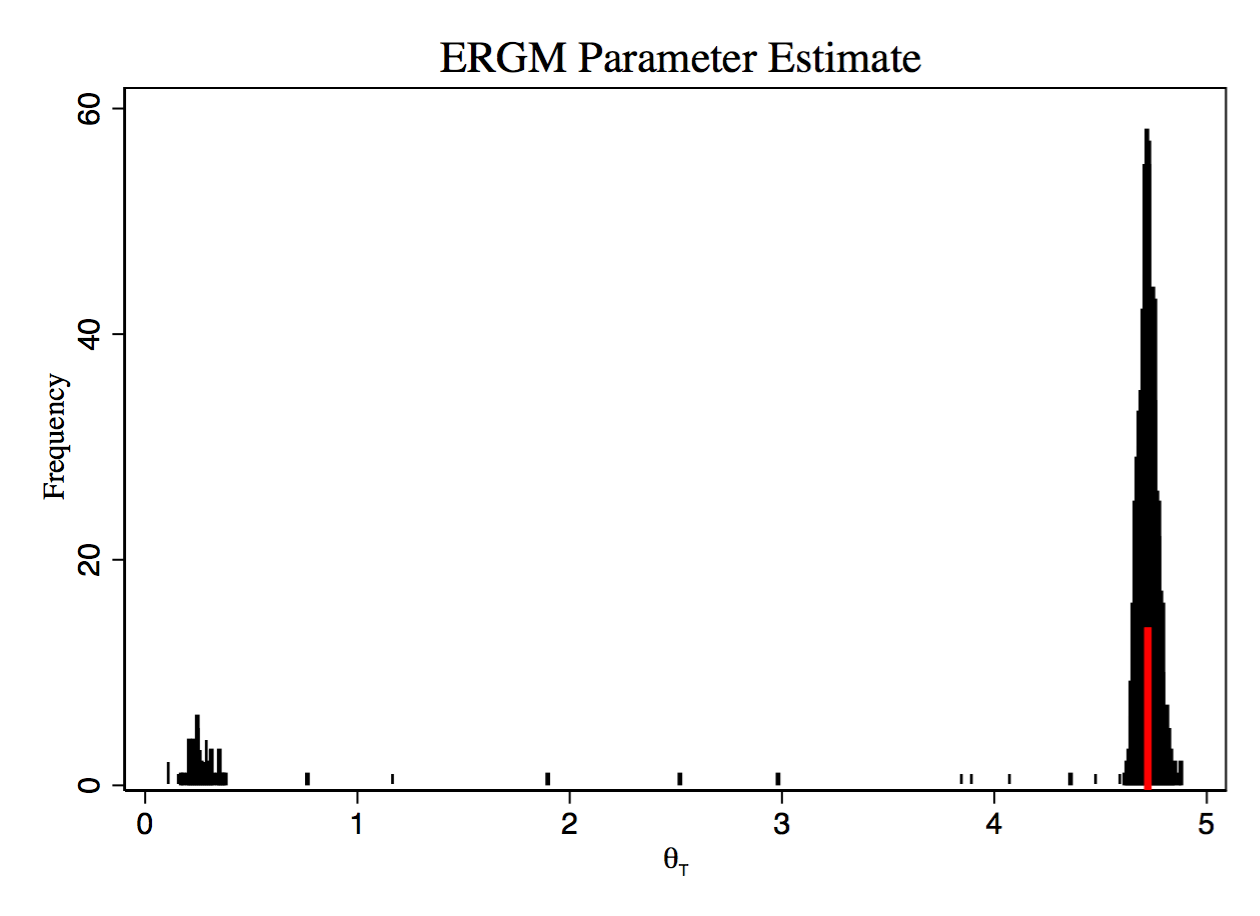}
}
\caption{\label{fig:ERGM-ILT2} Standard ERGM estimation software (Statnet) output for 1000 draws of networks on 50 nodes, each having \emph{exactly} 20 isolated nodes, 45 links, and 10 triangles. The red lines (on top of each other) are the median left and right 95 percent confidence interval lines (not capturing 95 percent of the estimates).
Networks with identical statistics should lead to identical parameter estimates: all of the variation comes from imprecisions in the estimation procedure.}
\end{figure}

There are two self-evident issues with the estimation.
First, the estimated parameters for links and triangles cover a wide range of values, in fact with the link parameter estimates being both positive and negative and
ranging from below -3 to above 3 (Figure \ref{fig:links2ERGM}) and triangles parameter estimates ranging from just above 0 to 5 (Figure \ref{fig:triangles2ERGM}).  Only the isolates
parameter estimates are remotely stable (Figure \ref{fig:isolates2ERGM}), but even those vary in three different regions with substantial variation.
Second, despite the enormous variation in estimated parameter values from very similar networks, the reported standard errors are quite narrow and almost always report that the parameter estimates are highly significant.  Moreover, the median left and right standard error bars essentially coincide and do not come close to capturing the actual variation.

For this example, there is no variation in the SERGM or SUGM estimates, as they are estimated via an exact 
 calculation, as we discuss in the next section.

In Appendix \ref{appendILT} we consider some additional tests - showing that the software has even more serious problems in simulating networks. Each of the 1000 simulated networks generates parameter estimates.  Using those parameter estimates we simulate a network using Statnet's simulation command.  We then check whether the simulated networks come anywhere close to matching the original networks.  The generated networks generally have hundreds of links and thousands of triangles (Figure \ref{fig:simulstats}), not at all matching
the original statistics.

\subsubsection{A Pr\'elude to Our Approach}

In order to overcome this problem, we develop two new classes of models, both of which
are partly built on the following insight.

Given the model specified in (\ref{ergm-LT}), any two networks that have the same numbers of isolates, links, and triangles have the same probability of forming.
That is, if $(S_I(g), S_L(g),S_T(g)) = (S_I(g'),S_L(g'),S_T(g'))$, then $\Pr_\beta(g)=\Pr_\beta(g')$ for any $\beta$.
This is simply an observation that $(S_I(g),S_L(g),S_T(g)) $ is a sufficient set of statistics for the probability of the network $g$.  

This observation can simplify the calculations dramatically.
Given a vector of statistics $S$  (e.g., $S= (S_I,S_L,S_T)$ in our example),
let
\[
N_S(s) := | \{ g\in \mathcal{G}^n :  S(g) = s\} |
\]
denote the number of graphs that have statistics $s$.
We can rewrite the denominator of the ERGM in (\ref{ergm-LT}) as
\[
\sum_{s'} N_{S_I,S_L,S_T}(s')\exp\left(\beta_I s_I'+\beta_L s_L'+\beta_T s_T'\right).
\]
Moreover, instead of considering the probability of observing a particular {\sl network}, we can instead ask what the probability is of observing a particular realization of network {\sl statistics}. For instance, what is the probability of observing a network with a given number of links and triangles?  Generally, this is what a researcher is interested in rather than which specific network that had a given list of characteristics was realized.
We can then express the model in the following form:
\begin{equation}
\label{sergm-LT}
\mathrm{P}_{\beta}\left((S_I,S_L,S_T)=s\right) =
\frac{N_{S}(s)\exp\left(\beta_I s_I+\beta_L s_L+\beta_T s_T\right)}{\sum_{s'} N_{S}(s')\exp\left(\beta_I s_I'+\beta_L s_L'+\beta_T s_T'\right)}.
\end{equation}
This is an example of what we call a Statistical Exponential Random Graph Model, or SERGM, which are defined in their more general form below.

We have thus reduced the complexity of the estimation problem from something that is exponential in the number of nodes, to something that depends on the size of the space of statistics, which is generally polynomial in the number of nodes.  For example, while the denominator of the ERGM in (\ref{ergm-LT}) was a summation over
a number of networks which is of order $2^{n^2}$, the summation now is over possible numbers of isolates, links, and triangles which is of order $n^6$.
As we further discuss in Appendix \ref{A-sum-sect},
there are further simplifications that reduce this even more dramatically.
The remaining challenge, therefore, lies in computing $N_S(\cdot)$, which still may be impossible to do for certain models. This motivates our SERGMs.

\subsection{Our approach}

\subsubsection{Statistical ERGMs}

\eqref{sergm-LT} defines a model over network statistics and, in principle, there is nothing special about the weighting function $N_S(\cdot)$, and at times it can be hard to compute or even approximate.  Noting that $N_S(\cdot)$ should have no privilage -- it neither has statistical advantages nor is it any more natural from the perspective of microfoundations -- we first think of a more general representation of SERGMs.
By replacing the weighting function $N_S$
with some other function $K_{S}: A\rightarrow \mathbb{R}$
we obtain a {\sl statistical exponential random graph model} (SERGM).  The associated probability of
seeing realized number of links and triangles $(S_I,S_L,S_T)=s$ is:
\[
\widehat{\mathrm{P}}_{\beta}\left((S_I,S_L,S_T)=s\right) =
\frac{K_{S}(s)\exp\left(\beta_I s_I+\beta_L s_L+\beta_T s_T\right)}{\sum_{s'\in A} K_{S}(s')\exp\left(\beta_I s_I'+\beta_L s_L'+\beta_T s_T'\right)}.
\]
This is a model that states that the probability that a network exhibits a specific realization of statistics $S=s$ is
given by an exponential function of the statistics $s$.   Note that SERGMs nests ERGMs as a special case.

\subsubsection{Subgraph Generation Models}

Our other model is not defined through an exponential form, but instead directly through the random formation of various subgraphs.  For instance, pairs of nodes, or triples of nodes, or some other configurations are directly randomly formed and a network results.
In the context of isolates, pairs and triangles, the process could be thought of as taking place as follows.
First, nodes decide to stay as isolates with some probability $p_I$.
Next, pairs of non-isolate nodes meet and decide whether to
form links with some probability $p_L$.
Also, triples of non-isolate nodes meet and decide whether to form triangles with some probability $p_T$.
The resulting network is the union of all links formed under the process.
All of these probabilities can be made dependent upon some list of node characteristics, as in Appendix \ref{ext}.
Thus, links and triangles are formed directly at random.
The model is then governed by the probabilities $p_I$ that any node is an isolate, $p_L$ that any given link is generated (on non-isolate nodes), and $p_T$ that any given triangle is generated (on non-isolate nodes).
We call this a {\sl Subgraph Generation Model} (SUGM).

The only challenge in estimating a SUGM is that we observe the resulting network and not the directly generated isolates, links and triangles.   For example, if the three links $12, 23, 13$ are all generated as links, then we would observe the triangle $123$ in the resulting network $g$ and not be sure whether it was generated as three links or as a triangle.
Nonetheless, by examining a large enough network we can accurately back out the probabilities in many cases.

We provide two sets of results on estimating SUGMs:  one concerning settings in which the networks are sparse enough so that
estimation can be made via direct counts, and a second concerning an algorithm for more general estimation when networks are dense enough so that there could be substantial overlap in the various subgraphs formed which makes counting more challenging.

\subsubsection{The Example Revisited}\label{ex:revisit}

Let us now return to the example presented in Section \ref{ergm-examplebad} that provided headaches for standard techniques for estimating ERGMs.
We can estimate that either as a SERGM or a SUGM.

The SUGM delivers direct estimates for the parameters:
$$ \widehat{p}_I := \frac{S_I}{n}  \mbox{ and } \widehat{p}_T := \frac{S_T}{\binom{n- S_I}{3}}  \mbox{ and } \widehat{p}_L := \frac{S_U}{\binom{n- S_I}{2}-3 S_T}.$$
These will be accurate estimates of the true parameters $p_I,p_T, p_L$ provided that the network is sparse enough, which is true in this example, as we show in Theorem \ref{sparse}.\footnote{Theorem \ref{sparse} does not explicitly include isolates, as we define subgraphs as connected objects for ease of notation.  However, the theorem extends easily to this case.  In particular, in the case of isolates, `sparse' actually puts a {\sl lower} bound on the probability of links - so that links are not so sparse as to generate extra isolated nodes.}
For cases of non-sparse networks, we provide an algorithm for estimating the parameters (after Theorem \ref{sparse}).

For all of the networks
\begin{equation}
\label{sugm-ex2}
\widehat{p}_I =\frac{{S}_I}{n}= \frac{20}{50}=.4, \  \widehat{p}_T=\frac{{S}_T}{\binom{n- {S}_I}{3}}= \frac{10}{\binom{30}{3}} = .002  \mbox{ and } \widehat{p}_L=\frac{15}{\binom{30}{2}-30}=.037.
\end{equation}

If we work with a SERGM (on unsupported links) that has weights
\begin{equation}
\label{sergm-ex2}
 K_I(s_I)=\binom{50}{s_I}  \mbox{ and } K_T(s_T)=\binom{\binom{30}{3}}{s_T} \mbox{ and } K_U(s_U)=\binom{\binom{30}{2}-30}{s_U},
\end{equation}
then as we show in Theorem \ref{prop-consistERGMs} and \ref{potential}, the SERGM parameters can be directly obtained as from the SUGM  binomial calculations, with an
adjustment for the exponential:
$$\widehat{\beta}_I = \log\frac{\widehat{p}_I}{1-\widehat{p}_I} = -.17,\  \widehat{\beta}_T = \log \frac{\widehat{p}_T}{1-\widehat{p}_T} = -2.7  \mbox{ and } \widehat{\beta}_U = \log \frac{ \widehat{p}_U }{1-\widehat{p}_U } = -1.4.$$
Thus we directly and easily obtain parameter estimates for the same networks that gave the ERGM estimation troubles.

\subsection{A Return to the Caste Example}

We can now use either of these approaches to test our hypotheses from the caste example.

Note that the probability that a ``same'' link forms is
$$P_L(same) = p_L(same)^2 \pi_L(same)$$
as it requires both agents to agree, and
the probability that a ``different'' link forms is
$$P_L(diff) = p_L(diff)^2 \pi_L(diff).$$
Analogously for triangles we have
$$P_T(same) = p_T(same)^3 \pi_T(same) \mbox{ and } P_T(diff) = p_T(diff)^3 \pi_T(diff),$$
where the cubic captures the
fact that it takes three agreements to form  the triangle.
The difference in the exponents reflects that it is more difficult to get a triangle to form than a link.  Hence, to perform a careful test, we have to adjust for the exponents as otherwise we would just uncover a natural bias due to the exponent that would end up favoring cross-caste links.

One challenge in identifying a preference bias is that it could be confounded by the meeting bias.  Thus, we first model the meeting process more explicitly and show that we still have identification as the meeting bias makes triangles relatively more likely to be cross-caste than links.  Thus, our test is conservative in the sense that if we find cross-caste links relatively more likely, that is evidence for a (strong) preference bias.

Consider a meeting process where people spend a fraction $f$ of their time mixing in the community that is predominantly of their own types and a fraction $1-f$ of their time mixing in the other caste's community. Then at any given snapshot in time, a community would have $f$ of its own types present and $1-f$ of the other type present, as depicted in Figure \ref{fig-geo}.  (Variations on this sort of biased meeting process appear in \citet{currarini2009economic,currarini2010pnas,bramoulle-etal}.)

\begin{figure}[!h]
\centering
\subfloat[Individuals all on own-community side of river]{
	\includegraphics[scale = 0.25]{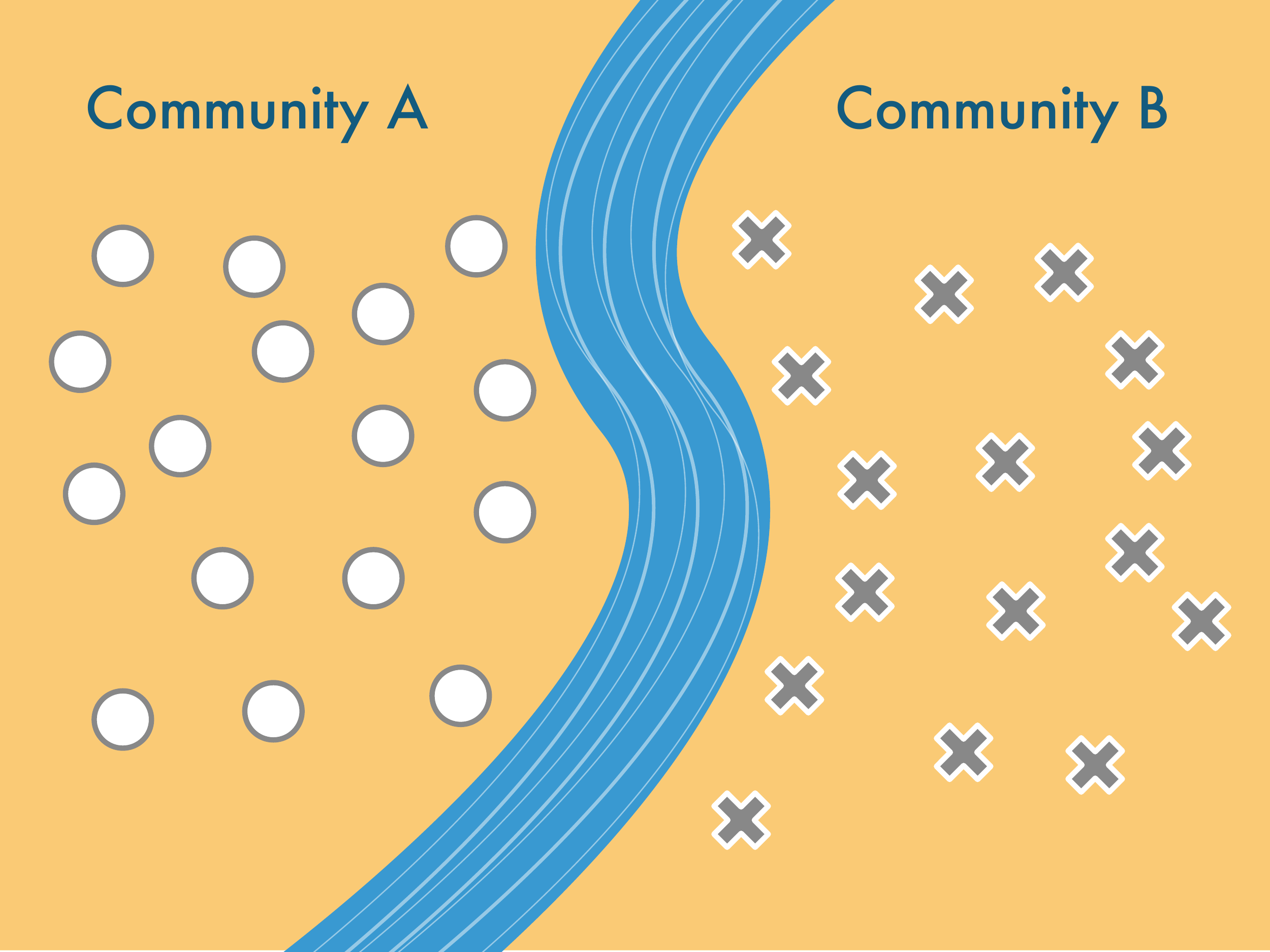}
}\ \ \ \ \
\subfloat[Fraction $f = \frac{1}{4}$ mixed across communities]{
	\includegraphics[scale = 0.25]{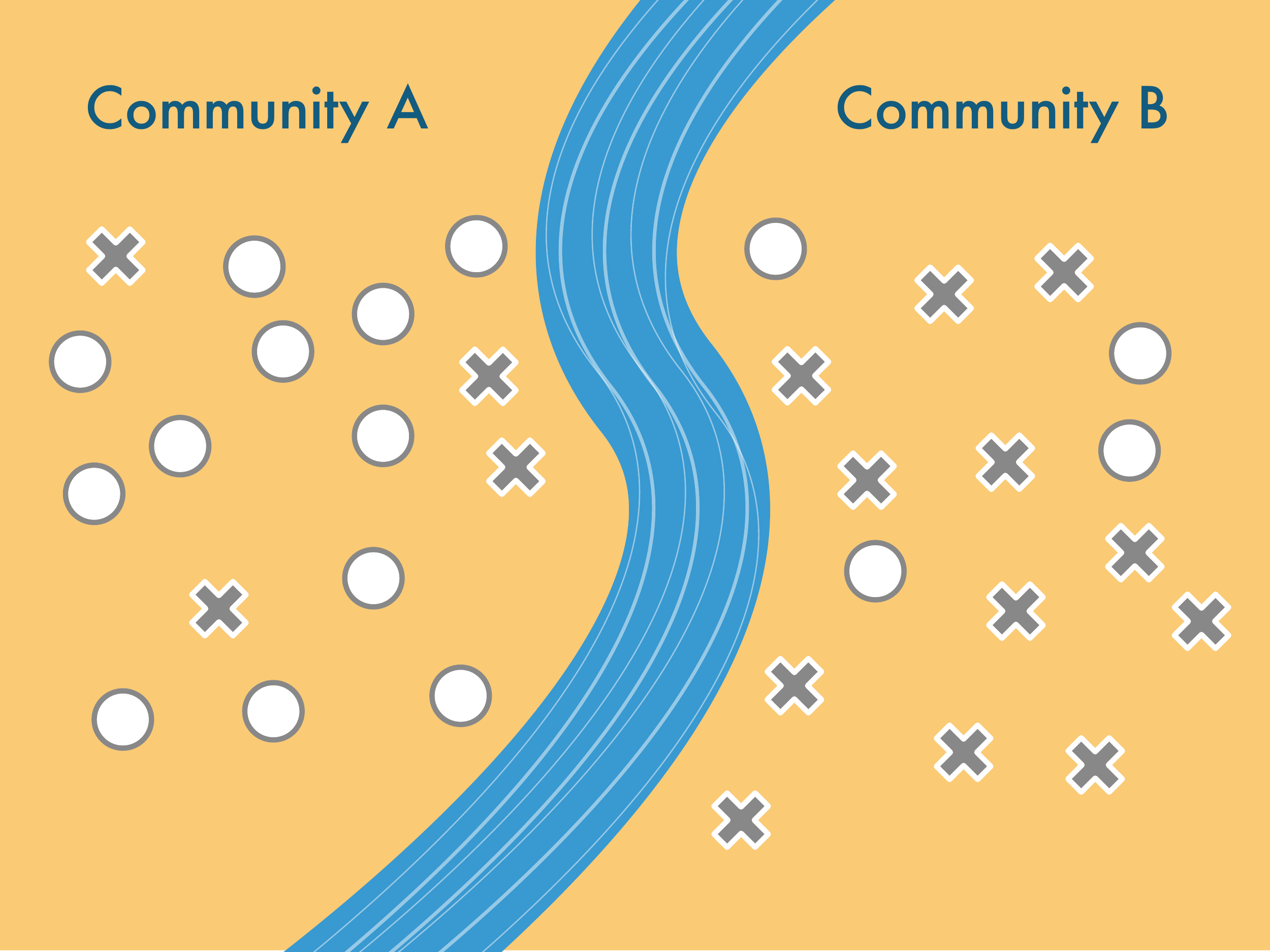}
}
	\caption{Geographically driven meeting process where agents spend 3/4 of their time in their own community.}
\label{fig-geo}
\end{figure}

\begin{lemma} \label{lem:geomodel} A sufficient condition for
$\frac{p_T(diff)}{p_T(same)} < \frac{p_L(diff)}{p_L(same)}$ is that
$\frac{P_T(diff)}{P_T(same)}<\left(\frac{P_L(diff)}{P_L(same)}\right)^{3/2}.$
\end{lemma}

The proof appears in the appendix, but follows from straightforward calculations.

Given Lemma \ref{lem:geomodel}, we can test our hypothesis directly from a SUGM that compares relative link and triangle counts (we can
also include isolated nodes, but those do not impact this hypothesis).
In particular, we only need examine whether $\frac{P_T(diff)}{P_T(same)}<\left(\frac{P_L(diff)}{P_L(same)}\right)^{3/2}.$

Figure \ref{fig-caste} shows the results.  For the bulk of villages, cross-caste relationships relative to within-caste relationships are more frequent as isolated links as opposed to being embedded in triangles, {\sl even when adjusting for the fact that triangles take more consent}.  The difference is significant at the 99 percent level.\footnote{This is from doing a conservative
nonparametric test: under the null that the number of villages for which the ratio is less should be 1/2 with a binomial distribution on the number above or below.}

\begin{figure}[!h]
\centering
	\includegraphics[trim = 0.7in 2in 0.7in 2in, clip = true, scale = 0.55]{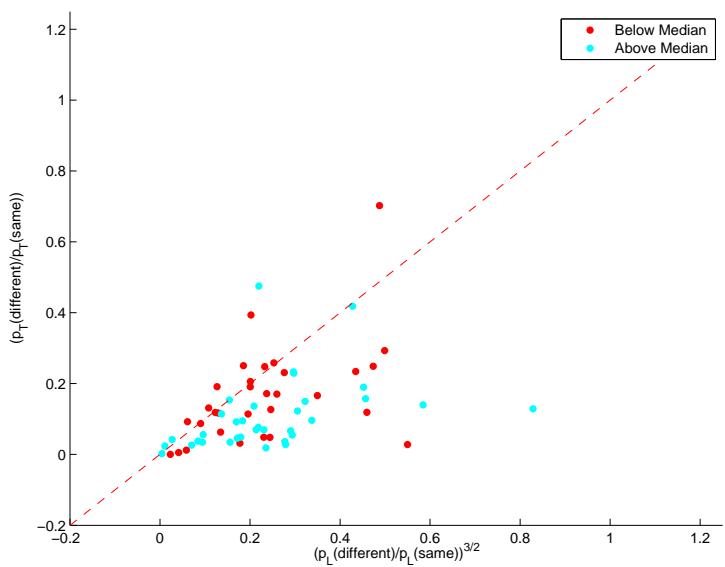}
	\caption{Comparison of the relative propensity to form cross-caste versus same-caste relationships for triangles ($y$-axis) compared to links ($x$-axis).  The propensity is lower for triangles than links in a significant number of villages, even when adjusting link propensities downwards by raising them to the 3/2 power to adjust for the number of consents needed to form the subgraphs.   The color coding 	 distinguishes those villages that have above/below the median size minority group.}
\label{fig-caste}
\end{figure}

In  Figure \ref{fig-caste} villages are color coded by the relative sizes of the two caste-based groups.   The red villages are such that one of the two caste designations dominates the village and the other group is relatively small, while the blue villages are ones in which the two caste designations are more balanced in terms of sizes.  In other contexts, homophily has been found to be strongest when groups are evenly balanced (e.g., see \cite{mcpherson2001,currarini2009economic,currarini2010pnas}).
Here we see that the social pressures against mixed-caste triangles are stronger when the two caste designations are more evenly balanced.

\bigskip

To sum up, we develop two classes of tractable models.  One are subgraph generation models (SUGMs) in which we think of subgraphs as being directly generated by subgroups of nodes.  The second is a more general statistical exponential random graph model (SERGM), in which a network is drawn based on its properties (e.g., a vector of sufficient statistics such as subgraph counts).  We now provide
formal definitions, and then theorems on asymptotic estimation of each of these classes of models, and also describe techniques that provide for tractable estimation even with large numbers of nodes in many cases.
We then further clarify the relationship between SUGMs and SERGMs, via Theorem \ref{potential}.

\section{Definitions}\label{framework}

We first present some needed definitions before describing our results.

\subsection{SERGMs}

\

The general set of SERGMs that we define is as follows. Consider a vector of network statistics $S=(S_1, \ldots, S_k)$  that takes on values in some set $A\subset \mathbb{R}^k$.\footnote{Given the finite number of possible networks, $A$ is taken to be finite.  The dimension of $A$ can easily be generalized to be larger than $k$, as the dimension plays no role in our results. If one wishes to work with weighted networks,
then obvious extensions to continuous ranges and integrals apply.}
A weighting function $K_S: A\rightarrow \mathbb{R}$, together with a set of parameters $\beta\in
\mathcal{B} \subset \mathbb{R}^k$,
define a SERGM. The associated probability of
seeing realized statistics $S=s$ is:
\begin{equation}
\mathrm{P}_{\beta, K_S}\left(s\right)=\frac{K_S(s)\exp\left(\beta \cdot s \right)}{\sum_{s'\in A} K_S(s')\exp\left(\beta \cdot s'\right)}.
\label{eq: sergm}
\end{equation}
The model is based directly on the properties of the network rather than the actual realized network.\footnote{
Which network forms given the realized statistics is secondary and could be uniform at random, or according to some other conditional distribution,
so long as given the realized $s$ a network $g$ such that $S(g)=s$ is drawn. Unless otherwise stated we take it to be uniform at random.} Recall all parameters can depend on $n$.

In the language of exponential families of random variables, $K_S(\cdot)$ is simply a \emph{reference distribution}.  Varying the reference distribution, of course, changes the resulting odds of various values of $s$ being drawn and can affect whether the model is consistently estimable.\footnote{Note that any SERGM with weights
$K_S$ and parameters $\beta$ also generates a distribution over networks, for example taking networks to be drawn uniformly at random from those with the given statistics. This can be written
as
\[
\Prob_{\beta,K}(g)=\frac{\frac{K_{S}(s)}{N_{S}(s)}\exp\left(\beta\cdot s(g)\right)}{\sum_{g'}\frac{K_{S}(s(g'))}{N_{S(s(g'))}}\exp\left(\beta\cdot s(g')\right)}=\frac{\exp\left(\beta\cdot s(g)+f_n(s)\right)}{\sum_{g'}\exp\left(\beta \cdot s(g')+f_n(s(g'))\right)},
\]
where $f_n(s) = \log\left( \frac{K_{S}(s)}{N_{S}(s)}\right)$.
Therefore the model is modified by a shift with weights $\log\left(K_{S}(s)/N_{S}(s)\right)$.
The sum is over $g'$ for which $N_S(s(g'))>0$.}

Recalling that $N_S(s)= \left\vert \right\{ g\in \mathcal{G}^n :  S(g) = s\left\} \right\vert$ is the number of graphs that have the same statistic value $s$, $K_S(\cdot)=N_S(\cdot)$ then corresponds to a standard ERGM.  Thus, SERGMs nest ERGMs as a special case.
Note, however, that there is no reason to maintain that $K_S(\cdot)$'s must approximate $N_S(\cdot)$'s.  Nature may choose properties of networks ($S$'s)  according to some alternative weighting.  The instance of studying $N_S(\cdot)$-weighted SERGMs may be a historical one: on another planet, people may have first modeled SERGMs with general $K$'s and would see those as natural with the ERGMs being a special case where the weights are specialized to the $N_S$'s.
As we shall see below, there are natural economic based network formation models for which the reference distributions will not be the $N_S(\cdot)$'s.
Moreover, even if one is interested in a sub-class of these models wherein the $K_S(\cdot)$'s approximate (or are) the $N_S(\cdot)$'s, the statistical representation greatly reduces the dimensionality of the space over which relative likelihoods must be estimated to the point at which practical estimation of SERGMs becomes feasible.

It is important to note that node characteristics can also be included in statistics.  For example, in terms of the question we raised in the introduction, we can keep track of nodes' castes.  Then we can keep separate counts of how many links there are between people both of caste A, between nodes both of caste B, and how many there are between castes A and B; as well as how many triangles involve only people of caste A, how many triangles involve only people of caste B, and how many triangles involve people of different castes, and so forth.

\subsubsection{Estimation of SERGMs}
\label{mle}

The maximum likelihood estimator
 solves
\[
\betahat =\argmax_\beta  \beta \cdot s-\log\left[\sum_{s'\in A} K_S(s')\exp\left(\beta \cdot s'\right)\right].
\]
Under regularity conditions such that the SERGM is sufficiently identified ($\beta\neq \beta'$ implies that $\E_{\beta} [S]\neq \E_{\beta'} [S]$), the MLE $\betahat$ of a SERGM of the form \eqref{eq: sergm} solves
\begin{equation}
s= \frac{\sum_{s'\in A}K_S(s') \exp\left(\betahat\cdot s'\right)s' }{\sum_{s'\in A}K_S(s') \exp\left(\betahat\cdot s'\right)} = \E_{\betahat} [S].
\label{eq:sergmest}
\end{equation}

For extreme values of $s$ this will not be well-defined.\footnote{For example, for a simple Erd\H{o}s-Renyi random network where the count statistic is simply the number
of links in the network, if it turns out that all links are present so that $s=n(n-1)/2$, then the 
$\beta= \log \left( \frac{p}{1-p}\right)$ corresponding to the maximum likelihood estimator of the link probability ($p=1$) is not well-defined. 
For more on the non-existence of well-defined maximum likelihood estimates for extreme networks see \citet{rinaldo2011maximum}.}
Here, we implicitly assume that the model is specified so that the probability of observing extreme statistics for which this is not satisfied is negligible, which will be true of the asymptotic specifications that we work with
provided that the $\beta$'s do not tend to extremes too quickly.\footnote{Parameters can still approach extremes.  The requirement here can be fairly weak.  For example, if one were counting links it must be that the probability of having absolutely no links (or all links) realized vanishes, which is true even if the probability of a link is larger than $1/n^x$ for some $x<2$.
}

\

\subsection{Subgraph Generation Models: SUGMs}\label{newmodel}

\

The idea behind a  SUGM is that subgraphs are directly generated by some process.  Classic examples of this are Erdos-Renyi random networks in which each link is randomly generated, and the generalization of that model,  stochastic-block models, in which links are formed with probabilities based on the nodes' attributes.
The more interesting generalization of those linked-based models to SUGMs is to allow richer subgraphs to form directly, and hence to allow for dependencies in link formation.  It might be that people of the same caste meet more frequently or are more likely to form a relationship when they do meet.  Similarly, groups of three (or more) randomly meet and can decide whether to form a triangle, with the meeting probability and decision potentially driven by their castes and/or other characteristics.   The model can then be described by a list of probabilities, one for each type of subgraph, where subgraphs can be based on the subgraph shape as well as the nodes' characteristics.

As we show in Theorem \ref{potential}, SUGMs have a representation in a SERGM form, but in some relevant cases SUGMs are easier and more intuitive to work with directly, and so we distinguish them from their SERGM representation.

SUGMs are formally defined as follows.
There is a a finite number of different types of nonempty subgraphs, indexed by $\ell\in \{1,\ldots, k\}$, on which the model is based.\footnote{This definition does not admit isolates since we define subgraphs to be nonempty, but isolates are easily be admitted with notational complications, and are already illustrated in the examples.}
In particular, a SUGM on $n$ nodes is based on some list of $k$ subgraph types: $(G^n_\ell)_{\ell\in \{1,\ldots,k\}}$
where each $G^n_\ell$ is a set of possible subgraphs on $m_\ell$ nodes, which are identical to each other (including node covariates) up to the relabeling of nodes.\footnote{Formally, there is a set $\mathcal{H}=\{H_1,...,H_k\}$ of \emph{representative subgraphs}, possibly depending on  covariates, each having $m_\ell$ nodes. 
$G_\ell^n$ contains all subgraphs that are homomorphic to $H_\ell$.
As an example, the set $G^n_\ell$ for some $\ell$ could be all triangles such that two nodes have characteristics $X$ and one has $X'$. These could also be directed subgraphs in the case of a directed network. With assumptions on smooth covariates and probability functions, one could have $p^n_\ell(x_\ell)$, described in Appendix \ref{ext}.}
The final ingredient is a list of corresponding parameters $p^n=(p_1^n, \ldots, p_k^n)\in [0,1]^k$  governing the likelihood that a particular subgraph appears, with $p^n_\ell$ indicating the probability that a subgraph in $G^n_\ell$ forms.

A network $g$ is randomly formed as follows.  First, each of the possible subnetworks in $G_1^n$ is independently formed with a probability $p^n_1$.    Iteratively in ${\ell\in \{1,\ldots, k\}}$ , each of the possible subnetworks in $G^n_\ell$ that is not a subset of some subgraph that has already formed is independently formed with a probability $p^n_\ell$.  The resulting $g$ is the union of all
the links that appear in any of the generated subgraphs.

We consider two variations of the model.  The first, as just defined, is one in which we only keep track of subnetworks in $G^n_2$ that are not already part of a subnetwork in $G^n_1$ that already formed.  The other variation is one in which we allow for redundant formation, and simply form subgraphs of each type disregarding the formation of any other subgraphs.

To see the issue, consider the formation of triangles and links.  Let $G_1^n$ be a list of all possible triangles and $G_2^n$ be a list of all possible links.  First form the triangles with the corresponding probability $p^n_1$.   This then leads to the creation of some of the links in $G_2^n$.  Do we allow those links to also form on their own?  Whether we then allow links that are already formed as part of a triangle to form again as links is inconsequential in terms of the network that emerges, and really is an accounting choice and leads to an equivalent distribution over networks.  Thus the two conventions for generating networks are equivalent.   It turns out sometimes to be easier to count subgraphs as if they can form in multiple ways, and at other times it is easier to keep track of smaller subnetworks that form only on their own and not already as part of some larger subnetwork.  We are explicit in which way we use in what follows.

When a subgraph $g'$ is generated in the $\ell^{th}$-phase, we say that it is \emph{truly generated}. This results in a network $g$, which is the union of all the truly generated subgraphs.
The resulting $g$ can also contain
some {\sl incidentally generated} subgraphs that result from combinations of links of unions of truly generated subgraphs, and we provide further definitions concerning this below.

This model differs from a SERGM because
the {truly generated} subnetworks are not directly observed. The actual counts of statistics under the resulting $g$ can differ from the number that were formed directly under the process.  Backing out how many of each type of subnetwork was truly generated is important in estimating the
true parameters of the model, the $p^n_\ell$'s, and is something that we discuss at length below.

\section{SERGM Estimation
}\label{SERGM}

Under what conditions does an estimator of a SERGM converge to the correct estimate in probability as $n$ grows? 
The primary challenge is that the data consists of a single network, the asymptotics are in terms of the number of nodes, but the relationships are correlated and so the data can be far from independent.
We consider sequences of SERGMs $(S^n, K^n_S, A^n, \beta_n)$, with $n \rightarrow \infty$.

\subsection{Count SERGMs}\label{simpleCount}

\

We begin by focusing on a natural subclass of SERGMs that we call ``count SERGMs''.  We show that these have parameters that are consistently and easily estimable with direct counts of subgraphs.

Let $S^n=(S^n_1, \ldots, S^n_k)$ be a $k$-dimensional vector of network statistics whose $\ell$-th entry takes on non-negative integer values with a maximum value $\overline{S}^n_\ell\rightarrow \infty$.  We call such a SERGM specified with
$K^n(s)= \prod_\ell \binom{\overline{S}^n_\ell}{s_\ell}$  a {\sl count SERGM}.  Let let $D_{n}={\rm Diag}\left\{ \overline{S}^n_\ell  \right\} _{\ell=1}^{k}$ be the associated normalizing matrix.

In a count SERGM, each statistic can be thought of as counting some aspect of the network:  the number of links between nodes of various types, various types of cliques, other subgraphs, the number of pairs of nodes at less than some distance from each other, etc.
It includes counts of subgraphs, but also allows for other counts as well (e.g., the number of pairs of nodes at certain distances from each other, as just mentioned; or the number of nodes that have more than a certain degree - so a degree distribution).

Associated with any vector of count statistics $S^n$ on $n$ nodes is a possible range of values. It could be that there are cross restrictions on these values.  For example, if we count links $S_L^n$ and isolates $S_I^n$,  then $S_L^n$ cannot exceed $\binom{n-S_I^n}{2}$.
In that case the set of possible statistics is a set $A^n$ where
$$A^n=\left\{ (s_L,s_I):   s_I\in \{0,1,\ldots,n\}, s_L\in \left\{0,1,\ldots,\binom{n-s_I}{2} \right\} \right\}.$$

Given that $A^n$ might not be a product space, in estimating count SERGMs, it will be helpful to know whether the realized statistics are likely to be close to having binding restrictions on the cross counts.  For example, if a model
is expected to generate one tenth of its nodes as isolates and, say, one in a hundred of all possible links, then in a wide band around the expected values there would be no conflict in the counts.

Formally, a sequence of count SERGMs $(S^n, K^n, A^n, \beta^n)$ have statistics that are {\sl non-conflicted} if there exists some $\varepsilon>0$ such that
\[
\prod_\ell
\left\{\lfloor \E_{\beta^n_\ell}[S^n_\ell](1 - \varepsilon) \rfloor,\lfloor \E_{\beta^n_\ell}[S^n_\ell](1 - \varepsilon) \rfloor+1,\ldots,  \lceil \E_{\beta^n}[S^n_\ell](1 + \varepsilon) \rceil \right\}\subset A^n
\]
for all large enough $n$.\footnote{$\E_{\beta^n_\ell}[S^n_\ell]$ refers to the expectation taken with respect to the one dimensional distribution of $S_\ell^n$ ignoring other statistics:
i.e., with respect to a SERGM $\frac{K^n_{S_\ell}(s_\ell)\exp\left(\beta \cdot s_\ell \right)}{\sum_{s_\ell'\leq \overline{S}^n_\ell} K^n_{S\ell}(s_\ell')\exp\left(\beta \cdot s_\ell'\right)} $.
This takes expectations with respect to the unconstrained range of $S^n_\ell$ rather than cross restrictions imposed under $A^n$.}

The ``non-conflicted'' condition simply asks that at least in some small neighborhood of the unconstrained expected values of the statistics - as if they were each counted completely on their own, they are non-conflicted so that they are jointly feasible.  Essentially, a local neighborhood of the expected statistics contains a product space.
This condition is quite easy to satisfy.\footnote{There are other things embodied here, as  certain counts of statistics might not be feasible: e.g., it is not possible to have a network with only one triangle missing. Once one triangle is removed it also removes many others.  Thus, the range of some statistics is not a connected (containing all adjacent entries) subset of the integers. Still, for lower values of triangles, this is not an issue.   In the relatively sparse ranges of networks that are often of empirical interest, this condition is easily satisfied.}

For models where counts are non-conflicted, with high probability the realized statistics lie in a product subspace, which helps us
prove the following consistency result.

\begin{theorem}[Consistency and Asymptotic Normality of Count SERGMs]
\label{prop-consistERGMs}
A sequence of count SERGMs that are non-conflicted is consistent; $ |  \betahat^n - \beta^n | \cvgto 0$.

Moreover, if $\exp \beta^n_\ell / (1+ \exp \beta^n_\ell ) \cdot \bar{S}_\ell^n \rightarrow \infty$ for every $\ell$, the parameter estimates are asymptotically normally distributed: $$D_{n,\ell}^{1/2}\left(\betahat_\ell^n - \beta_\ell^n\right)\rightsquigarrow \mathcal{N}\left(0, \frac{1}{\frac{\exp\beta_{\ell}^{n}}{1+\exp\beta_{\ell}^{n}} \cdot \left( 1- \frac{\exp \beta_{\ell}^n}{1+\exp \beta_{\ell}^n} \right)}\right).$$

Finally, letting $\widehat{p}_\ell=s_\ell / \overline{S}^n_\ell$, an approximation of the MLE estimator can be found directly as
$\betahat_\ell:=\log\left(\phat_\ell /(1- \phat_\ell)\right) = \log\left(s_\ell / (\overline{S}^n_\ell- s_\ell )\right). $

\end{theorem}

The proof of Theorem \ref{prop-consistERGMs} works via showing that the model can be locally approximated by a product of appropriately defined binomial random variables.  In fact those binomial random variables provide a direct estimator for count SERGMs.
Our proof shows that following what would seem to be a naive
technique is valid: one can simply estimate parameters $p_\ell$ as if the subgraphs were generated according to a binomial distribution with a maximum number of possible realizations $\overline{S}^n_\ell$ and $s_\ell$ as its realization.

It is important to emphasize that count SERGMs still allow for strong interdependencies and correlations in link appearances, both within and across statistics.  What our proof takes advantage of is a local approximation of such count SERGM distributions in non-conflicted regions. Theorem \ref{prop-consistERGMs} tells us that non-conflicted count SERGMs form a consistently estimable class whose statistical properties we understand very well.

\subsection{Consistency of SERGM Estimation Beyond Count Statistics}

\

The above results apply to a fairly general class of SERGMs, count SERGMs, for which we can derive explicit asymptotic distributions and simple estimators.  We also provide results about consistency for the more full class of SERGMs in Appendix \ref{addconsistency}.

Briefly, there are two sorts of conditions that we outline as being sufficient for consistency (and as we show, effectively, necessary).  One is an identification condition that
requires that different parameters distinguish themselves with different expected statistics.  It is a minimal
condition (essentially necessary) since if two different parameter values generate very similar expected statistics, then observing the realized statistic will not allow us to distinguish the parameters.
The second condition
requires that the (appropriately normalized) statistics concentrate around their means.
If the statistics are not concentrated, then even though different parameters lead to different expected statistics, observing a statistic would not allow one to back out the parameters.  Various combinations of such conditions (see Appendix \ref{addconsistency}) ensure consistent estimation.

\section{SUGM Estimation
}\label{sparsity}

Next, we discuss the estimation of SUGMs.  The main challenge here is that subnetworks can be incidentally generated:  forming links can lead some triangles to form indirectly.  Thus, to estimate the actual true generation rates, we need to estimate incidental formation.  We take two approaches.  One takes advantage of the fact that many social and economic applications are in the context of sparse networks.
 We show in large and sparse enough networks, incidental generation does not significantly bias estimation, and direct counts provide
 asymptotically accurate estimates of generating probabilities.  The second is to provide an explicit algorithm for estimation networks where incidentals may be nontrivial, which we return to in Section \ref{fsc}.

\subsection{Incidentally Generated Subgraphs}

\

To see the issue of incidental subgraph generation in SUGMs consider the following example.
Suppose that the subgraphs in question are triangles and single links, so that $G^n_1(g)$ is the set of all triangles possible among the $n$ nodes, and
$G^n_2(g)$ is the set of links on $n$ nodes.   The triangle $\{12, 23, 31\}$ could be incidentally generated by the subgraphs $g^1, g^2, g^3$ where
$g^1 =\{12, 24, 41\}$, $g^2=\{23, 25, 53\}$ and $g^3=\{31\}$.
Figure \ref{fig: overlap_multipanel} provides an illustration.

\begin{figure}[h!]
\centering
\subfloat[]{
\label{fig:overlap_1}
\includegraphics[trim = 15mm 15mm 15mm 15mm, clip = true, scale = 0.13]{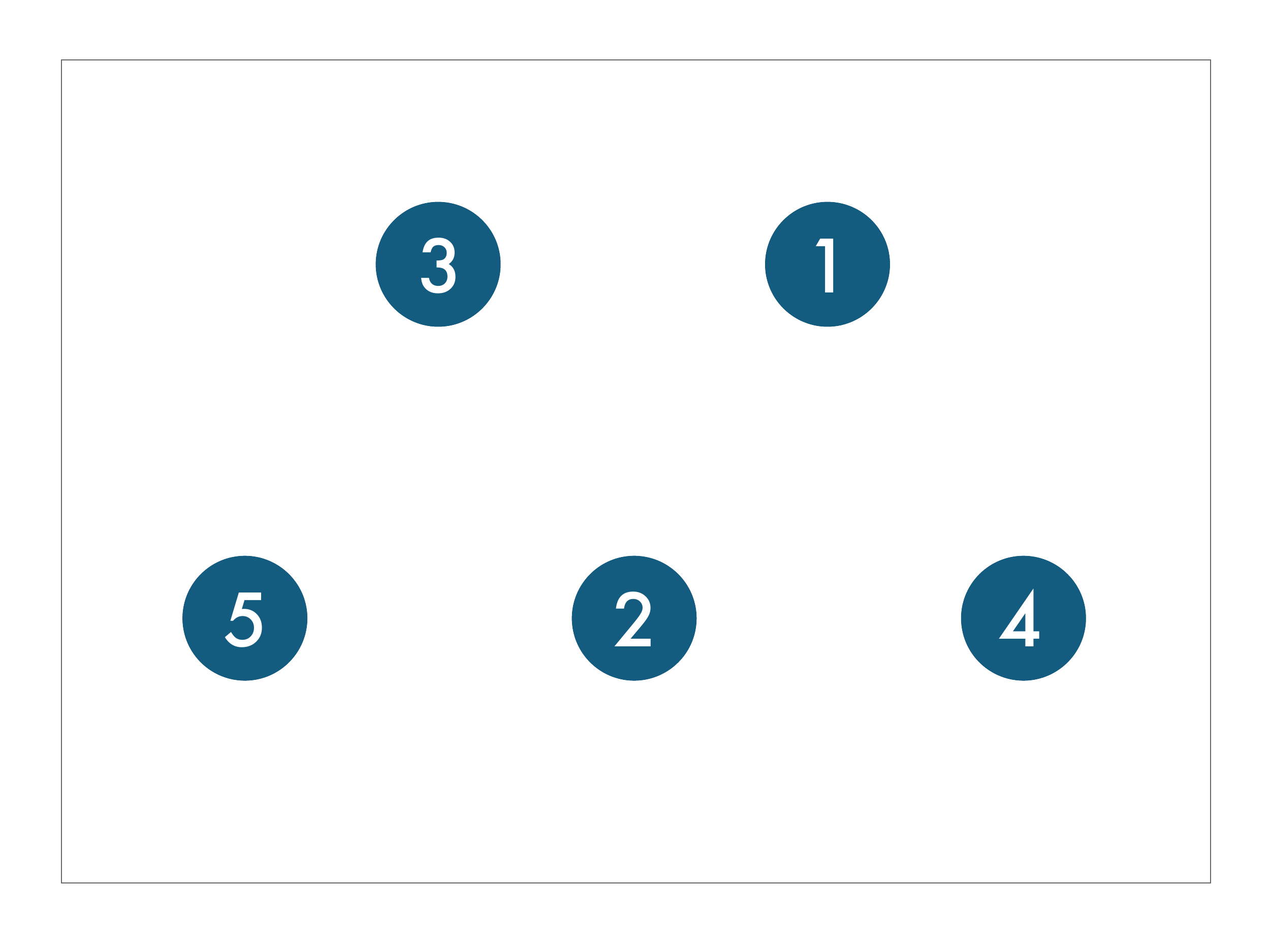}
}
\hspace{10mm}
\subfloat[]{
\label{fig:overlap_2}
\includegraphics[trim = 15mm 15mm 15mm 15mm, clip = true, scale = 0.13]{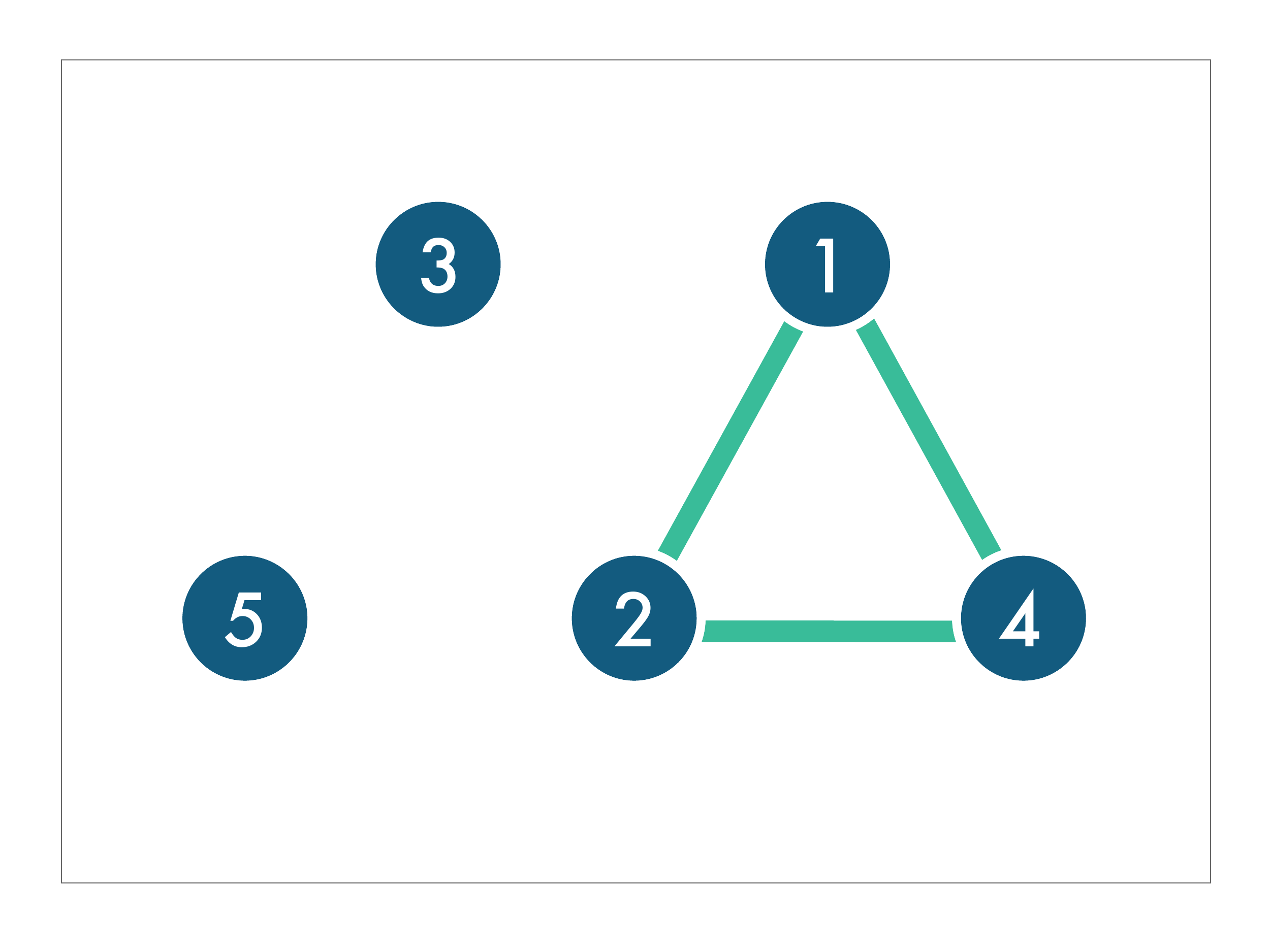}
}

\centering
\subfloat[]{
\label{fig:overlap_3}
\includegraphics[trim = 15mm 15mm 15mm 15mm, clip = true, scale = 0.13]{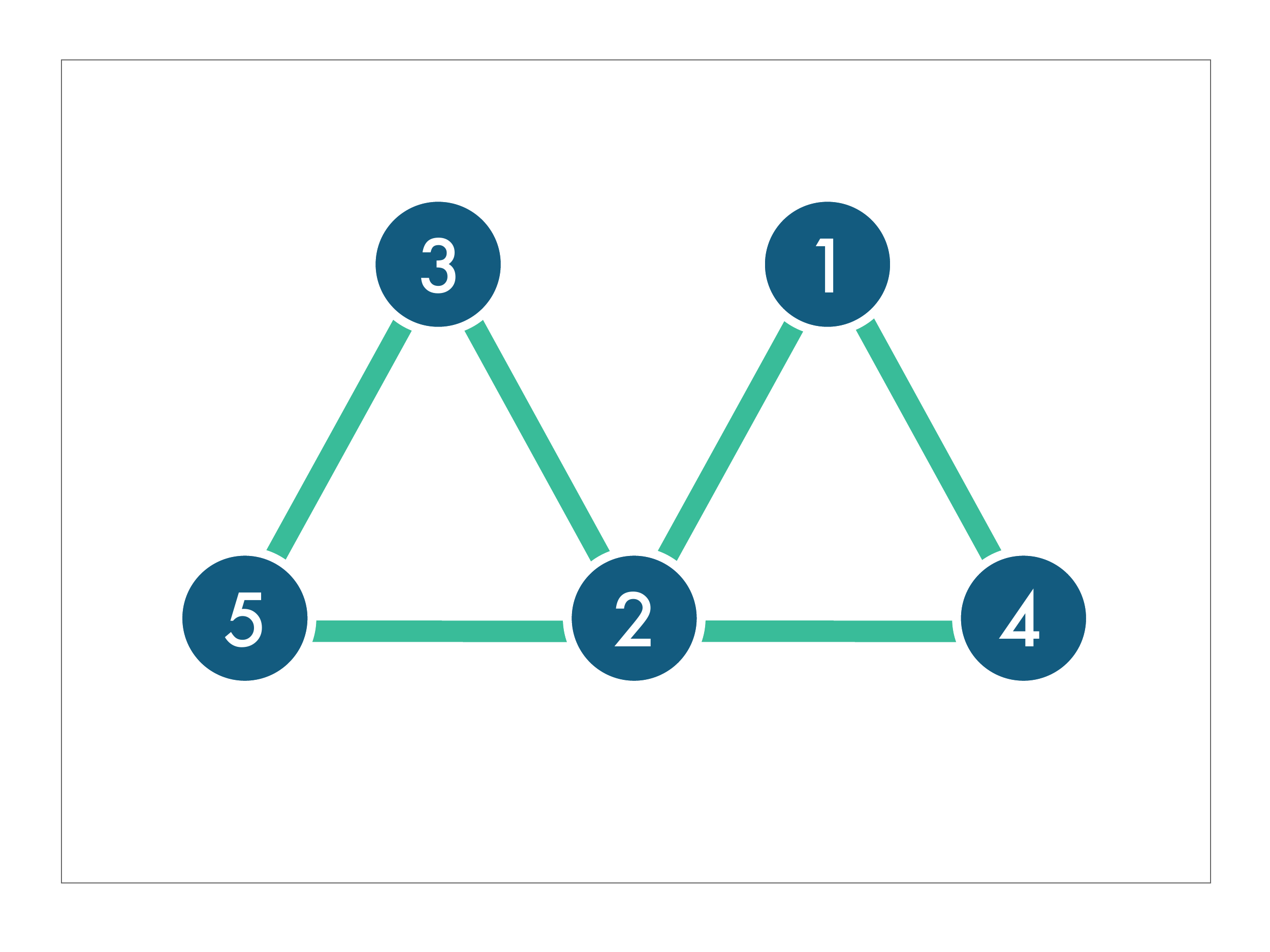}
}
\hspace{10mm}
\subfloat[]{
\label{fig:overlap_4}
\includegraphics[trim = 15mm 15mm 15mm 15mm, clip = true, scale = 0.13]{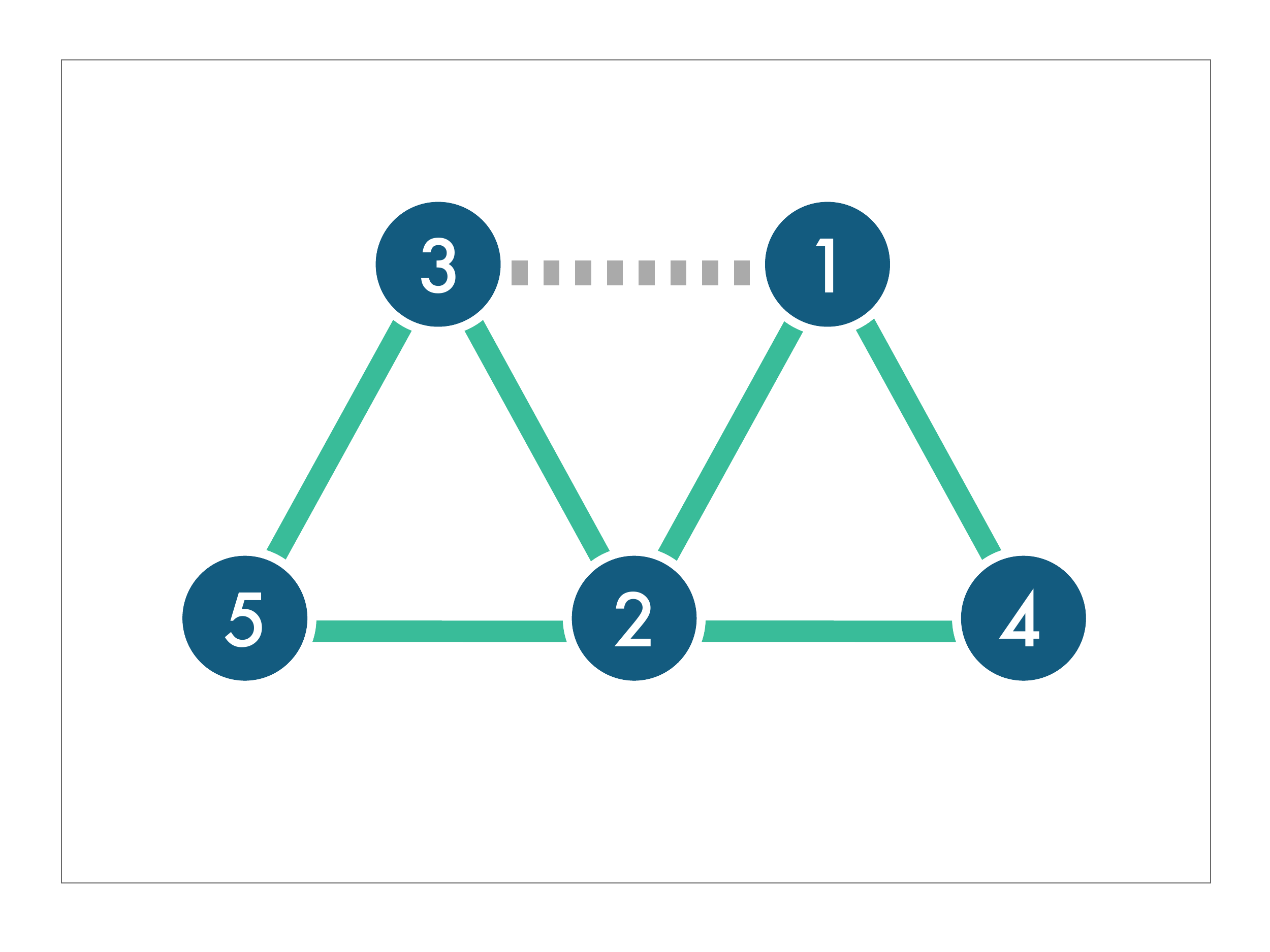}
}
\hspace{10mm}
\subfloat[]{
\label{fig:overlap_5}
\includegraphics[trim = 15mm 15mm 15mm 15mm, clip = true, scale = 0.13]{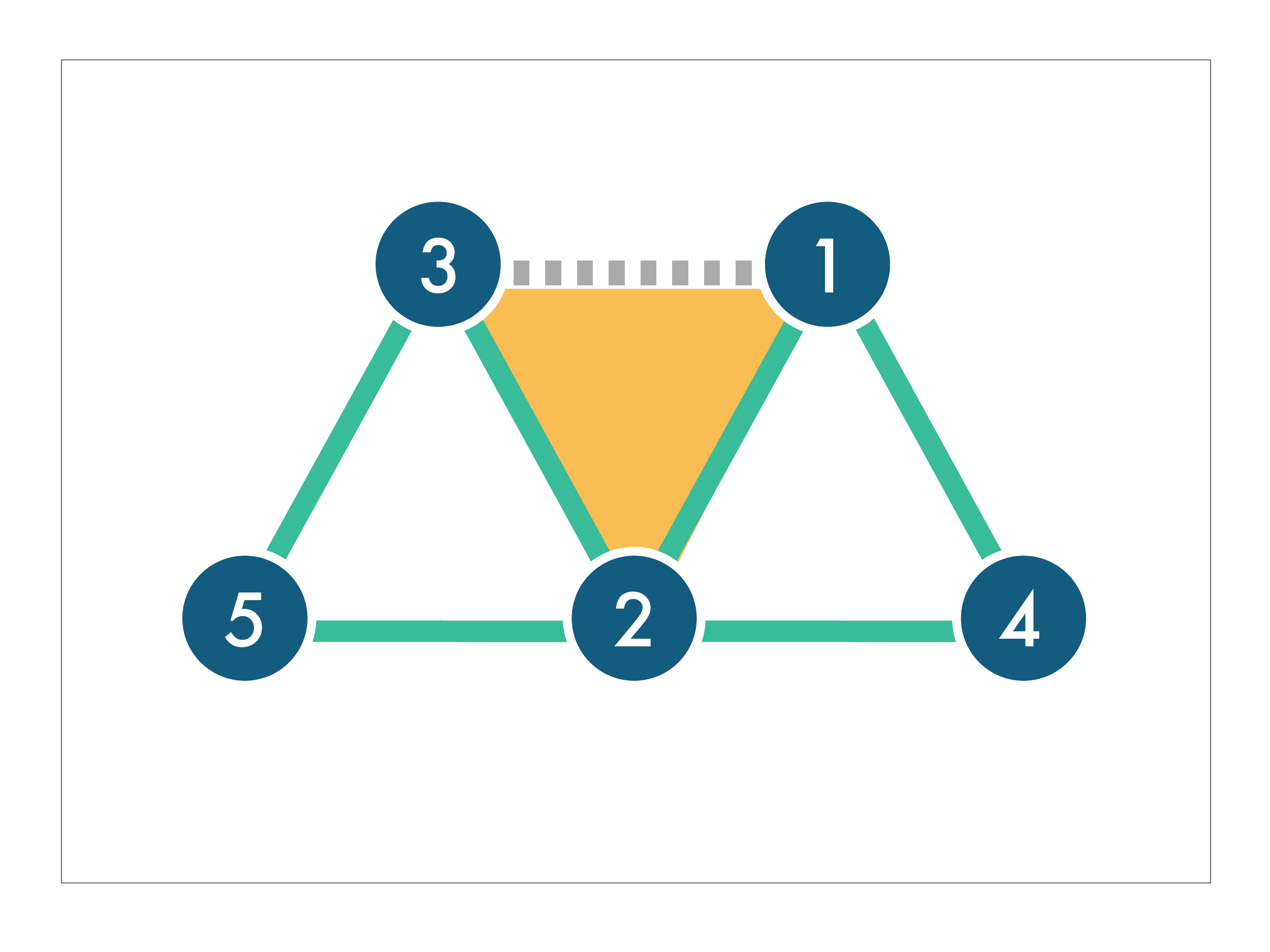}
}
\caption{\label{fig: overlap_multipanel} An incidentally generated triangle. (a) Triangle $124$ forms. (b) Triangle $124$ forms. (c) Triangle $235$ forms. (d) Link $13$ forms. (e) Triangle $123$ is incidentally generated.}
\end{figure}

This presents a challenge for estimating a parameter related to triangle formation since some of the triangles that we observe were
truly generated in the formation process, and others were ``incidentally generated;'' and similarly, it presents a challenge to estimating a parameter for link formation since some truly generated links end up as parts of triangles.

The key to our estimation in this section is that in cases where networks are sparse enough, then the fraction of incidentally generated subgraphs compared to truly generated subgraphs is
negligible.  Many applications satisfy the sparsity conditions and so the estimation techniques are applicable in many cases of interest.

To state results on the estimation of sparse SUGMs, we first need a few definitions.

Consider a sequence of SUGMs indexed by $n$, each with some $k$ sets of subgraphs that are counted, $G^n= (G^n_1, \ldots G^n_k)$, where $k$ is fixed for the sequence.
We say that the vector of sets of subgraphs $G^n= (G^n_1, \ldots G^n_k)$ is {\sl nicely-ordered} if the subnetworks in $G^n_\ell$ cannot be a
subnetwork of the subnetworks in $G^n_{\ell'}$ for $k\geq \ell'>\ell \geq 1$:
$$g_\ell\in G^n_{\ell} \ {\rm and } \ g_{\ell'} \in G^n_{\ell'}\ \ {\rm implies\ that} \ \ g_\ell \not\subset g_{\ell'}.$$
Note that any vector of sets of subgraphs can be nicely ordered: simply order them so that the number of links in the subgraphs
are non-increasing in $\ell$:  so that $\ell'> \ell$ implies that the number of links in
a subnetwork of type $\ell'$ is no more than the number of links in a subnetwork of type $\ell$.  For example, triangles precede links.

We then follow our accounting convention so that statistics count subgraphs in order and those which are not part of any previous subgraph:

\noindent $S_{\ell'}^n(g) = | \{ g_{\ell'} \in G^n_{\ell'} :  g_{\ell'}\subset g \ {\rm and \ }  g_{\ell'}\not\subset g_{\ell} {\rm \ for \ any \ }
g_{\ell} \in G^n_{\ell} {\rm \ such \ that \ } g_{\ell}\subset g {\rm \ for \ some \ } {\ell<\ell'}  \} |$.

We now define incidental generation and sparsity.

Consider a realization of a SUGM in which the truly generated subgraphs are given by $\Gamma \subset \cup_\ell G^n_\ell$,
 and let $g$ denote the realized network $g=\cup_{g'\in \Gamma} g'$.  The researcher observes $g$ and must make some inferences about $\Gamma$.
Fix a specific subgraph $g' \subset g$. We say that $g'$ is \emph{incidentally generated} by a subset of the (truly generated) subgraphs $\{g^j\}_{j\in J}\subset \Gamma$, indexed by $J$,
if:
\begin{itemize}
	\item[(i)] $g'$ was not \emph{truly generated} ($g'\notin \Gamma$),
	\item[(ii)] $g' \subset \cup_{j\in J} g^j$, and
  \item[(iii)] there is no $j'\in J$ such that $g'\subset \cup_{j\in J, j\neq j'} g^j$.
\end{itemize}
Part (ii) states that the subgraph is incidentally generated, and part (iii)  of the condition ensures that the set of generating subgraphs is minimal.

Despite minimality, a subgraph could still be generated in multiple ways.
For example, in Figure \ref{fig:overlap_5}, if the researcher
only observes the resulting network, there are various possibilities to be considered: the triangle $\{12, 23, 31\}$ could have been truly generated,  it could also have been
incidentally generated by the subgraphs $g^1, g^2, g^3$ where
$g^1 =\{12, 24, 41\}$, $g^2=\{23, 25, 53\}$ and $g^3=\{31\}$, it could have been incidentally generated by the subgraphs $g^1, g^2, g^3$ where
$g^1 =\{12\}$, $g^2=\{23\}$ and $g^3=\{31\}$, and still other possibilities.

\subsubsection{Generating Classes}

\

In order to define sparsity, we have to keep track of the various ways in which a subnetwork could have been incidentally generated.

Out of the many ways in which $g_\ell \in G^n_\ell$ could be incidentally generated, some of them are equivalent up to relabelings.
For instance, in a large graph any different combinations of triangles and edges could incidentally generate a triangle $g_\ell=\{12, 23, 31\}$,
however there are only eight ways in which it can be done if we ignore the labelings of the nodes outside of
$g_\ell$:  link 12 could be generated either by a triangle or link, and same for links 23 and 31, leading to $2^3=8$ ways in which this
could happen.

Consider  $g_\ell \in G^n_\ell$  that is incidentally generated by a set of subnetworks $\{g^j\}_{j\in J}$ with associated indices $\ell_j$ and also by another set
$\{g^{j'}\}_{j'\in J'}$.   We say $\{g^j\}_{j\in J}$ and $\{g^{j'}\}_{j'\in J'}$ are equivalent generators of $g_\ell$ if for each $g^j$ there is
$g^{j'}$ such that $\ell_j=\ell_{j'}$ and $g_j\cap g_\ell = g_{j'} \cap g_\ell$.
So generating sets play the same roles in $g_\ell$ but might involve different nodes outside of $N(g_\ell)$.
Equivalent sets of generators must have the same cardinality as they must both be minimal and involve the same intersections with $g_\ell$.

Given this equivalence relation, there are equivalence classes of generating sets of networks for any $g_\ell$.
There are at most $\left( \sum_{\ell'=1}^k m_{\ell'}\right)^{m_\ell}$ equivalence classes of (minimal) generating sets for any subnetwork $g_\ell$.\footnote{For each link in $g_\ell$ there
are at most $ \sum_{\ell'=1}^k m_{\ell'}$  links that could generate that link out of various subgraphs, and then the power is just the product of this across links in $g_\ell$, producing an upper bound.}
For each equivalence class $J$ of generating sets of some $\ell$, we have some list $(\ell_j, h_j)_{j\in J}$ of the types of subnetworks and the number of nodes that
the each subnetwork has intersecting with $g_\ell$.  We call these the {\sl (minimal) generating classes} of a subgraph $g_\ell$ and note that these are
the same for all members of $G^n_\ell$, and so we refer to them as the generating classes of $\ell$.

So, for a links and triangles example, where $G^n=(G_T,G_L)$ are triangles and links respectively, there are four generating classes of a triangle:  a triangle could be incidentally generated by three other triangles, two triangles and one link, two links and one triangle, or three links.\footnote{Here, our upper bound $\left( \sum_{\ell'=1}^k m_{\ell'}\right)^{m_\ell}$ is $4^3$, which is conservative.}
Here, then we would represent a generating class of two triangles and a link as $(T,2;T,2;L,2)$.

\subsubsection{Relative Sparsity}

\

Consider a set of nicely ordered subgraphs $G^n=(G_1^n, \ldots, G_k^n)$ and any $\ell\in \{1,\ldots, k\}$ and any
generating class of some $\ell$, denoted $J=(\ell_1,h_1;\ldots; \ell_{|J|,h_{|J|}}$.
Let\footnote{Note that $M_J\geq 1$ since $|J|\geq 2$ and each set of $h_j$ nodes intersects with at least one other set of $h_{j'}$ nodes for some $j'\neq j$. Recall that under the nice ordering, smaller subgraphs cannot be generated as a subset of some single larger one .}
$$M_J=(\sum_{j\in J} h_j) -m_\ell.$$
For example, in forming a triangle from any combination of triangles and links, each $h_j=2$ and so $M_J=6-3=3$.

We say that a sequence of models as defined in Section \ref{newmodel} with associated nicely-ordered subgraphs $G^n=(G_1^n, \ldots, G_k^n)$
and parameters $p^n=(p_1^n, \ldots, p_k^n)$  is {\sl relatively sparse} if for each
$\ell$ and associated generating class $J$ with associated $(\ell_j, h_j)_{j\in J}$:
$$\frac{\prod_{j\in J} \E_{p^n} (S^n_{\ell_j}(g))}{ n^{M_J} \E_{p^n} (S^n_\ell(g) )}\rightarrow 0.$$
This is a condition that limits the relative frequency with which subgraphs will be incidentally generated (the numerator) to directly generated (the denominator).

To make this concrete, consider our example with triangles and links.
A triangle can be generated by other combinations of links and triangles.  The expected number of triangles that nature generates directly is
$\E_{p_T}[ S_T^n(g)]= p_T \binom{n}{3}$ and the number of links not in triangles is (approximately) $\E_{p_L}[ S_L^n(g)]= p_L \left(\binom{n}{2} - O\left( p_T \binom{n}{3} \right) \right)$.
Thus it must be that for each generating class,
$$\frac{\prod_{j} \E_{p^n} (S^n_{\ell_j}(g))}{ p_T n^6}\rightarrow 0.$$
For the generating class of all triangles, this implies that
$p_T^2 n^3 \rightarrow 0$, so $p_T=o(n^{-3/2})$.
For the generating class of all links, this implies that\footnote{Given that $p_T=o(n^{-3/2})$, it follows that $\E_{p_L}[ S_L^n]= p_L \left(\binom{n}{2} - O\left( p_T \binom{n}{3} \right) \right)$ is proportional to
 $p_L n^2$.}
$p_L^3 / p_T \rightarrow 0$, which is the obvious condition that triangles formed by independent links are rare compared to triangles formed directly.
This implies that (but is not necessarily implied by) $p_L=o(n^{-1/2})$.
The conditions on the remaining generating classes (some links and some triangles) are implied by these ones.

For example, letting $p_T= a(n)/n^2$ and $p_L = b /n$, where $a(n)=o(n^{1/2})$ satisfies the sparsity conditions.\footnote{This leads to an expected degree of $b+a(n)/3$ and an average clustering of roughly $\frac{a(n)}{6(b+a(n)/3)(b+a(n)/3+1)}$.
This can be consistent with various clustering rates, and admits rates of links and triangles found various observed networks. To match very high clustering rates the model can be altered to include cliques of larger sizes.}

\subsection{Estimation of Sparse Models}

\

Let
$\widetilde{S}^n$ denote the vector of the numbers of subnetworks of various types that are truly generated; this is not observed by the researcher since the resulting $g$ may include incidental  generation.  Let ${S}^n(g)$ the observed counts including the incidentally generated subnetworks. In Figure \ref{fig:multi_panel}, $\widetilde{S}^n_T=9$ but ${S}_T^n(g)=10$ and from observing $g$ there is no way to know exactly what the true $\widetilde{S}^n_T$ is, we just have an upper bound on it.  Meanwhile, $\widetilde{S}_U^n = 23$, but as one truly generated link becomes part of an incidentally generated triangle, it follows that ${S}_U^n = 22$.

\begin{figure}[h]
\centering
\subfloat[$n$ nodes]{
\label{fig:image_1}
\includegraphics[width=0.33\textwidth]{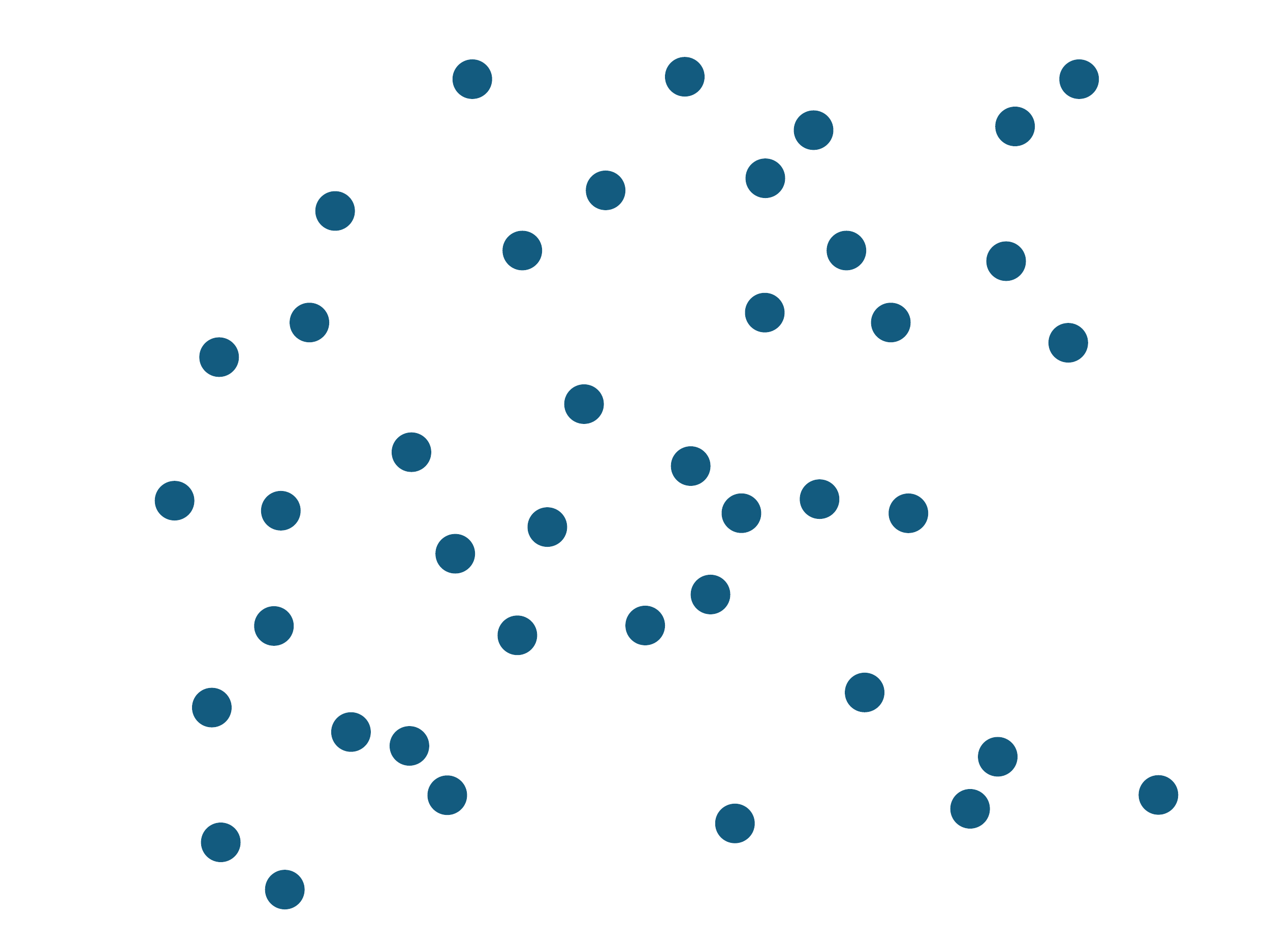}
}
\subfloat[Triangles form]{
\label{fig:image_2}
\includegraphics[width=0.33\textwidth]{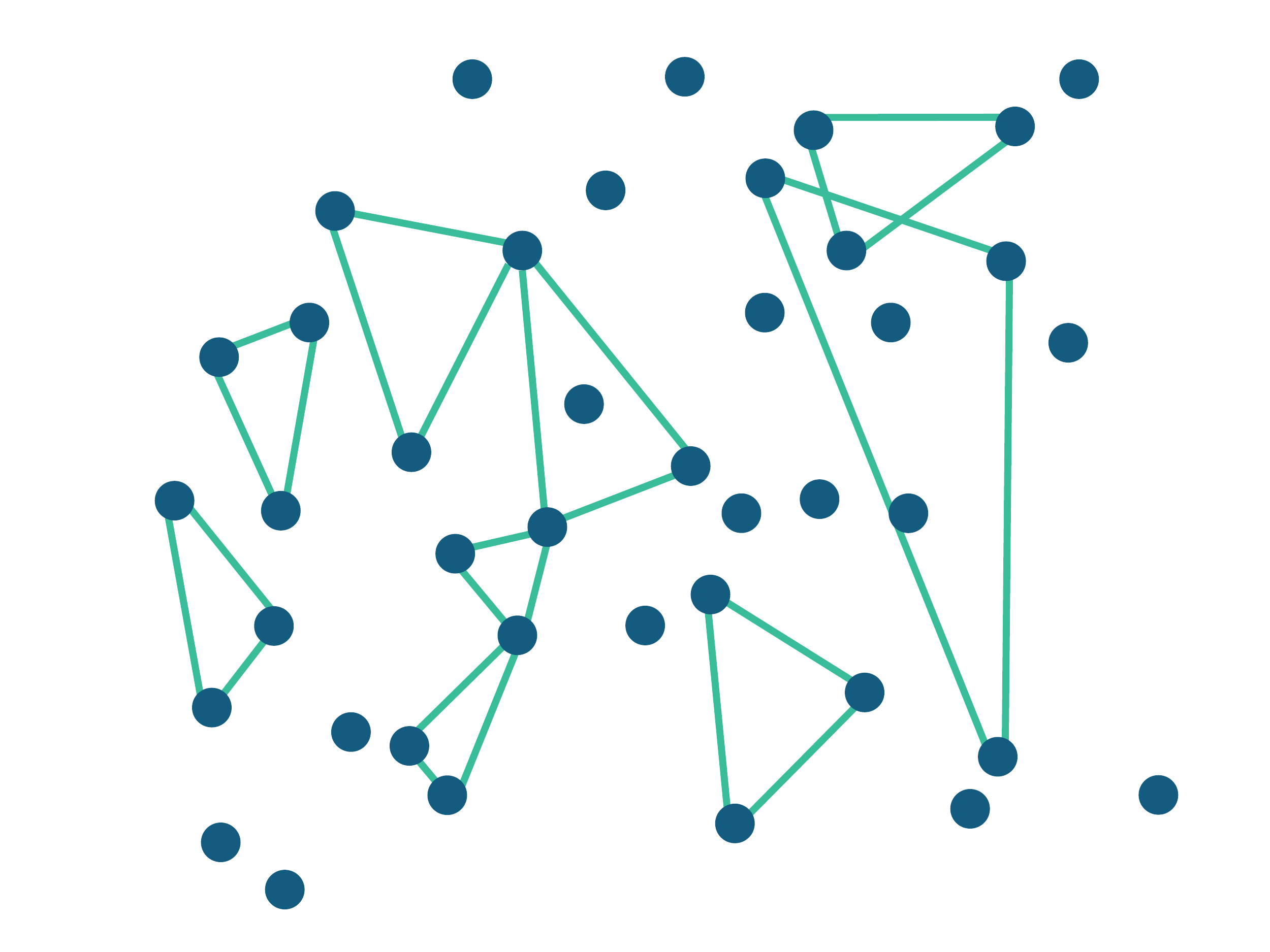}
}

\centering
\subfloat[Links form]{
\label{fig:image_1}
\includegraphics[width=0.33\textwidth]{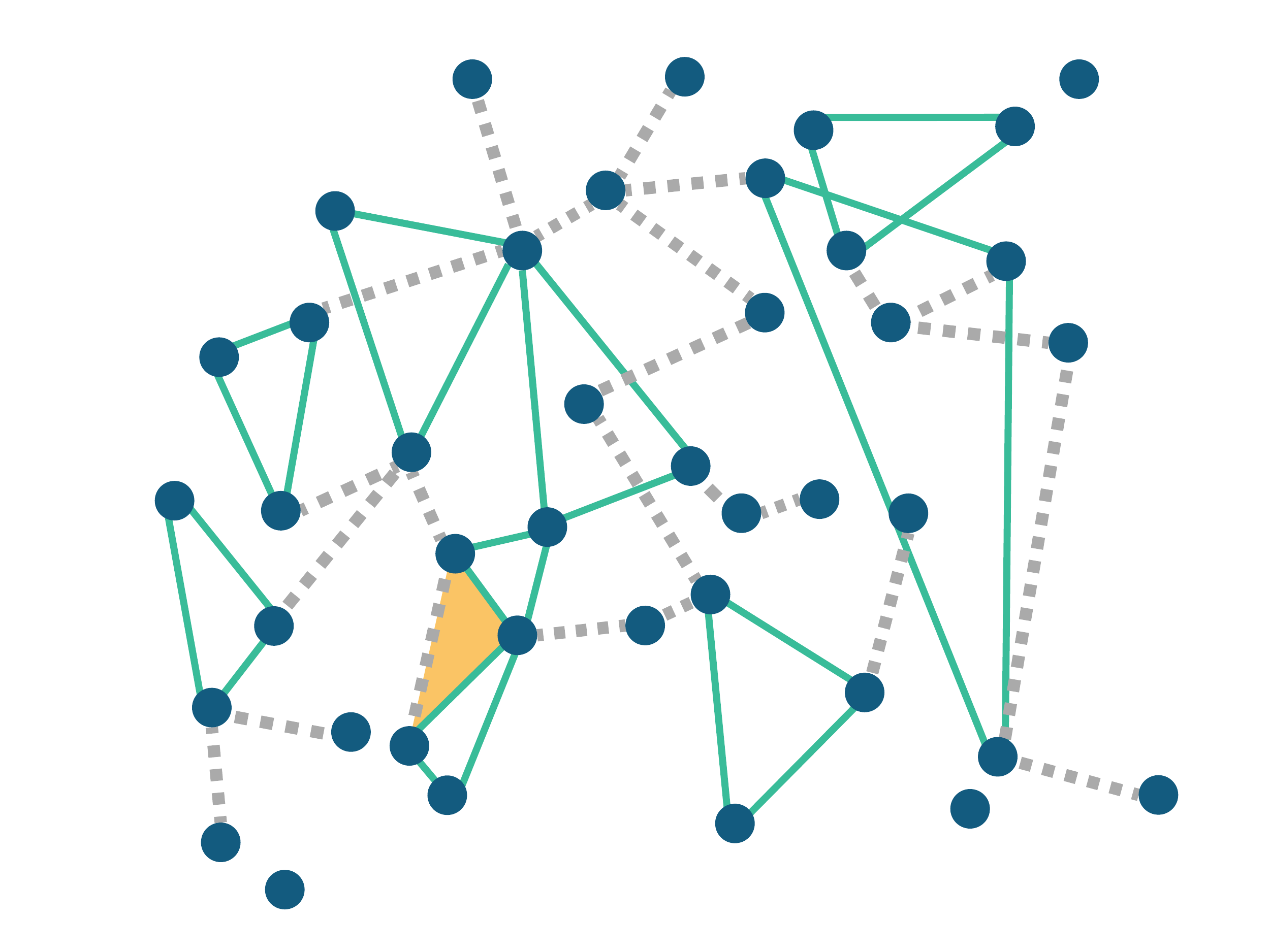}
}
\subfloat[Resulting network]{
\label{fig:image_1}
\includegraphics[width=0.33\textwidth]{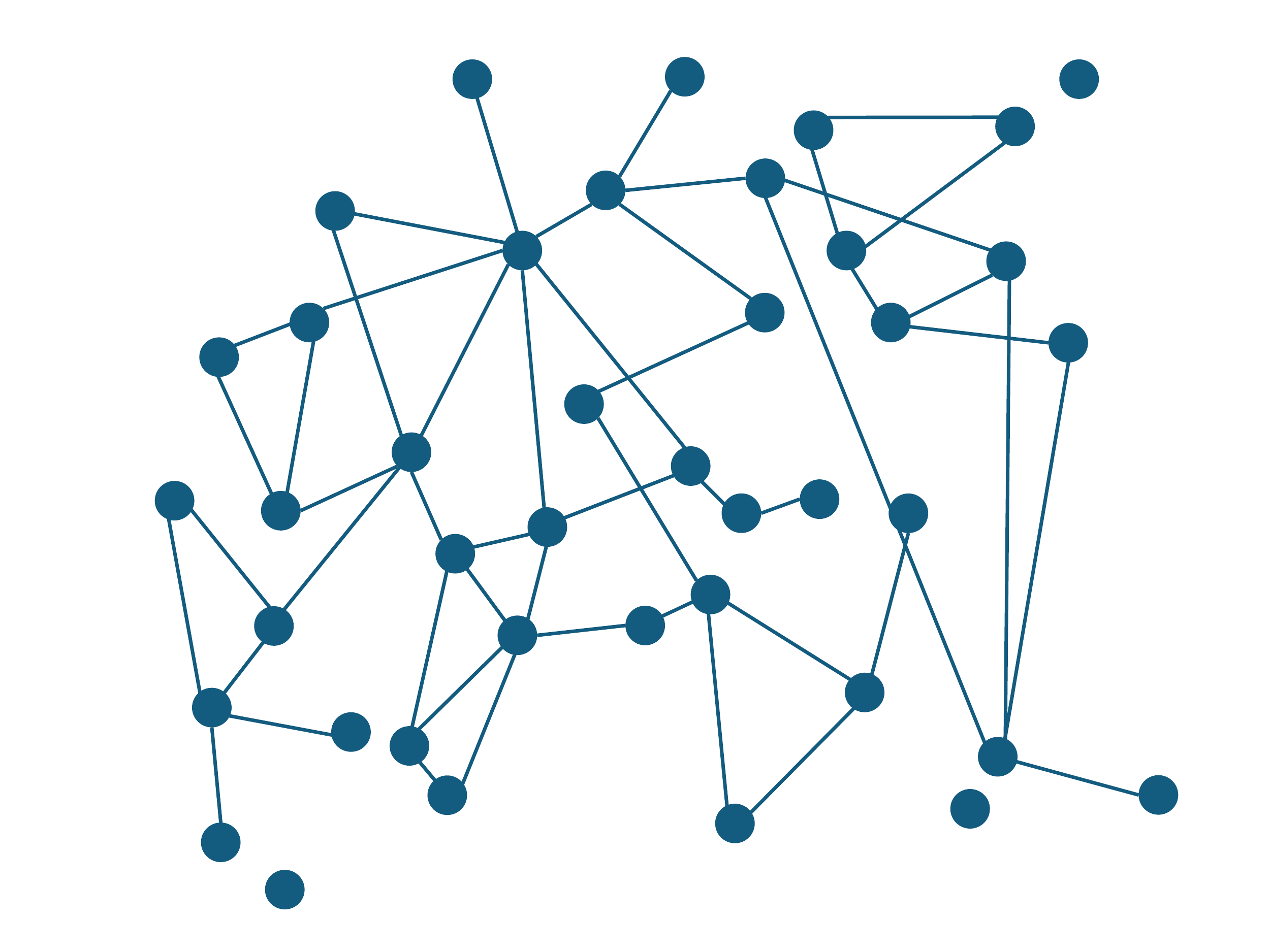}
}
\caption{\label{fig:multi_panel} The network that is formed  and eventually observed is shown in panel D.
The process can be thought of as first forming triangles independently with probability $p_{T}^n$ as in (B), and then forming links  independently with probability $p_{L}^n$ on the remaining part of the graph as in (C).  We see that there is one incidence of an extra triangle  generated by this process.  In this network we would count ${S}_T^n(g)=10$ and ${S}_U^n(g)=22$ from (D), while the true process generated $\widetilde{S}_T^n(g)=9$ and $\widetilde{S}_U^n(g)=23$.}
\end{figure}

Nonetheless, as we prove, under the sparsity condition we can accurately estimate the true statistics and thus the true parameters.

To state our next result, we need the following notation.
Let $\overline{S}^n_\ell(g)$ be the maximum count of $S^n_\ell$ that is possible on network $g$.  If we are counting triangles and links not in triangles, then $\overline{S}^n_T(g) = \binom{n}{3}$ and $\overline{S}^n_U(g) = \binom{n}{2} - L_T(g)$ where $L_T(g)$ is the number of links that are part of triangles in $g$.\footnote{In sparse networks, $L_T(g)$ would be vanishing relative to $\binom{n}{2}$ and so could be ignored.  Typically, in sparse networks, $\overline{S}^n_L(g)$ will be well approximated by $y_\ell\binom{n}{m_\ell} $, where $y_\ell$ is the number of possible different subgraphs of type $\ell$ that can be placed on $m_\ell$ nodes (e.g., $y_\ell$ is 1 for a triangle, $m$ for a star on $m$ nodes, etc.).}
Let
\begin{equation}\label{phat}
\widehat{p}_{\ell}^{n}(g) = \frac{{S}^n_\ell(g)}{\overline{S}^n_\ell(g)}.
\end{equation}
So, $\widehat{p}_{\ell}^{n}(g)$ is the fraction of possible subgraphs counted by $S_\ell^n$ that are observed in $g$ out of all of those that could possible exist in $g$.
This is a direct estimate of the parameter $p_\ell^n$, as if these subgraphs were each independently generated and not incidentally generated.

In order to have $\widehat{p}_{\ell}^{n}(g)$ be an accurate estimator of ${p}_{\ell}^{n}(g)$ in the limit two things must be true.  First, the network must be relatively sparse, which limits the number of incidentally generated subgraphs.  And, second, it must be that the potential number of observations of a particular kind of subgraph grows as $n$ grows.  This would happen automatically in a sparse network setting if we were simply counting triangles and links not in triangles.
However, if nodes have different characteristics (say some demographics), and we are counting triangles and links by node types, then it will also have to be that the number of nodes that have each demographic grows as $n$ grows.   If there are never more than 20 nodes with some demographic, then we will never have an accurate estimate of link formation among those nodes.

We say a SUGM is {\sl growing} if the probability that $\widetilde{S}^n_\ell(g)\rightarrow \infty$ for each $\ell$ goes to 1.

\begin{theorem}[Consistency and Asymptotic Normality]
\label{sparse}
Consider a sequence of growing and relatively sparse SUGMs with associated nicely-ordered subgraph statistics $S^n=(S_1^n, \ldots, S_k^n)$
and parameters $p^n=(p_1^n, \ldots, p_k^n)$.  The estimator \eqref{phat} is ratio consistent: 
$ \frac{\widehat{p}_\ell^n(g)}{p_\ell^n} \cvgto 1 $ for each $\ell$.
Moreover, for $D_{n}={\rm Diag}\left\{ \widehat{p}_{\ell}^n(g) n^{m_{\ell}}\right\} _{\ell=1}^{k}$
\[
D_{n}^{1/2}\left(\left(\widehat{p}_{1}^n,...,\widehat{p}_{k}^n\right)'-\left(p_{1}^n,...,p_{k}^n\right)'\right)\rightsquigarrow\mathcal{N}\left(0,I\right).
\]
\end{theorem}

Theorem \ref{sparse} states that growing and relatively sparse SUGMs are consistently estimable via easy estimation techniques: ones that are direct and trivially computable.

The proof of the theorem involves showing that under the growing and sparsity conditions the fraction of incidentally generated subnetworks vanishes for each $\ell$, and so the observed counts of subnetworks converge to the true ones.
Given that these are essentially binomial counts, then, as the second part of the theorem states, a variation on a central limit theorem applies and then normalized errors in parameter estimation are normally distributed, and we know the rates at which the parameters converge to their limits.
For inference and tests of significance for single parameter values we note that analytic estimates of the variances are directly computable from the analytic expression of the diagonal of the variance matrix.  Of course, more complex inferential procedures and tests can be executed through a standard parametric bootstrap as the model is easily simulated.

To make the convergence rates concrete, consider the example with links and triangles and let $p_T = a/n^2$ and $p_L = b/n$.  These are well within the bounds that would be needed to satisfy sparsity, but provide an example of a realistic level of sparsity that satisfies our conditions for asymptotic normality.  Then one can check the incidental generations for triangles is $o_p(n^{1/2})$, which means that the {\sl fraction} of incidentally generated triangles is $o_p(n^{-1/2})$.  Here, the normalization $D$ means that the errors on link estimation will be of order $n^{-1/2}$ and on triangle estimation of order $n^{-3/2}$, and so parameter estimates converge very quickly.

Again, we emphasize that although the estimator here is based on binomial approximations, a SUGM still
incorporates interdependencies directly through the subgraphs that are
generated.  The results make use of the fact that in sparse settings, the picture of interdependencies is clear and are measured by the statistics one-by-one.

\subsection{An Algorithm for Estimating SUGMs without Asymptotic Sparsity}\label{fsc}

\

We provide an algorithm for estimating SUGMs for cases where sparsity may not be satisfied,
or for small graphs where finite-sample corrections could be useful.

The idea behind the algorithm is that we create a network by randomly building up subgraphs in a way that ends up matching the observed network, and we keep track of how many {\sl truly} generated subgraphs of each type were needed to get to a network that matched the observed statistics.
In order to estimate the truly generated subnetworks of each type, $\widetilde{S}_{\ell}(g)$,  we carefully construct a simulated network
$g_{sim}$ and keep track of both its truly generated subgraphs $\widetilde{S}_{\ell}(g_{sim})$ and its observed subgraphs ${S}_{\ell}(g_{sim})$.
We construct $g_{sim}$ to have ${S}_{\ell}(g_{sim})$ match ${S}_{\ell}(g)$ as closely as possible, and then use its true subgraphs $\widetilde{S}_{\ell}(g_{sim})$ to infer the true subgraphs of $g$, $\widetilde{S}_{\ell}(g)$.

Consider a SUGM with nicely-ordered subgraphs indexed by $\ell \in \{1, \ldots , k\}$.
We describe the algorithm for the case where 
subgraphs of type $k$ (the smallest subgraph - links in most models) cannot be incidentally generated by other subgraphs.\footnote{
If the smallest subgraphs can be generated incidentally (for instance if a model only included triangles and cliques of size 4), then begin the algorithm at step t and treat subgraphs of type $k$ symmetrically with all other subgraphs (so drop the first part of step t).}

\

\noindent {\bf {\sc Algorithm}}

\begin{itemize}

	\item[0.] Start with an empty graph $g_{sim}^0$. Set 
	 $S_{\ell}(g_{sim}^0)=0$ and $\Stil_{\ell}(g_{sim}^0) = 0$ for all $\ell$.

	\item[1.] Place $S_{k}(g)$ subgraphs uniformly at random (these will be links in most models).  Call the new network $g_{sim}^1$.  This may generate
some incidental subgraphs.   Update counts of each $S_{\ell}(g_{sim}^1)$ and $\Stil_{\ell}(g_{sim}^1)$ (with the latter only having truly generated links so far).

	\item[$t$.]
		\begin{itemize}
\item If $S_{k}(g_{sim}^{t-1}) < S_{k}(g) $, then place  $S_{k}(g) - S_{k}(g_{sim}^{t-1})$ subgraphs down uniformly at random.  Call the new network $g_{sim}^{t}$ and proceed to step $t+1$.
			\item Otherwise, pick subgraph of type $\ell$ with the minimal ratio $\Stil_{\ell}(g_{sim})/S_{\ell}(g)$.  Add one subgraph of type $\ell$ uniformly uniformly at random.\footnote{Add it uniformly at random out of candidate subgraphs that
are not already a subgraph of some existing subgraph of $g_{sim}^{t-1}$.  For instance, if adding a triangle, only consider triangles that are not already a subset of some clique of size 3 or more of the generated network through this step.}  Call the new network $g_{sim}^{t}$ and proceed to step $t+1$.
\item If $\Stil_{\ell}(g_{sim})\geq S_{\ell}(g)$ for all $\ell$, stop.
		\end{itemize}
\end{itemize}

The estimates are $\phat_{\ell} = \Stil_{\ell}(g_{sim})/\overline{S}_{\ell}(g).$

\

To see the intuition behind the algorithm consider a case with just links and triangles.  The algorithm takes advantage of the fact that links can generate triangles, but not the other way around.    First the algorithm generates unsupported links up to the number observed in $g$.
This might lead to some triangles, and lowering the number of observed links.  The algorithm then tops up the links and keeps doing so until the correct observed number of links
are present.  If there are fewer triangles than in $g$, it begins adding triangles one at a time (as they might incidentally generate more).  At each step, if the number of links drops below
what are in $g$, then new links are added.  It continues until the correct number of links and triangles are obtained.  It can never overshoot on links, and may slightly overshoot on triangles,
only by the incidentals generated in the last steps.

There are many variations one could consider on the algorithm.\footnote{More generally a Method of Simulated Moments (MSM) approach could also  be taken.
 For that, one simply searches on a grid of parameters, in each case simulating the SUGM and then picking $\phat$ for $$\phat := \argmin_{p} \left( S(g) - \E_{p}\left[S(g^{\rm Sim})\right]\right)'C\left( S(g) - \E_{p}\left[S(g^{\rm Sim})\right]\right).$$}
For example, if one is conditioning on various covariates, then there might be more than one type of link, and since all types of links cannot be incidentally generated one can ``top up'' several types of subgraphs and not just $k$.  Thus, in step 1 instead of using just $k$ above,
one might also use $k-1$, etc., for however many types of links there are, and similarly for the first part of step $t$.\footnote{
To fix ideas, consider a SUGM in which there are two types of triangles and two types of links that are generated, accounting for covariates (as we will use
in Section \ref{sims}).   For instance, links between pairs of nodes that are `close' in terms of the characteristics and pairs of nodes that are `far', and triangles involving
nodes that are all `close' and triangles that involve some nodes that are `far' from each other.  The statistics that we count for
a network $g$ are: $S_{T,C}(g)$, $S_{T,F}(g)$,$S_{U,C}(g)$,$S_{U,F}(g)$, where $U$ is for unsupported links and $T$ for triangle, and $C$ is for `close' and $F$ is for `far'.}

\subsection{The Relation between SUGMs and SERGMs}\label{SUGM_SERGM}

\

We now show a relationship between SUGMs and SERGMs.

It is easiest to see the connection by considering the variation of a SUGM where nature forms various subgraphs with a probability $p_\ell$
of a given subgraph $g_\ell\in G^n_\ell$ forming without worrying about whether they overlap, so it could form a triangle and also form a link that already belongs to that triangle.
For instance, in a nicely ordered SUGM, if nature first formed triangles and then links outside of triangles, if the triangle between nodes 1,2, and 3 was formed, then the links 12, 23, and 13 would not be added later.  In this variation of a SUGM, nature forms links and triangles without caring about overlap, so it might form the triangle 1,2,3  and then also the link 12.
The formation of a given subgraph is independent of other subgraphs. Again, let $\widetilde{S}_\ell$ denote the count of truly generated subgraphs $g_\ell\in G^n_\ell$.

Define $\theta_\ell$ by
\begin{equation}
p_\ell=\frac{\exp(\theta_{\ell})}{\exp(\theta_{\ell})+1} \ {\rm or} \  \theta_\ell=\log \left( \frac{p_\ell}{1-p_\ell} \right).
\label{logtheta}\end{equation}

\begin{theorem}[SERGM Representations of SUGMs]
\label{potential}
Suppose that the probability that subgraph of type $\ell$ forms is given by (\ref{logtheta}).
This form of SUGM can be represented in a SERGM form:
\begin{equation}
\label{potent2}
\Prob_\theta^n(\widetilde{S}) = \frac{K^n(\widetilde{S})\exp\left(\widetilde{S}\cdot\theta\right)}{\sum_{s' }K^n(s')
\exp\left(s'\cdot \theta\right)
},
\end{equation}
where $K^n_\ell(s_\ell)=\binom{\overline{S}_\ell^n}{s_\ell}$ and $K^n(s)= \prod_\ell K^n_\ell(s_\ell)$.
\end{theorem}

Theorem \ref{potential} provides a relationship between SUGMs and SERGMs.  The two are closely related, although the statistics counted here are the {\sl actual subgraphs} (including overlaps) that nature generated, which can be estimated but not precisely known.

Note that this also provides a reason why in specifying SERGMs it is useful to have $K$'s that differ from the $N$'s that correspond to some ERGM model.   Here specific $K$'s are natural and yet differ from an ERGM formulation.

A direct implication of Theorem \ref{potential} is the following, which provides a general result on dynamic processes of network formation, where subgraphs are repeatedly considered and added and deleted over time.

\begin{corollary}
Consider any dynamic process such that with probability one, each subgraph is considered infinitely often, and when a subgraph is considered it is added with probability (\ref{logtheta})
if not already present and deleted with the complementary probability if it is already present.  The resulting dynamic process has a steady state distribution given by
(\ref{potent2}).
\end{corollary}

\section{Strategic/Preference-Based Random Network Models}\label{strategic-utility}

\

As we have discussed above, SERGMs and SUGMs admit models where both choice and chance are important, and we describe a couple of examples to illustrate how preferences of
individuals over networks can be incorporated.

\subsection{Mutual Consent Formation Models}

\

Here we describe a strategic network model that harnesses some of the power of Theorem \ref{sparse}.
The key aspect of the model is that decisions to link are not only bilateral but instead multilateral: sub-groups of
individuals decide whether or not to form
subgraphs.   A pair of individuals may meet and decide (mutually) as to whether to add a single link, but also
a group of three (or more) may meet and decide whether to form a some subgraph such as a clique or some other
form (e.g., a ring, star, etc.).\footnote{For additional theoretical underpinnings of coalition-based network formation models see \cite{jacksonvanden2005,caulier2009}.}
Moreover, the probability that they form the subgraph could depend
upon the characteristics of the individuals involved.

Consider subgraphs $g_\ell\in G^n_\ell$ with associated individuals $N(g_\ell)$.

The members of $N(g_\ell)$ meet according to a random process and have the opportunity to form $g_\ell$.
Both the probability with which the members meet and their preferences for forming $g_\ell$ can depend
on their characteristics $X(g_\ell)=(X_i)_{i\in N(g_\ell)}$.

There are certain aspects of the members' characteristics, $H_i(X(g_\ell))$, that affect $i$'s benefits from the subgraph.\footnote{
For instance if $X(g_\ell)$ were a list of the individuals' ages, then it might be that $i$'s benefit from the subgraph is a function of $i$'s distance from the average characteristics:
$h_i(X(g_\ell))= \left\Vert X_i -\sum_{j\neq i } X_j / (m_\ell-1)\right\Vert$.  It could also be that $i$ benefits from the maximum value of $X_{-i}$, or suffers from variation in the characteristics.  $h$ can be tailored to the specific application and list characteristics.  }

There is a probability
$\pi_\ell$ that a subgraph $g_\ell$ of individuals with characteristics $X(g_\ell)$ meets and decides whether
to form $g_\ell$.
So it might be more likely that individuals of similar ages meet than ones with different ages.

Individual $i$ obtains a utility\footnote{Here we simplify notation by omitting the dependence of the utility on a given individual's position in the subnetwork. Everything stated here extends directly allowing utility to depend on position:  for instance, getting higher utility from being the center of a star rather than on its periphery, but the notation becomes cumbersome.}
\[
U_{i,\ell}(X(g_\ell))=\gamma_{\ell,X_i} H_i(X(g_\ell)) - \epsilon_{i,\ell}
\]
from the formation of a given subnetwork $g_\ell$, where $\gamma_{\ell,X_i}$ depends on the subnetwork in question and possibly on the
characteristics of $i$ and $\epsilon_{i,\ell}$ is a random idiosyncratic term.

The subnetwork then forms conditional upon it having met if and only if $U_{i,\ell}(X(g_\ell)) \geq 0$ for all $i$ (say with at least one strictly positive).\footnote{ This then corresponds nests pairwise stability as defined by \citet{jackson1996strategic}, subject to the meeting process. One can adjust this to take into account other rules for group formation, and this also easily handles directed networks.}  If the error term has an atomless distribution, then the strictness is inconsequential.
Let $F_{\ell}(\cdot)$ be the distribution of error terms for the formation of subgraphs in $G^n_\ell$.
So, the probability that a subgraph $g_\ell\in G^n_\ell$ with characteristics $X(g_\ell)=(X_i)_{i\in N(g_\ell)}$ is formed is
\begin{equation}
p_{\ell}=\pi_\ell \prod_{i\in N(g_\ell)} F_{\ell}\left(\gamma_{\ell,X_i} H_i(X(g_\ell))\right).
\label{eq:utility}
\end{equation}

Let $S_{\ell}(g)$ denote the number of subnetworks in $G^n_\ell$ (consisting of individuals with characteristics $X(g_\ell)$) that form in a network $g$, counted excluding
networks already counted as subset of some $\ell'<\ell$, as under the well-ordered condition.
Under the conditions of Theorem \ref{sparse}, the SUGM in (\ref{eq:utility}) is then easily estimated and consistent.\footnote{
Although the probabilities of various subgraphs are directly estimable (and hence identified) under the conditions of Theorem \ref{sparse}, of course whether the various parts of
$\pi_\ell \prod_{i\in N(g_\ell)} F_{\ell}\left(\gamma_{\ell} H_i(X(g_\ell))\right)$
are well identified depends on the specifics of the the functional forms involved.
Just as an example, consider a situation with
two types.  Set $X_i=1$ if $i$ is of type 1 and $X_i=2$ if $i$ is of type 2 and $H_i(X)= X_i-X_j$  and consider $m=2$.
Then $\gamma_2$ and $-\gamma_2$ both lead to the same value of $ F_{2}\left(\gamma_2 (X_i-X_j)\right)  \times F_{2}\left(\gamma_2 (X_j-X_i)\right)$.
So here it could not be judged whether the type 2 has a greater expected utility (net of the random term) from the match than the type 1 or whether it is the reverse.
There are some obvious simplified
formulations that allow for identification, for example setting instead $h_i(X)= \left\vert X_i-X_j \right\vert $. It might also require specifying a (nonlinear) functional form for $\pi_\ell$ as in
\citet{currarini2010pnas}.}

It is important to note that such a formulation allows us to do welfare computations and
changes in welfare due to changes in, say, the distribution of $X$
if they include parameters that a policymaker may control - or it may be that a policymaker could change the $\pi$ function in some well-defined way by say, subsidizing interactions among groups with certain sorts of
characteristics who might tend to meet infrequently.

\subsection{Strategic Network Formation and Potential Functions}\label{pairwise-strategic}

\

There is a nice connection between strategic network formation models and potential functions that spans a series of papers:  \cite{jacksonwatts2001exist}, \cite{butts2009}, \cite{mele2011}, \cite{badev2013}.   For example, \cite{butts2009} and \cite{mele2011} show that
if links are recognized independently over time, and then added or deleted based on individual choices according to a logistic function, then the steady-state distribution can be represented as an ERGM.  In those models, only one agent makes a decision at a time and links must be directed.
Here we generalize the set of models that are covered, and also extend to allow for mutual consent.
We also provide a directed version of the formation model which generalizes the results of \cite{mele2011}.

Agent $i$'s payoff from network $g$ can depend on which subgraphs $i$ is a member of, as well as things such as to whom his or her neighbors are connected.

Let $\mathcal{G}$ denote some set of subgraphs from which agents derive utility.

Consider some agent $i$ and a subgraph $g_\ell \in \mathcal{G}$ that $i$ is a member of, $i\in g_\ell$.  Let
the members of $g_\ell$ (including $i$) have a vector of characteristics described by $ X_\ell$.
Agent $i$ gets some marginal payoff,
$$  v(g_\ell, X_\ell), $$
from having this subgraph in the network
where this function can depend on the type of subgraph and the characteristics of all the agents involved in the subgraph.

Agent $i$'s utility from a network $g$ is then
\begin{equation}
\label{upotential}
u_i(g) =  \sum_{g_\ell \in \mathcal{G}, i\in g_\ell}  v(g_\ell, X_\ell).
\end{equation}
Note that this allows the agent's utility to  depend on the presence of ``friends of friends'' by including subgraphs of the form $g_\ell=\{ij,jk\}$.  Of course, it also allows agents' payoffs to depend on direct links, cliques, and so forth.

Next, consider a network formation process where agents can form links in pairs and they add the link whenever their {\sl mutual} gain
is positive.   The idea is that they can bargain and make side payments (either in cash or by exchange of favors) to add links whenever those links are mutually beneficial.

In particular, a network $g$ is to be {\sl pairwise stable with transfers} if:
\footnote{This definition is from \cite{blochjackson2006} and is related to the definition of pairwise stability allowing for side payments that appears in the conclusion of \cite{jackson1996strategic}.}
\begin{itemize}
\item $ij\in g$
implies that $u_i(g)+u_j(g) \geq u_i(g-ij)+u_j(g-ij)$,
and
\item
$ij\notin g$
implies that
$u_i(g)+u_j(g) \geq u_i(g+ij)+u_j(g+ij)$.
\end{itemize}

For this setting, we can then define the following function $f$, which is a potential function for network formation under pairwise stability with transfers:
\begin{equation}
\label{potentialC}
f(g)  =  \sum_{g_\ell \subset g}2v(g_\ell, X_\ell).
\end{equation}
It follows that $f(\cdot)$ is a potential function for a network formation game that follows pairwise stable with transfers.
In particular, direct calculations show that for any $g$ and $ij\in g$:
\begin{equation}
\label{potentialB}
f(g) - f(g-ij)  = \left(u_i(g)+u_j(g)\right) - \left(u_i(g-ij)+u_j(g-ij)\right).
\end{equation}
Thus, the difference between the value that $f$ assigns to $g$ and what it assigns to $g-ij$ is exactly the
sum of the differences that $i$ and $j$ assign to the two networks.

For any such setting, Theorem 1 in \cite{jacksonwatts2001exist} implies that there exists at least one network that is pairwise stable with transfers, and moreover that there are no cycles in the improving paths.\footnote{An improving path is a sequence of networks that differ from each other by one link such
that if a link is added or deleted then the pair of agents in the link see an increase in their summed utilities.}

Now let us describe a dynamic process of network formation.  Let $g^0$ be some starting network, and let $g^t$ denote the network in place at the end of time $t$.
Let $g^t_{-ij}$ denote the network of links other than $ij$.
In each period, there exists some positive probability of each given link being recognized (the two agents in question ``meet'').  The recognition probabilities can depend on the pair in question and the network in place at the time and the probability that link $ij$ is recognized conditional on the network in place $g^t$ is denoted $p(ij, g^t_{-ij})$.\footnote{Since each link probability could depend on the network other than the link $ij$, if each link's recognition probability does not exceed $1/\binom{n}{2}$ then the sum of all link recognition probabilities
does not exceed 1, and that leave the residual probability $1-\sum_{ij} p(ij, g^t_{-ij})$ to be the probability that there is no link recognized in the current period and the period simply advances.  The scaling of probabilities is irrelevant to the steady state, and so it is fine to allow periods to pass without any recognition.}

We emphasize that the meeting process is quite general as it is allowed to depend on the attributes of the agents $i$ and $j$ as well as the network in question.  Thus, for example, it allows their meeting probability to depend on whether or not the pair have friends in common, and can even depend on how many friends in common they have, and can depend on any other aspect of the network, as well as the agents' characteristics.

Once recognized at some time $t$, $i$ and $j$ decide whether to add or delete the link, conditional upon the rest of the network in place at that time $g^{t-1}_{-ij}$.
The probability that the link is added/kept is a logistic function of the mutual value of the link:\footnote{With a slight abuse of notation, we allow $g^t_{-ij}{+ij}$ to denote
the network where $ij$ is present and the network of other links is described by $g^t_{-ij}$, and similarly $g^t_{-ij}{-ij}$ denotes
the network where $ij$ is not present and the network of other links is described by $g^t_{-ij}$.}
\begin{equation}
\label{log-add}
\frac{\exp\left(u_i(g^t_{-ij}+ij)+u_j(g^t_{-ij}+ij)\right)
}{\exp\left(u_i(g^t_{-ij}+ij)+u_j(g^t_{-ij}+ij)\right) + \exp\left(u_i(g^t_{-ij}-ij)+u_j(g^t_{-ij}-ij)\right)}.
\end{equation}

This defines an aperiodic and irreducible Markov chain over the space of all networks, and so it has a unique steady-state distribution.
Moreover, it is a reversible Markov chain and
the unique steady-state distribution
is given by\footnote{We omit the standard proof as, for instance, it is a direct extension of the proof of Theorem 1 in \cite{mele2011}, noting that the link recognition probability
can depend on $g^t_{-ij}$ without affecting the steps of his proof.}
\[
\Pr(g)=\frac{\exp(f(g))
}{\sum_ {g'} \exp(f(g'))}.
\]
Thus, this is an ERGM:
\[
\Pr(g)= \frac{\exp\left( \sum_{g_\ell \subset g}2v(g_\ell, X_\ell)\right)}{\sum_{g'}\exp\left(  \sum_{g_\ell \subset g'}2v(g_\ell, X_\ell)\right)}.
\]
We can then rewrite $\sum_{g_\ell \subset g}2v(g_\ell, X_\ell)$ as a function that simply keeps track of statistics of how many subgraphs of a network $g$ are of a given form, $(\ell, X_\ell)$, denoted $S_{\ell, X_\ell}(g)$.
 This then has a SERGM representation:
\[
\Pr(s)= \frac{\exp\left( \sum_{\ell} s_{\ell, X_\ell} 2v_{\ell, X_\ell}\right) }{\sum_{s'}\exp\left(  \sum_{\ell} s'_{\ell, X_\ell} 2v_{\ell, X_\ell}\right)}.
\]

\noindent We remark that link recognition probabilities do not enter the final steady-state distribution, which is only determined by the preferences as captured via the $v$ functions.

\subsection{Directed Network Formation}

\

Everything stated above has an analog for directed links $ij$ where the decision to add the link is taken by agent $i$ (with stability defined by Nash equilibrium), and where the subgraphs, $g_\ell$'s, are directed.  The only change is to drop the `2' in the above formulas and require that each agent obtain utility from each subgraph in which they direct some link.\footnote{The model can be specified to allow agents to derive utility from subgraphs in which they have some ``in-links'' but no ``out-links'', but can also allow them not to.}  The directed version of the above generalizes a result in \cite{mele2011}.

\subsection{Search Intensity Models}

\

Another interesting class of strategic/random network formation models that we can extend to the setting here are where agents face overall costs of forming relationships - not just costs associated with various subgraphs (as in the models above).
To account for such overall tradeoffs in the network formation processes, we can also include search intensities as have been analyzed in various formation models such as
\citet{currarini2009economic,currarini2010pnas,borgs2010,golub2010netform}.
Those models are of bilateral link formation; but are easily extended to more general SUGMs as we briefly describe.

Each agent $i$ with characteristics $X_i$ puts in a search effort $e(X_i,m,X)\in [0,1]$ to form cliques of size $m$ with characteristics $X$.
$m=2$ indicates links, and so $e(X_i,2,(X_i,X_j))$ is the effort that agent $i$ expends in trying to form links with agents who have characteristics $X_j$; and $e(X_i,3,(X_i,X_j,X_h))$ is the effort that agent $i$ expends in trying to form triangles where the other two agents have characteristics $X_j$ and $X_h$, and so forth.
``Effort'' is simply a shorthand for either the time spent socializing, 
or else it could simply indicate a relative openness to forming relationships of various types.

An agent obtains a utility
$
u (X_i,m, X)
$
from being of type $X_i$ in a clique of size $m$ with characteristics $X$.

The probability that a given clique $Cl_m$ forms depends on the vector of efforts for such cliques of those in the clique, $(e_j(,m,X))_{j\in Cl_m}$, according to a function $\pi_{m,X}((e_j(m,X))_{j\in Cl_m})$
that is nondecreasing in each of its arguments.

An agent also pays a cost of network formation: $c(X_i, (e_i(m,X))_{m,X})$ that depends on his or her characteristics $X_i$ and the
search efforts that he or she exerts in forming various links and cliques, $(e_i(,m,X))_{m,X}$.

Thus, an agent $i$'s overall expected payoff as a function of the all of the agents' efforts is described by
\[
\left(\sum_{m,Cl_m: i\in Cl_m} \pi_{m,X}((e_j(m,X))_{j\in Cl_m})u (X_i,m, X)\right) -  c(X_i, (e_i(m,X))_{m,X})
\]

In a case where the $u$'s are nonnegative, this defines a supermodular game: agent $i$'s change in payoff from increasing any dimension
of $(e_i(m,X))_{m,X}$ is nondecreasing in the vector of strategies $(e_j(m,X))_{j\neq i, m,X}$.
In such games, pure strategy equilibria exist and form a complete lattice (e.g., see \citet{topkis2001}).
Additional conditions on $\pi$, $u$, and $c$ can ensure uniqueness of equilibrium, depending on the specific functional forms that are used
to parameterize the model, or one can appeal to equilibrium selection.\footnote{Here there are positive spillovers/externalities from strategies,
and so generally the maximal equilibrium will Pareto dominate the others, and so a standard refinement would be to look at the Pareto efficient equilibrium which is then unique and pure (e.g., see \citet{vives2007} for some background).}

Models of this structure define SUGMs, where the relative frequencies $p_{m,X}$ of cliques of size $m$
consisting of agents with characteristics described by the profile $X$.
Specifying functional forms for $\pi$, $u$ and $c$ then allows for estimation of parameters of the model and of the equilibrium, provided the specification is tight enough to be well-identified.

Although the above formulation is described for cliques, it is easily adjusted for any subgraphs (for instance an agent may value being the center of a star with $m$ agents).  In the obvious extension one needs to keep track of the positions of the various types of agents in the subgraph as there are then asymmetries in positions and, for instance, agents might care about the characteristics of the agent at the center of a star.

\section{Network properties generated by SUGMs}\label{applications}

\

First we compare a simple model that estimates linking probabilities based on node characteristics (caste and geography) with a SUGM that also includes triangles.  The idea is to compare how well these replicate various features of actual networks, such as clustering, the size of the giant component, average path length, degree distributions, and various eigenvalue properties of the adjacency matrices.

For this exercise we use the \citet{banerjee2013diffusion} data consisting of networks in 75 Indian villages with an average of 220 nodes.  Here we focus on ``advice'' networks: an edge represents whether a household speaks to another household when having to make an important decision and we use an undirected, unweighted graph. This is a simple representation of the informational network structure within the sample of villages, and the networks are reasonably connected (with more than two-thirds of the nodes being in a giant component) and yet also reasonably sparse for small networks.

In addition to the average degree and clustering (which are at least partly captured by links and triangles), we are interested in other quantities motivated by theory.
We look at the first eigenvalue of the adjacency matrix, which is a measure of diffusiveness of a network under a percolation process (e.g., Bollob\'{as} et al., \citet{jackson2008social}).  A related quantity is the spectral gap, which is the difference in the magnitudes of the first and second eigenvalues of the adjacency matrix. This is intimately related to the expansiveness of the network -- namely, for any subset of nodes the number of links leaving the subset relative to the number of links within the subset.
We are also interested in the second eigenvalue of the stochasticized adjacency
matrix.\footnote{The stochastized adjacency matrix $T$ is defined as
$T_{ij} = \frac{g_{ij}}{\sum_{k} g_{ik}}$, where either $g_{ii}=1$, or
$g_{ik}>0$ for some $k\neq i$, as this captures the set of people to whom $i$ listens.}
This is a quantity that is key in local average learning processes and modulates
the time to consensus (\citet{demarzo2003persuasion}, \citet{golub2009homophily}).
Additionally, we look at the fraction of nodes that belong to the giant component
of the network, as empirical networks are often not completely connected.
Finally, we consider average path length (in the largest component).

Our procedure is as follows.  For every village, we estimate each of two network formation models.  The first network formation model
is a link-based model where the probabilities can also depend on geographic and caste covariates.  In particular, pairs of household are categorized as either being ``close'' or ``far'' and then separate probabilities of links are estimated for ``close'' and ``far'' pairs.  ``Close'' refers to pairs of nodes that are of the same caste and are below the median geographic distance (the median GPS distance taken across all pairs of households), and ``far'' to those that either differ in caste or are further than the median distance.
The second network formation model is a SUGM with the same structure except for the addition of triangles.\footnote{Similarly, we categorize triangles as being ``close'' if all nodes are of the same caste and
all pairs are below the median distance, and ``far'' otherwise. }
We estimate parameters for the village network for each model and then
generate a random network from the model based on the estimated parameters.   We do 100 such simulations for each of
the 36 villages and for each of the two models.  We then compare the aforementioned network characteristics from the simulations with the actual data.\footnote{We have complete GPS and caste data for only 36 villages.}

Table \ref{tab-emprops}  presents the results.  We find that networks simulated from the SUGM better match the structural
properties exhibited by the empirical Indian village networks than those simulated
from a link-based model.

\begin{table}[h]
\caption{Network Properties}\label{tab-emprops}
\begin{center}
\includegraphics[trim = 15mm 194mm 0mm 29mm, clip = true, scale = 0.9]{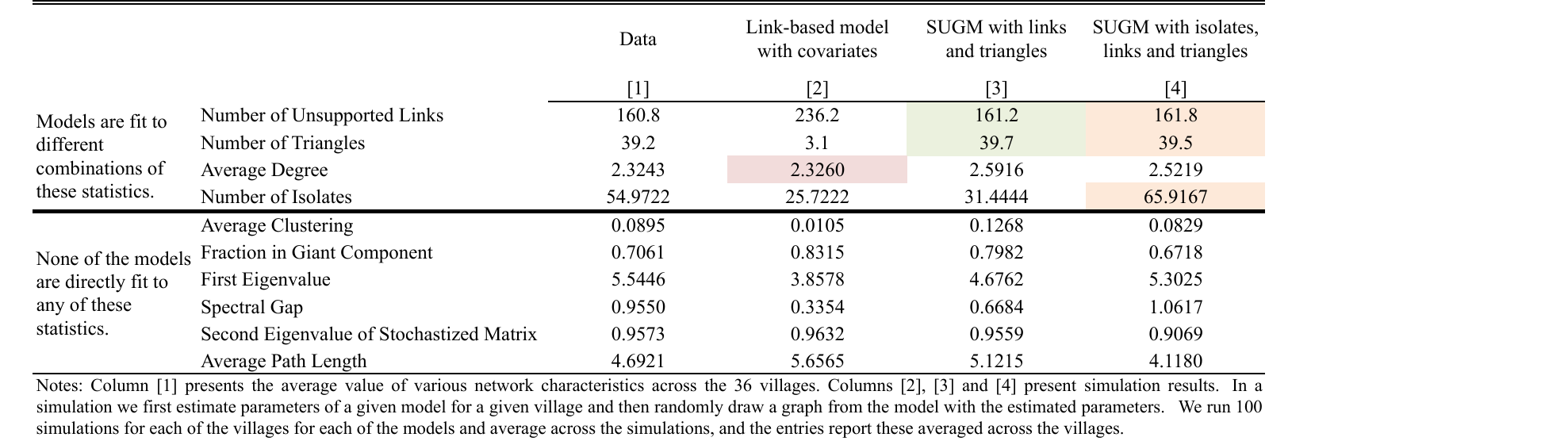}
\par\end{center}
\end{table}

Both the SUGM and the link-based model do quite well for average degree.  As expected, the SUGM matches the triangle count and the unsupported link count (as these are the statistics on which the model is based) whereas the link-based model matches average degree quite closely (as this is the moment on which this model is based).

Neither model is based on the remaining statistics.
The first and most obvious thing to note is that the link-based model does extremely poorly when it comes to matching clustering while the SUGM does much better, which is natural given that the SUGM explicitly includes triangles.
More interestingly, conditioning on the triangles in the SUGM is enough to deliver better matches on all of the other dimensions. For instance, the link-based model considerably underestimates the first eigenvalue (3.86 as compared to 5.54), whereas the SUGM performs better (4.68). Similarly, the link-based model underestimates the expansiveness of the networks with a spectral gap of 0.34 instead of 0.96.  The SUGM again performs considerably better (0.67). These sorts of results also hold true for the average path length, fraction of nodes in the giant component, and the second eigenvalue of the stochasticized matrix.

Beyond these two models, we also fit a SUGM that includes isolates, in addition to links and triangles.   Not surprisingly, it fits isolates better than either of the previous models.
The more interesting aspects are in the other features to which none of the models are fit.  Here we see that including isolates significantly improves, beyond the improvement from triangles, the fits on clustering, the size of the giant component, the first eigenvalue, and spectral gap.   Accounting for isolated nodes changes the density among remaining nodes in ways that better match the overall structure of the network.   The dimension on which it does not perform as well is the second eigenvalue (the homophily measure).  However, that is likely because the model is not sufficiently geared towards the covariates that affect segregation, and so densifying the remaining network reduces segregation.  Including a richer set of covariates into the model would help counter-act that, but is beyond our illustrative purposes here.

We also examine distributional outcomes. In Figure \ref{fig-Compare}, we show CDFs of node degrees and clustering. The CDFs from the empirical data are computed as follows. For every village, we compute the degree and clustering coefficient for each 5th percentile from 5 to 95. We then average these values across the villages in our sample. The simulated CDFs are computed by taking the analogous cross-village average from simulated data as described in Table \ref{tab-emprops}.  For parsimony, we compare only the isolates-links-triangles SUGM and the links-based model.

Figure \ref{fig:DegreeDistribution} shows the degree distributions. The SUGM does considerably better than the links-based model in matching the entire degree distribution. Specifically, the links-based model undershoots both the lower and upper tails of the degree distribution, despite hitting the average correctly. The SUGM, though slightly overshooting the average degree, better matches the distribution overall.

Figure \ref{fig:Clustering} shows the distribution of clustering coefficients. The link-based model is unable to generate any non-trivial clustering and essentially has a degenerate distribution (the short red curve in the upper left). The SUGM generates a distribution similar to the data, significantly outperforming the link-based model.

\begin{figure}[h!]
\centering
\subfloat[Average CDF of Degree]{
\label{fig:DegreeDistribution}
\includegraphics[trim = 0mm 0mm 0mm 0mm, clip,width=0.5\textwidth]{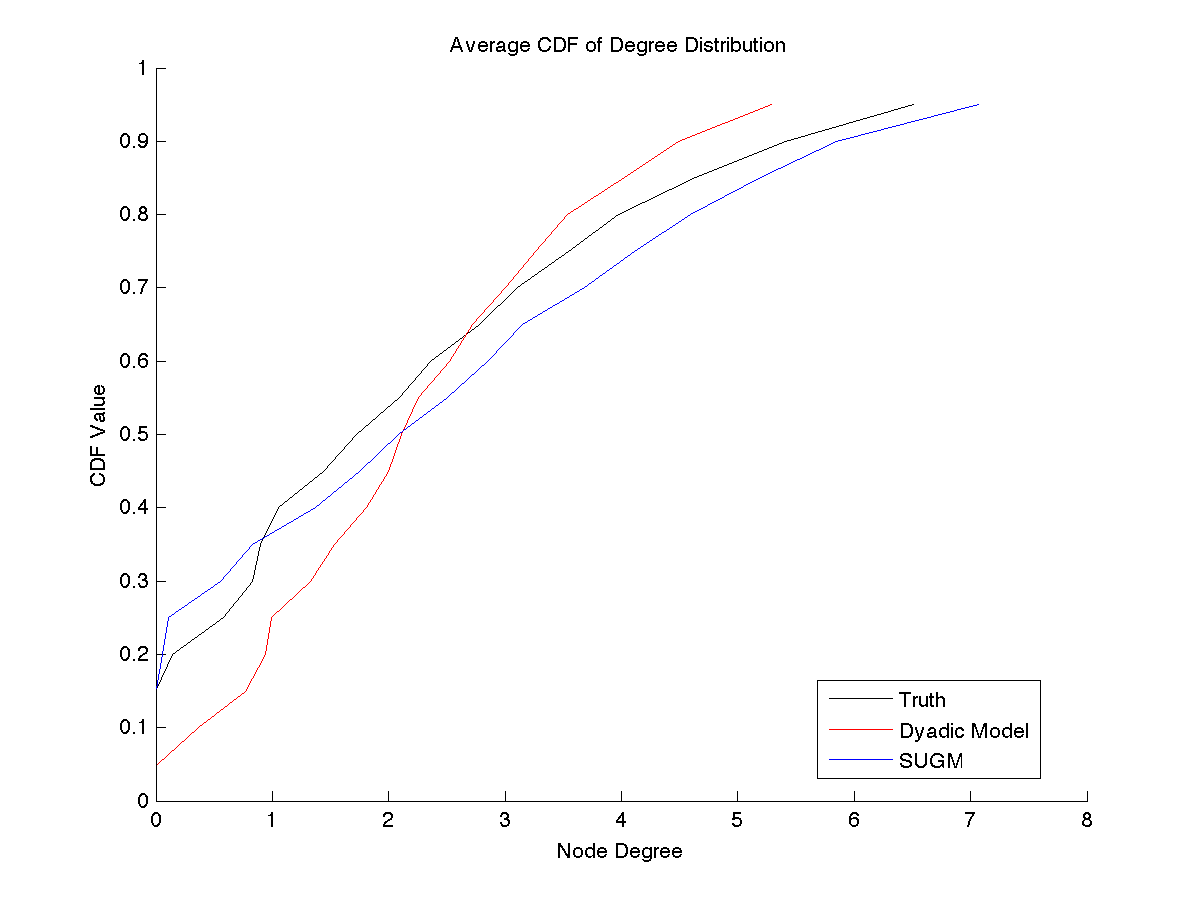}
}
\subfloat[Average CDF of Clustering Coefficient]{
\label{fig:Clustering}
\includegraphics[trim = 0mm 0mm 0mm 0mm, clip,width=0.5\textwidth]{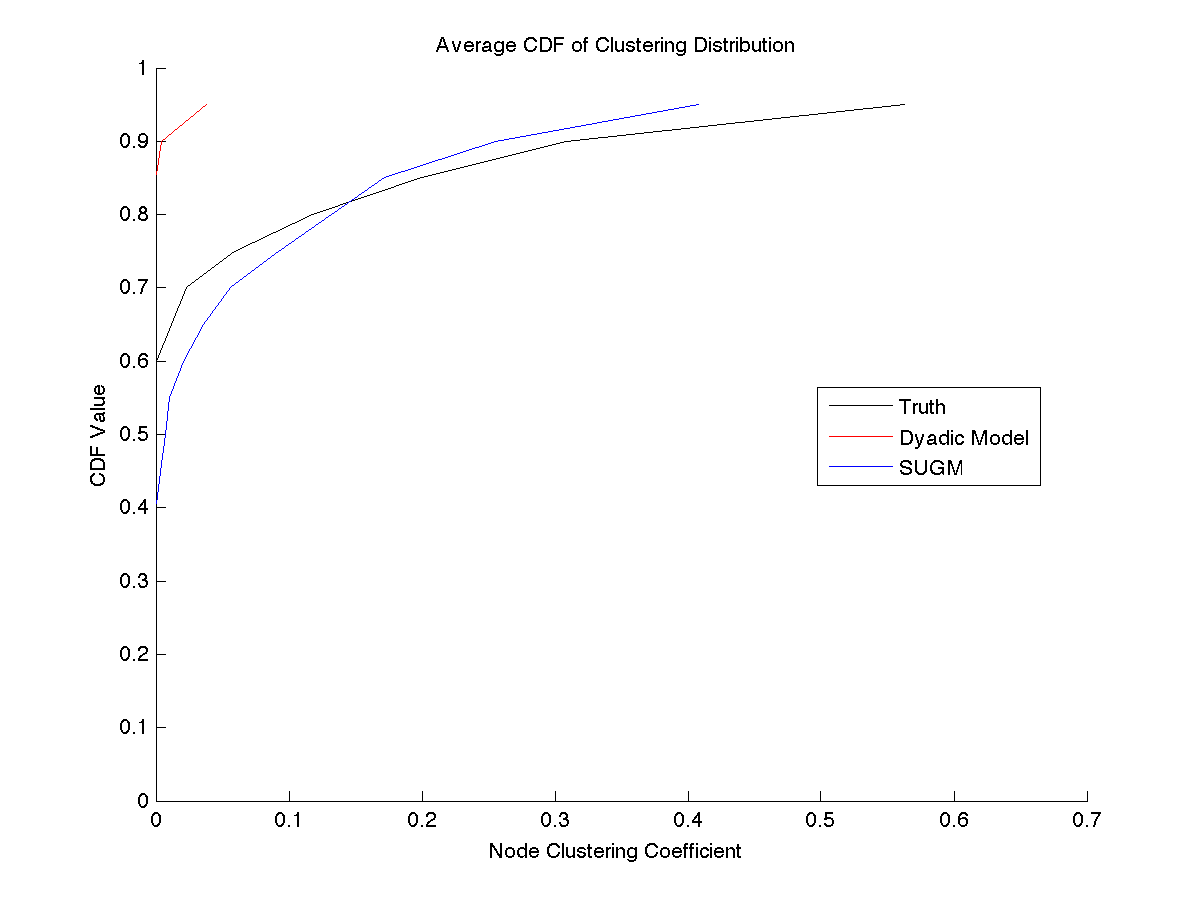}
}
\caption{\label{fig-Compare} Distributions of degree and clustering coefficients - averaged across the 75 villages. The figure displays the CDFs from the data (grey), the isolates-links-triangles SUGM (blue), and the link-based model (red).}
\end{figure}

The results of the analysis in this section are not sensitive to the covariates included. That is, it is not simply that
the SUGM allows for more parameters that enable it to better match the data.  It is that it includes richer network structures.
In Appendix \ref{simsappend}, we enrich the links-based model to include polynomials of a large set of demographic covariates including geographic distance, caste composition, quality of access to electricity, quality of latrines in the household, number of beds, number of rooms, etc.  We show that the links-based model, even aided by a considerable amount of data and more degrees of freedom, cannot replicate structural features of the network that are captured by very simple SUGMs that rely on minimal amounts of covariate data.

It is perhaps not surprising that SUGMs do a much better job at recreating network structures that standard link-based models, but nonetheless it is important.
 Moreover, the fact that the SUGMs do a better job than a link-based model of recreating not only local clustering and triangle patterns but also many other features of the real networks that it is not based upon suggests that there is substantial value added of modeling the formation of triangles and isolates.

Finally, knowing that our model is better able to capture the realistic correlational features of links within real world networks should make us more confident in trusting the results of the empirical application in Section \ref{crosscaste}. When we look at links across social boundaries, we can be comfortable that to first order, thinking about a SUGM with links and triangles across and within caste groups is a reasonable cut of the data.

\section{Conclusion}\label{conclusion}

We presented two new classes of models of network formation, SERGMs and SUGMs, based on the
idea that network formation is driven by the properties of a network and/or by the formation of various subgraphs.
This turns the focus away from the network as the unit of analysis, and instead focuses on its properties, subgraphs and statistics.
This perspective allows us to develop direct estimation techniques and to derive consistency and asymptotic distributions of parameter estimates.  Given the growing literature on estimation of network formation, such results are essential. Our framework provides a set of sufficient conditions that a user could check to see if their preferred microfoundation is likely to satisfy and we give a number of such examples.

Models of network formation, such as those described here, can be useful beyond simply studying
network patterns, such as in situations where network features are drivers of economic behaviors.
 For example, farmers may be significantly more likely to learn to use a new farming technology in one network than another, and so influencing
  the network of communication among farmers could be useful.
SERGMs and SUGMs, can help us understand the drivers of network formation and thus which sorts of interventions might lead to improvements in network features of interest.
Moreover, it is apparent that accounting for endogeneity of networks is important in studying peer influence (e.g., \cite{aral2009,goldsmith-pinkham2013,jackson2013,lindquist2013}) and so having practical models
of network formation is useful beyond direct estimation.

It is important to emphasize that our models can easily simulate networks.  This can be difficult (impossible) for standard ERGMs.  Here, our work suggests several avenues.  First, in the case of sparse networks, the SUGMs are perfectly suited for easily generating networks:  the model directly translates into an algorithm for generating networks by generating various subnetworks.
Second, in the case of SERGMs, the models are well-adapted to generating statistics of networks, even though in some cases it might be difficult to generate the networks themselves.   For example, if some profile of statistics $S$ is generated, then randomly picking a network $g$ that exhibits statistics exactly $S$ can be a hard problem.
This suggests a third avenue in line with our interpretation of SUGMs.    Instead of viewing $S$ as the realized statistic of the network, nature forms a network by forming subnetworks, even when they are dense.  So, the $S$ profile of generated subnetworks is picked by nature based on the SERGM.  This will generate some incidental subnetworks, and so a different observed $S'$ from $S$, but is still a perfectly well-defined model for generative purposes.\footnote{In fact, by simulating such processes one can estimate the relationship between $S'$ and $S$ which can then be used for estimation purposes when networks are not sparse, similar to the algorithm we provide for SUGMs.  We leave the general treatment of such algorithms and estimations for future research.}

Also, given our results, a researcher can use variations of standard approaches of model selection to deduce which statistics to incorporate in a SERGM or SUGM.
Consider the following.  Suppose that a researcher is interested in developing a model that captures the density, homophily, and clustering in an observed network.
An objective function can be built where a model's score is based on the total difference between its predictions of the relevant statistics (under best-fit parameters) and the observed statistics; and, as is commonly done, a penalty can be included for the number of parameters in the model.  Then one can examine SUGMs that incorporate various subsets of characteristic-based links, triangles, larger cliques, isolated nodes, and so forth, and find which model minimizes the objective function and is thereby selected. A similar logic could be used to construct goodness-of-fit tests for the model.

Finally, we note that the approach we have taken can be further extended.  In fact, once one adopts a SERGM formulation, many other sorts of applications beyond networks, such as matching problems, partitioning problems, club membership and others can also be incorporated.

\bibliographystyle{ecta}
\bibliography{metrics,networks}

\appendix

\section{Proofs}\label{proofs}

\noindent
\begin{proof}[{\bf Proof of Theorem \ref{prop-consistERGMs}}]

Consider a sequence of count statistics $S^n=(S^n_1, \ldots, S^n_k)$ whose $\ell$-th entry takes on nonnegative integer values with some maximum value $\overline{S}^n_\ell\rightarrow \infty$, and the count
SERGMs specified with
$K^n(s)= \prod_\ell \binom{\overline{S}^n_\ell}{s_\ell}$.

\begin{equation}
\mathrm{P}_{\beta}\left(S^n=s\right)=\frac{\prod_\ell \binom{\overline{S}^n_\ell}{s_\ell}\exp\left(\beta^n \cdot s \right)}{\sum_{s'\in A^n} \prod_\ell \binom{\overline{S}^n_\ell}{s'_\ell}\exp\left(\beta^n \cdot s'\right)}.
\label{eq:scs}
\end{equation}

We rewrite (\ref{eq:scs}) as
\begin{equation}
\mathrm{P}_{\beta}\left(S^n=s\right)=\frac{\prod_\ell \left[\binom{\overline{S}^n_\ell}{s_\ell}\exp\left(\beta^n_\ell  s_\ell \right)\right]}{\sum_{s'\in A^n} \prod_\ell \left[\binom{\overline{S}^n_\ell}{s'_\ell}\exp\left(\beta^n_\ell  s'_\ell \right)\right]}.
\label{eq:new0}
\end{equation}

If we consider the distribution conditional on $S^n$ lying in
\[
B^n=
\prod_\ell
\{\lfloor \E_{\beta^n_\ell}[S^n_\ell](1 - \varepsilon) \rfloor,\ldots,  \lceil \E_{\beta^n_\ell}[S^n_\ell](1 + \varepsilon) \rceil \}\subset A^n
\]
(for large enough $n$ under the non-conflicted condition), then
we can write
\[
\mathrm{P}_{\beta}\left(S^n=s| s\in B^n\right)
=\frac{\prod_\ell \left[\binom{\overline{S}^n_\ell}{s_\ell}\exp\left(\beta^n_\ell  s_\ell \right)\right]}{\sum_{s'\in B^n_\ell} \prod_\ell \left[ \binom{\overline{S}^n_\ell}{s'_\ell}\exp\left(\beta^n_\ell  s'_\ell \right)\right]}
=\frac{\prod_\ell \left[\binom{\overline{S}^n_\ell}{s_\ell}\exp\left(\beta^n_\ell  s_\ell \right)\right]}{\prod_\ell \left[\sum_{s'_\ell\in B^n_\ell}  \binom{\overline{S}^n_\ell}{s'_\ell}\exp\left(\beta^n_\ell  s'_\ell \right)\right]},
\]
or
\begin{equation}
\mathrm{P}_{\beta}\left(S^n=s| s\in B^n\right)=\prod_\ell \frac{ \binom{\overline{S}^n_\ell}{s_\ell}\exp\left(\beta^n_\ell  s_\ell \right)}{ \sum_{s'_\ell\in B^n_\ell}  \binom{\overline{S}^n_\ell}{s'_\ell}\exp\left(\beta^n_\ell  s'_\ell \right)}.
\label{eq:new1}
\end{equation}

Next, we consider binomial distributions which we will relate back to the above expressions for the SERGM.

For each $\ell$, consider a binomial distribution that has $p^n_\ell= \frac{\E_{\beta^n_\ell}[S^n_\ell ]}{\overline{S}^n_\ell}$, with range from $0$ to $\overline{S}^n_\ell$.
Taking a product of independent binomial distributions
\begin{equation}
\mathrm{P}^{Bin} \left(\widetilde{S}^n=s\right)= \prod_\ell \frac{ \binom{\overline{S}^n_\ell}{s_\ell}\exp\left({\beta}^n_\ell  s_\ell \right)}{ \sum_{s'_\ell\in [0,\overline{S}^n_\ell]}  \binom{\overline{S}^n_\ell}{s'_\ell}\exp\left({\beta}^n_\ell  s'_\ell \right)} =  \frac{ \prod_\ell \binom{\overline{S}^n_\ell}{s_\ell}\exp\left({\beta}^n_\ell  s_\ell \right)}{ \sum_{s'_\ell\in [0,\overline{S}^n_\ell]} \left[ \prod_\ell \binom{\overline{S}^n_\ell}{s'_\ell}\exp\left({\beta}^n_\ell  s'_\ell \right)\right]}.
\label{eq:new2}
\end{equation}
The corresponding probability that $\widetilde{S}^n=s$ given $ s\in B^n$ can be written as
\begin{equation}
\mathrm{P}^{Bin} \left(\widetilde{S}^n=s| s\in B^n\right)=\prod_\ell \frac{ \binom{\overline{S}^n_\ell}{s_\ell}\exp\left(\beta^n_\ell  s_\ell \right)}{ \sum_{s'_\ell\in B^n_\ell}  \binom{\overline{S}^n_\ell}{s'_\ell}\exp\left(\beta^n_\ell  s'_\ell \right)}.
\label{eq:new3}
\end{equation}

For a binomial distribution, the probability that $\widetilde{S}^n_\ell \in B^n_\ell \rightarrow 1$.\footnote{This comes from Chernoff's inequality since the probability that the ratio of the total number of successes to the expected number of successes is within an $\epsilon$ window of 1 as $n \rightarrow \infty$.}
Thus, under independent binomial distributions, the probability that  $\widetilde{S}^n \in B^n \rightarrow 1$,
and so
it follows from (\ref{eq:new2}) and (\ref{eq:new3}) that
\begin{equation}
\frac{\mathrm{P}^{Bin} \left(\widetilde{S}^n=s| s\in B^n\right)}{\mathrm{P}^{Bin} \left(\widetilde{S}^n=s\right)}=
\frac{ \sum_{s'_\ell\in [0,\overline{S}^n_\ell]}  \binom{\overline{S}^n_\ell}{s'_\ell}\exp\left({\beta}^n_\ell  s'_\ell \right)}
{ \sum_{s'_\ell\in B^n_\ell}  \binom{\overline{S}^n_\ell}{s'_\ell}\exp\left(\beta^n_\ell  s'_\ell \right)}
\rightarrow 1,
\label{eq:new4}
\end{equation}
uniformly for $s\in B^n$.

Collecting from (\ref{eq:new0}) and (\ref{eq:new2}) it follows that for $s\in B^n$
\begin{equation}
\mathrm{P}_\beta \left(\widetilde{S}^n=s\right) \geq   \mathrm{P}^{Bin} \left(\widetilde{S}^n=s\right).
\label{eq:new51}
\end{equation}
Collecting from (\ref{eq:new1}) and (\ref{eq:new3}) it follows that for $s\in B^n$
\begin{equation}
\mathrm{P}^{Bin} \left(\widetilde{S}^n=s| s\in B^n\right) = \mathrm{P}_\beta \left(\widetilde{S}^n=s| s\in B^n\right) .
\label{eq:new52}
\end{equation}
Then (\ref{eq:new51}) and (\ref{eq:new52}) and the fact that $\mathrm{P}_\beta \left(\widetilde{S}^n=s| s\in B^n\right) \geq \mathrm{P}_\beta \left(\widetilde{S}^n=s\right)$, together with (\ref{eq:new4}), imply that
\begin{equation}
\frac{\mathrm{P}_\beta \left(\widetilde{S}^n=s\right)}{\mathrm{P}^{Bin} \left(\widetilde{S}^n=s\right)}
\rightarrow 1,
\label{eq:new6}
\end{equation}
uniformly for $s\in B^n$.

The remainder of the claimed results then follows easily from standard properties of the binomial distribution.

In particular, the variance terms are computed as follows. First consider a binomial ${\rm Bin}(p^n_\ell;n)$ with $p^n_\ell n \rightarrow \infty$. By the Lindeberg-Feller Central Limit Theorem
$
\sqrt{n}\left(\widehat{p}^n_\ell-p^n_\ell\right)\rightsquigarrow\mathcal{N}\left(0,p^n_\ell\left(1-p^n_\ell\right)\right).
$
Letting
$
\beta^n_\ell=g\left(p^n_\ell\right)=\log\frac{p^n_\ell}{1-p^n_\ell},
$
note
\[
g'(p)=\frac{1}{p}+\frac{1}{1-p}=\frac{1}{p(1-p)}.
\]
Finally, by the delta method,
\[
\sqrt{n}\left(\widehat{\beta^n_\ell}-\beta^n_\ell\right)\rightsquigarrow\mathcal{N}\left(0,p^n_\ell\left(1-p^n_\ell\right)\left[g'(p^n_\ell)\right]^{2}\right)
\]
and therefore observing that $p^n_\ell=\frac{\exp\beta^n_\ell}{1+\exp\beta^n_\ell}$, $p^n_\ell(1-p^n_\ell)\left[\frac{1}{p^n_\ell(1-p^n_\ell)}\right]^{2}=\frac{1}{p^n_\ell(1-p^n_\ell)}$
and substituting for $\beta^n_\ell$, we have
\[
\sqrt{n}\left(\widehat{\beta^n_\ell}-\beta^n_\ell\right)\rightsquigarrow\mathcal{N}\left(0,\frac{1}{\frac{\exp\beta^n_\ell}{1+\exp\beta^n_\ell}
\left(1-\frac{\exp\beta^n_\ell}{1+\exp\beta^n_\ell}\right)}\right).
\]
This argument replacing $n$ with $\bar{S}_{\ell}^{n}$
shows the result.\end{proof}

\bigskip

\begin{proof}[{\bf Proof of Theorem \ref{sparse}}]
We provide the proof without covariates to save on notation, but it extends directly. We begin the proof by showing the following.  For any $\ell$, the fraction of counts of subnetworks $\ell$ generated incidentally by some other subnetworks goes to 0.
That is, consider some $\ell$ and  $g_\ell \in G^n_\ell$ on $m_\ell$ nodes.  Let
$$\widetilde{p}^n_\ell =\frac{\E(\widetilde{S}^n_\ell)}{y_\ell^n \binom{n}{m_\ell}}. $$
This is no more than $p^n_\ell$, as the denominator includes all possible subgraphs of size $\ell$ (where $y_\ell^n$ is the number of subgraphs of type $\ell$ that can be formed on any given $m_\ell$ nodes).\footnote{Here we provide the proof that applies without subgraphs being characteristic dependent.  The extension to characteristic dependent subgraphs is straightforward simply by adjusting all numbers to reflect possible
networks with given node characteristics as dependent on $n$ and the sets of nodes that have particular characteristics as a function of $n$, but it is notationally much more intensive.}
Let us consider the probability $z^n_\ell$ that $g_\ell$ is incidentally generated by other subnetworks.

We show that $z^n_\ell/ \widetilde{p}^n_\ell \rightarrow 0$, which implies that $z^n_\ell/ {p}^n_\ell \rightarrow 0$.

Consider $g_\ell \in G^n_\ell$ and an incidentally generating subclass  $(\ell_j, h_j)_{j\in J}$.

We show that the probability $z^n_\ell(J)$ that it is generated by this subclass goes to zero relative to $\widetilde{p}^n_\ell$,
so that $z^n_\ell(J)/ \widetilde{p}^n_\ell \rightarrow 0$ for each $J$, and since there are at most
$M_\ell\leq k^{m_\ell}$ such generating classes, this implies that $z^n_\ell/ \widetilde{p}^n_\ell \rightarrow 0$.

For a subnetwork in $G^n_{\ell_j}$, the probability of getting {\sl at least one} such network
that has the $h_j$ nodes out of the $m_\ell$ in $g_\ell$ is no more than\footnote{This is a loose upper bound as it simply
adds the probability that each possible one forms - but becomes more accurate as the probability of any one occurring vanishes.}
$$y_{\ell_j}^n \binom{n}{m_{\ell_j}-h_j} \widetilde{p}^n_{\ell_j} \leq y_{\ell_j}^n n^{m_{\ell_j}-h_j}\widetilde{p}^n_{\ell_j}.$$
Thus,
$$\frac{z^n_\ell(J)}{ \widetilde{p}^n_\ell} \leq  \frac{\Pi_{j\in J} n^{m_{\ell_j}-h_j}y_j^n \widetilde{p}^n_{\ell_j}}{ \widetilde{p}^n_\ell}.$$
Therefore
$$\frac{z^n_\ell(J)}{ \widetilde{p}^n_\ell} \leq  \frac{\Pi_{j\in J}y_j^n n^{m_{\ell_j}}\widetilde{p}^n_{\ell_j}}{n^{\sum_j h_j} \widetilde{p}^n_\ell}.$$

Recall that $M_J=\sum_{j\in J} h_j -m_\ell$ and that $M\geq 1$ (since $|J|\geq 2$ and each set of $h_j$ intersects with at least one other set of $h_{j'}$ for some $j'\neq j$).

Therefore
$$\frac{z^n_\ell(J)}{ \widetilde{p}^n_\ell}\leq  \frac{\Pi_{j\in J}y_{\ell_j}^n n^{m_{\ell_j}}\widetilde{p}^n_{\ell_j}}{n^{M_J} n^{m_\ell} \widetilde{p}^n_\ell}.$$
The numerator is of the order $\Pi_{j\in J} \E(S^n_j)$ while the denominator is of the order $n^{M_J} \E(S^n_\ell)$.

Under the sparseness condition,
$$\frac{\Pi_{j\in J} \E(S^n_{\ell_j})}{ n^{M_J} \E(S^n_\ell)}\rightarrow 0, $$
and so we have verified the claim.

To finish the ratio consistency proof, note that the claim then implies that $\frac{\widehat{S}^n_\ell(g)}{{S}^n_\ell} \rightarrow 1$.
Thus, dividing top and bottom by $\overline{S}^n_\ell(g)$, it follows
that $\frac{\widehat{p}^n_\ell(g)}{{S}^n_\ell/\overline{S}^n_\ell(g)}\rightarrow 1$.
Given the growing condition and properties of the binomial distribution, it also follows that
$\frac{{S}^n_\ell/\overline{S}^n_\ell(g)}{{p}^n_\ell}\rightarrow 1$, and so $\frac{\widehat{p}^n_\ell(g)}{p^n_\ell}\rightarrow 1$.

Next, note that the above
also implies that the distribution
$D_{n}^{1/2}\left(\widehat{p}_{1}^n(g),...,\widehat{p}_{k}^n(g)\right)'$ converges to the distribution of
$D_{n}^{1/2}\left(\widetilde{p}_{1}^n(g),...,\widetilde{p}_{k}^n(g)\right)'$, where $\widetilde{p}_{\ell}^n(g) = {S}^n_\ell/\overline{S}^n_\ell(g)$.

The asymptotic normality of the (joint) distribution then follows from the usual Linderberg-Feller central limit theorem applied to triangular arrays of binomial random variables.  This applies under the growing conditions of the theorem.\footnote{\label{LF_clt}
Let $X_{1T},...,X_{TT}$ be a triangular array of ${\rm Ber}(p_T)$ random variables and $p_{T}T\rightarrow \infty$ (as under the growing condition).
To apply the Lindeberg-Feller central limit theorem for triangular arrays, we check the Lindeberg condition: for any $\epsilon>0$,
\[
\frac{1}{Tp_T(1-p_T)}\sum_{t\leq T} \E \left[ X_t^2 {\bf 1}\left\{|X_t| \geq \epsilon \sqrt{Tp_T(1-p_T)}\right\} \right] = o(1).
\]
The condition is implied by the fact that $Tp_T\rightarrow \infty$, and from the sparsity conditions which imply that $p_T$ is bounded away from 1.}
\end{proof}

\

\begin{proof}[{\bf Proof of Theorem \ref{potential}}]

Let $\mathcal{G}^n$ denote the set of all possible subgraphs in the model:
\begin{equation}
\mathcal{G}^n = \cup_{\ell}G^n_\ell.
\end{equation}

Letting $G$ be the set of subgraphs that are truly formed, we can write
$$
\Pr(G) = \prod_{g_\ell \in G} \frac{\exp(\theta_{\ell})}{ \exp(\theta_{\ell}) + 1}
\prod_{g_\ell\notin G} \frac{1}{ \exp(\theta_{\ell}) + 1}.
$$
This can be rewritten as
$$
\Pr(G) =  \frac{\exp\left(\sum_{ g_\ell\in G}\theta_{\ell}\right)}{\prod_{g_\ell\in \mathcal{G}^n} \left(\exp(\theta_{\ell}) + 1\right)}.
$$
Note that
$$\prod_{g_\ell\in \mathcal{G}^n} \left(\exp(\theta_{\ell}) + 1\right)
=
\sum_{A\subset \mathcal{G}^n} \left(\exp(\sum_{\ell} |A\cap G^n_{\ell}|\theta_{\ell})\right).$$
Given that the number of $A$'s that have counts $s' \in \prod_\ell \{0,\ldots, \overline{S}_\ell^n\}$ is exactly
$$\prod_\ell \binom{\overline{S}_\ell^n}{s'_\ell},$$
it follows that
$$\prod_{ g_\ell\in \mathcal{G}^n} \left(\exp(\theta_{\ell}) + 1\right)
=
\sum_{s' \in \prod_\ell \{0,\ldots, \overline{S}_\ell^n\}}\left(\prod_\ell K^n_\ell(s'_\ell)\right) \exp(\sum_{\ell} s'_\ell \theta_{\ell}),$$
where $K^n_\ell(s'_\ell)=\binom{\overline{S}_\ell^n}{s'_\ell}$.
This means that we can write
\begin{equation}
\label{potent9}
{\rm P}(G) =\frac{\exp\left(\sum_{\ell} |G^n_{\ell}\cap G |\cdot\theta_{\ell}\right)}{\sum_{s' \in \prod_\ell \{0,\ldots, \overline{S}_\ell^n\}}\left(\prod_\ell K^n_\ell(s'_\ell)\right)
\exp\left(\sum_\ell s'_\ell \theta_\ell\right)
}.
\end{equation}
Note that there are $K_\ell^n(\widetilde{S})$ different collections of truly generated subgraphs $G$ that have the same value $\widetilde{S}$ and that each is
equally likely.
Thus
\begin{equation}
\label{potent}
{\rm P}(\widetilde{S}) =
\frac{K_\ell^n(\widetilde{S})\exp\left(\sum_{\ell} \widetilde{S}_\ell\cdot\theta_{\ell}\right)}{\sum_{s' \in \prod_\ell \{0,\ldots, \overline{S}_\ell^n\}}\left(\prod_\ell K^n_\ell(s'_\ell)\right)
\exp\left(\sum_\ell s'_\ell \theta_\ell\right)
},
\end{equation}
which is the same as
\eqref{potent2},
which completes the proof.
\end{proof}

\

\begin{proof}[{\bf Proof of Lemma \ref{lem:geomodel}}]
Having two randomly picked nodes bump into each other within a community, there is a $f^2+(1-f)^2$ probability of the nodes being of the same type, and a $1-(f^2+(1-f)^2)$ probability of them being of different types.\footnote{ To keep things simple, we consider equal-sized groups, but the argument extends with some adjustments to asymmetric sizes.}     Thus, the relative meeting frequency of different type links compared same type links is
\[
\frac{\pi_L(diff)}{\pi_L(same)}=\frac{1-(f^2+(1-f)^2)}{f^2+(1-f)^2}.
\]
For triangles, picking three individuals out of the community at any point in time would lead to a $f^3+(1-f)^3$ probability that all three are of the same type, and $1-(f^2+(1-f)^2)$ of them being of mixed types,
and so
\[
\frac{\pi_T(diff)}{\pi_T(same)}=\frac{1-(f^3+(1-f)^3)}{f^3+(1-f)^3}.
\]
It follows directly that for $f\in (0,1)$:
\begin{equation}
\label{ratios}
\frac{\pi_T(same)}{\pi_T(diff)}<\frac{\pi_L(same)}{\pi_L(diff)}.
\end{equation}
So different type triangles are more likely to have opportunities to form under this random mixing model than different type links.
In particular, note that
$$\frac{p_T(diff)}{p_T(same)} < \frac{p_L(diff)}{p_L(same)} {\rm \ if \  and \  only\  if\ }\left(\frac{P_T(diff)}{P_T(same)}\frac{\pi_T(same)}{\pi_T(diff)}\right)^{1/3} < \left(\frac{P_L(diff)}{P_L(same)}\frac{\pi_L(same)}{\pi_L(diff)}\right)^{1/2}.$$

In summary, given (\ref{ratios}),
a sufficient condition for
$\frac{p_T(diff)}{p_T(same)} < \frac{p_L(diff)}{p_L(same)}$ is that
$$(P_T(diff)/P_T(same))<(P_L(diff)/P_L(same))^{3/2}$$
which completes the argument.
\end{proof}

\

\section{A Useful Lemma on SERGM Statistic Domains}\label{A-sum-sect}

Although polynomial, the denominators of a SERGM can still involve large numbers of calculations.
There are substantial simplifications that can be made.
For example, we can estimate the denominator of (\ref{sergm-LT}) by
summing across some subset of $s'$ that has high probability rather than summing over the full set, as although $n^6$ is polynomial it still can be a large sum to do exhaustively as $n$ grows.
In particular, suppose that for some parameter $\beta$, the probability that the observed statistic ends up
taking a value in some set $A$ is at least $1-\varepsilon$: $\Pr_\beta(s\in A)\geq 1-\varepsilon$.
Then
by setting
\[
\overline{\mathrm{P}}_{\beta}\left((S_I,S_L,S_T)=s\right) =
\frac{N_{S}(s)\exp\left(\beta_I s_I+\beta_L s_L+\beta_T s_T\right)}{\sum_{s'\in A} N_{S}(s')\exp\left(\beta_I s_I'+\beta_L s_L'+\beta_T s_T'\right)}.
\]
it follows that for any $s\in A$
\[
\frac{1}{1-\varepsilon}\geq \frac{\overline{\mathrm{P}}_{\beta}\left((S_I,S_L,S_T)=s\right)}{\Pr_{\beta}\left((S_I,S_L,S_T)=s\right)} \geq 1.
\]
Thus, we can work with $\overline{\mathrm{P}}_{\beta}\left((S_I,S_L,S_T)=s\right)$ which only requires computations over $s'\in A$ in its denominator.

While we have to worry about determining $A$, which depends on $\beta$ which is presumed to be unknown to the researcher, if the probability of various statistics concentrates in a small neighborhood around the observed statistics with high probability, the above approximation becomes quite useful.   By observing $s$, and then choosing $A$ to be a large enough neighborhood around the observed $s$, one can be sure that under the true (unobserved) $\beta$,  $\Pr_\beta(A)\geq 1-\varepsilon$.   In particular, it is easy to choose a small set $A$ based on the observed $s$ over which to sum the denominator {\sl without knowing $\beta$}, and which with arbitrarily high probability will give an arbitrarily accurate estimate for large enough $n$.
A general version of such a result is Lemma \ref{A-sum}:

\begin{lem}
\label{A-sum}
Consider a SERGM with associated $\Pr_\beta(s)$ as described in (\ref{eq: sergm}) and consider any $\varepsilon>0$.
Suppose that for each $\beta$ in some set $B$ there exists $A_{\beta}$ such that $\Pr_\beta(A_{\beta})\geq 1-\varepsilon$.
For each $s$ let $A_s =\cup_{\beta\in B: s\in A_\beta} A_\beta$.
Letting
\[
\overline{\mathrm{P}}_{\beta}\left(s\right) =
\frac{K_S(s)\exp\left(\beta s\right)}{\sum_{s'\in A_s} K_S(s')\exp\left(\beta s'\right)},
\]
it follows that for any $\beta\in B$ and $s\in A_\beta$:
\[
\frac{1}{1-\varepsilon}\geq \frac{\overline{\mathrm{P}}_{\beta}\left(s\right)}{\Pr_{\beta}\left(s\right)} \geq 1.
\]
\end{lem}

\begin{proof}[{\bf Proof of Lemma \ref{A-sum}}]

Let
\[
\widehat{\mathrm{P}}_{\beta}\left(s\right) =
\frac{K_S(s)\exp\left(\beta s\right)}{\sum_{s'\in A_\beta} K_S(s')\exp\left(\beta s'\right)}.
\]
Since $\Pr_\beta(A_{\beta})\geq 1-\varepsilon$,
it follows that
\[
\frac{1}{1-\varepsilon}\geq
\frac{\sum_{s'} K_S(s')\exp\left(\beta s'\right)}{\sum_{s'\in A_\beta} K_S(s')\exp\left(\beta s'\right)} \geq 1.
\]
this implies that for any $\beta$ and $s\in A_\beta$:
\[
\frac{1}{1-\varepsilon}\geq \frac{\widehat{\mathrm{P}}_{\beta}\left(s\right)}{\Pr_{\beta}\left(s\right)} \geq 1.
\]
Note also that for any $\beta$ and $s\in A_\beta$
\[
\widehat{\mathrm{P}}_{\beta}\left(s\right)\geq \overline{\mathrm{P}}_{\beta}\left(s\right) \geq \Pr_{\beta}\left(s\right).
\]
The claimed result follows from the last two sets of inequalities.
\end{proof}

\newpage

\setcounter{page}{1}

\setcounter{table}{0}
\renewcommand{\thetable}{C.\arabic{table}}

\setcounter{figure}{0}
\renewcommand{\thefigure}{C.\arabic{figure}}

\section{Online Appendix: Extension to Continuous and Interdependent Covariates}\label{ext}

For simplicity, we have focused on models in which covariates are captured by indexing subgraphs by covariates.
This encompasses covariates that take on a finite set of values or are approximated by a finite set of values, and is a flexible approach,
although it may not work as well with fully continuous data that take on a wide range of values. Such continuous covariates can also easily be handled, as our models and results have natural extensions. 

We discuss the SUGM extension.
Let node $i$ be associated with a covariate vector $X_{i}$ that lies in a compact subset of $\mathbb{R}^{d}$.
Let the probability that a given subnetwork $g_\ell \in G_\ell$ forms be a function
$p_\ell^n(X_\ell;\gamma)$ of the vector of node covariates, where $\gamma$ is some vector of parameters.

Estimating the parameters $\gamma$ depends on the functional form of $p_\ell^n(x_\ell;\gamma)$.
It could take many forms, such as a linear probability model, a logistic form, etc.  Consistency and asymptotic normality of the estimators depend on the rate at which $\gamma$ tends to extremes -- thereby affecting the probabilities of various subgraphs and their dependence on covariate values.  We provide some sufficient conditions for consistency and asymptotic normality of the estimators below.

We consider an environment in which nodes draw covariates that can be continuous and even interdependent. Then, based on their characteristics, they form a graph via the SUGM process.  We are interested in estimating both probability functions as well as possible parameters which may correspond to random utility foundations (e.g., coefficients in a logistic regression term).

\

\subsection*{Environment}

Every node $i\in\left\{ 1,...,n\right\} $ draws a $d$-dimensional
covariate vector $x_{i}^{n}\in\mathcal{X}$. For simplicity we let
$\mathcal{X}=\prod_{k=1}^{d}\left[x_{L,k},x_{H,k}\right]$ be a $d$-dimensional
product of intervals of $\mathbb{R}$.\footnote{We will allow these covariates to be interdependent. The substantive assumption we need to make is that the sequences of design matrices and have full rank.} Letting $x_\ell$ denote the $d \times n$ matrix of data, we assume $x_\ell x_\ell'$ has full rank along the sequence. For expositional simplicity in our proofs we considering a sequence of fixed-regressors, $x_{\ell,n}$ where $n$ indexes the sequence. Clearly stochastic regressors can be accomodated.

\subsubsection*{Example \ref{ext}.1} Let $x^n_i = (1,u^n_i)$, where $u^n_i \in [0,1]$ such that the design matrix carries full rank.  In the simulation exercise corresponding to this example, we will draw them as independent $U[0,1]$ random variables.

\

\subsection*{SUGM Formation}

Given characteristics, the $n$ nodes engage in a SUGM graph formation
process. The realized data sequence consists of a triangular array
of random graphs and covariate vectors drawn from a random field $\left\{ \left(g^{n},\left(x_{1}^{n},...,x_{n}^{n}\right)\right):\ n\in\mathbb{N}\right\} $.
The researcher observes this for a given $n$ and a given realization.

Specifically, consider a set of nicely ordered statistics $\left(S_{\ell}^{n}\right)$
again with each statistic counting subgraphs $H_{\ell}$ with $m_{\ell}$
nodes, where the statistics $S_{\ell}$ do not condition on covariates.
We are therefore counting, for instance, 4-cliques, triangles (not
in 4-cliques), and unsupported links.

A group of size $m_{\ell}$ forms with a probability $p_{\ell}^{n}(x_{\ell,j};\gamma_{\ell})$
which depends on some function of the $m_{\ell}$ individuals' characteristics
and a parameter $\gamma_{\ell}$, whose value in theory may depend
on $n$.%
\footnote{It is easy to modify this such that $f_{\ell}=f_{\ell,i}$ so that
every node makes its own decision to be in the group or not, and its
covariates are not treated symmetrically with the other $m_{\ell}$
nodes.%
}

To make things concrete, examples of $p_{\ell}^{n}(x_{\ell};\gamma)$
include:
\begin{enumerate}
\item a linear probability model with uniform link function $p_{\ell}^{n}(x_{\ell,j};\gamma_{\ell})=\gamma_{\ell,n}'x_{j,\ell},$
\item a logistic regression model $p_{\ell}^{n}(x_{\ell,j};\gamma_{\ell})=\frac{\exp(\gamma_{\ell,n}'x_{j,\ell})}{1+\exp(\gamma_{\ell,n}x_{j,\ell})}$,
\end{enumerate}
for $j\in\{1,...,\overline{S}_{\ell}(g)\}$. It should be clear that
there are any number of examples here that could be used and the choice
is up to the modeler's discretion as to what best describes the nature
of the problem at hand.

A truly generated object is a subgraph on $m_{\ell}$ nodes that is
generated in the $\ell$th phase independently with probability $p_{\ell}^{n}(x_{\ell,j};\gamma_{\ell})$.
Incidental generation may occur and the union is the graph $g^{n}$.

The group-level characteristic, $x_{\ell}$, is of course a function
of individual level characteristics: $x_{\ell,i_{1},...,i_{m_{\ell}}}=f_{\ell}\left(x_{i_{1}},....,x_{i_{\ell}}\right)$.
For example, $f_{\ell}\left(x_{i},x_{j}\right)=\left|x_{i}-x_{j}\right|$.

\subsubsection*{Example \ref{ext}.1}[Continued] The sequence of graphs $g^n$ are triangles and links-based.  A triangle forms with probability defined by log-odds
\[
\log \frac{p^n_T(x_T;\gamma_T)}{1-p^n_T(x_T;\gamma_T)} = \gamma_{0,n,T}'x_T =(\alpha_{0,n,T},\beta_{0,T} ) x_T
\]
where $x_T = (1, u_T)$ and $u_T = (|u_i - u_j| + |u_j - u_k| + |u_k - u_i| ) / 3$.

A link forms with probability
\[
\log \frac{p^n_L(x_L;\gamma_L)}{1-p^n_L(x_L;\gamma_L)} = \gamma_{0,n,L}'x_L = (\alpha_{0,n,L},\beta_{0,L} ) x_L
\]
where $x_L = (1, u_L)$ and $u_L = |u_i - u_j|  /2$.

Pairs and triples that are further in covariate space are less likely to link.

\

\subsection*{Estimation}

The above defines a well-defined network-generation process.  As before, we need a relative sparsity condition to hold so that when we count a structure, with probability approaching one it was not incidentally generated.  Here we provide a sufficient condition for relative
sparsity hold as the continuous covariates vary. The condition is
that given $m_{\ell}$ nodes, no matter what the value of each covariate
is among these nodes, the probability of forming the subgraph isomorphic
to $H_{\ell}$ shrinks at the same as $n$ grows to infinity. This
will ensure the relative rate of incidentally generated objects is
unaffected by the particular values of the covariates.%
\footnote{Such an assumption excludes the possibility that individuals who are
close in wealth are more likely to form pairs than triads for wealth
levels below some threshold but beyond this threshold it is when individuals
are far from others in wealth that pairs are more likely to form than
triads. (More specifically, in this example a wealth covariate should
not be used, but rather, a wealth covariate with an indicator for
whether individuals are below or above the threshold must be used.)%
}

\

\begin{lem}Given a growing sequence of graphs with associated covariates and
covariate space $\mathcal{X}$, and probability functions $p_{\ell}^{n}(x_{\ell};\gamma_{\ell})$
smooth in both arguments,
\[
\min_{x_{1},...,x_{m_{\ell}}\in\mathcal{X}^{m_{\ell}}}p_{\ell}^{n}\left(x_{\ell};\gamma_{\ell}\right)=O\left(\max_{x_{1},...,x_{m_{\ell}}\in\mathcal{X}^{m_{\ell}}}p_{\ell}^{n}\left(x_{\ell};\gamma_{\ell}\right)\right).
\]
If relative sparsity
is satisfied at $x_{i}=1$ for all $i$, then relative sparsity is
satisfied for any sequence of covariates.\end{lem}

\begin{proof}We can always replace incidental generation probabilities
with their maximal values over the covariates, the truly generating
probability with its minimal probability. These are all of the same
order as when evaluated with $x_{i}=1$ by hypothesis.\end{proof}

\

\noindent Of course this isn't the only condition to maintain relative sparsity, but it may often be a natural condition to assume.

\

We now show properties of estimators from the two examples of $p_\ell^n(x_\ell; \gamma_\ell)$ we have discussed.

\subsubsection*{Linear Probability Model}

Consider the linear probability model discussed above:
$$ p^n_\ell (x_\ell ; \gamma_\ell)  = \sum_k \gamma^k_\ell x_{k,\ell}, \ k=1,...,d$$
where $\gamma_{0,n,\ell}^k \rightarrow 0$ as $n\rightarrow \infty$. It is straightforward to check that the following is true.

\begin{thm} Assume $\left\Vert \gamma_{0,n,\ell} \right\Vert_1 = \Theta(1/n^{m_\ell - h_\ell})$\footnote{$f_n \in \Theta(g_n)$ means $\exists k_1>0, \exists k_2>0, \exists n_0>0, \forall n > n_0$ such that $g_n k_1 < f_n < g_n k_2$.} with $0<h_\ell<m_\ell$ and the $h_\ell$ are such that relative sparsity condition is satisfied.  Then
$$\sqrt{n^{m_\ell+h_\ell}}\left( \widehat{\gamma} - \gamma_{0,n,\ell} \right) \wkto \mathcal{N}(0,V)$$ where $V = \plim \frac{1}{n^{m_\ell}}(x_\ell'x_\ell)^{-1} (\frac{n^{h_\ell}}{n^{m_\ell}}x_\ell' \epsilon_\ell \epsilon_\ell'x_\ell) \frac{1}{n^{m_\ell}}(x_\ell'x_\ell)^{-1}$.
\end{thm}

\noindent We omit the proof, which is entirely standard. We get super-consistent rates as the parameters are going to zero rapidly, but not too rapidly so that a central limit theorem still applies. Because relative sparsity applies, only a vanishing proportion of $\ell$-objects are incidentally generated.

\

\subsubsection*{Logistic Regression}

We turn to our main example where $p_{\ell}^{n}(x_{\ell,j};\gamma_{\ell})$ is given by a logistic link function. In all that follows $\gamma_{0,n}$ consists of elements that are either order constant or tending to $-\infty$. The rates will be set in the assumptions.

\begin{thm} Assume that $\left\Vert \gamma_{0,n} \right\Vert_1 \cdot \sup_{x \in \mathcal{X}}\left\Vert x \right\Vert_\infty \lesssim h_\ell \cdot \log n^{m_\ell}$ for $0\leq h_\ell < m_\ell$.  Additionally, assume that relative sparsity holds. Then
$$ J_n^{1/2} \left( \gammahat_\ell -  \gamma_{0,\ell,n} \right) \wkto \mathcal{N} \left(0, I_d\right).$$
\end{thm}

\begin{proof} Follows from Lemma \ref{lem:mismeasure}. The first hypothesis of the lemma is the same as that in Lemma \ref{lem:hjort-pollard} and is assumed here for each $\ell$. Additionally, assumption (2) of Lemma \ref{lem:mismeasure} follows from relative sparsity. Relative sparsity implies that the $h_\ell$ are ordered such that for every $\ell$ share of incidentally generated $\ell$-th objects goes to zero, corresponding to the number of incidentals being on the order of $O_p(z_{n,\ell} \cdot n^{m_\ell})$ in Lemma \ref{lem:mismeasure}.
\end{proof}

\

This means that the rate of convergence of the parameters governing the probability is given by $\sqrt{n^{m_\ell - k_\ell}}$ where $0 < h_\ell < m_\ell$ tunes the sparsity of the model.

\

\subsubsection*{Example \ref{ext}.1}[Continued] Consider $\alpha_{0,L}^n = \log (1/n^{0.7})$ and $\alpha_{0,T}^n = \log (1/n^{1.75})$,  $\beta_{0,L} = -2$ and $\beta_{0,T} = -3$.   Then triangles form at order $1/n^{1.75}$ and links at order $1/n^{0.7}$.  The theorem shows that all parameters have estimators that are consistent and, in the case of links, are asymptotically normally distributed at $\sqrt{n^{1.3}}$-rate and $\sqrt{n^{1.25}}$-rates (for links and triangles, respectively).

\

For some intuition as to why this works, first consider the case of a triangular array of $n$ i.i.d. Bernoulli random variables distributed with probability $p_n \downarrow 0$ at a rate $\Theta(1/n^h)$ for $0<h<1$.  Then the log odds is given by $\log\frac{p_n}{1-p_n}=\alpha_{n}$ where $\alpha_n = -h \log (C\cdot n)$ for some constant $C > 0$.  It is easy to show by the Lindeberg-Feller central limit theorem for triangular arrays that in this case  $$\sqrt{n}\left(\frac{\widehat{p}_{n}-p_{n}}{\sqrt{p_{n}}}\right)\rightsquigarrow\mathcal{N}\left(0,1\right)$$ provided $p_n n\rightarrow \infty$.
 This implies that  $\sqrt{np_{n}}\left(\widehat{\alpha}-\alpha_{n}\right)=\sqrt{n^{1-h}}\left(\widehat{\alpha}-\alpha_{n}\right)\rightsquigarrow\mathcal{N}\left(0,1\right)$.  This follows from observing that $\alpha_{n}$ will be consistent\footnote{
$\left|\widehat{\alpha}-\alpha\right| 
=  \left|\log\frac{\widehat{p}_{n}}{1-\widehat{p}_{n}}-\log\frac{p_{n}}{1-p_{n}}\right|\leq\left\{ \left(\frac{1-\bar{p}_{n}}{\bar{p}_{n}}\right)\left(\frac{1}{\left(1-\bar{p}_{n}\right)^{2}}\right)\right\} \left|\widehat{p}_{n}-p_{n}\right|
\lesssim_{p}  \frac{\left|\widehat{p}_{n}-p_{n}\right|}{\bar{p}_{n}}=O_{p}\left(\sqrt{\frac{1}{np_{n}}}\right)\rightarrow0.
$
}
and by the delta method
\[
\sqrt{\frac{n}{p_{n}}}\left(\widehat{\alpha}-\alpha\right)\rightsquigarrow\mathcal{N}\left(0,\left[\partial_{p}\left\{ \log\frac{p_{n}}{1-p_{n}}\right\} \right]^{2}\right)=\mathcal{N}\left(0,\frac{1}{(p_{n}(1-p_n))^{-2}}\right)
\]
which implies $
\sqrt{np_{n}}\left(\widehat{\alpha}-\alpha\right)\rightsquigarrow\mathcal{N}\left(0,1\right),$ noting that clearly the $(1-p_n)$ term is irrelevant in the rate normalization under the hypothesized asymptotic sequence.
\

Next we offer an intuition for why this works with a finite set of discrete covariates.  Let $\log\frac{p\left(x\right)}{1-p\left(x\right)}=\alpha_{n}+\beta x$
for $x$ in some finite discrete set.   It is clear that repeating the above argument delivers the same rate of convergence at every covariate value.

\

We now consider the general case. The data consists of a triangular array $\{(y_{i,n},x_{i,n}):\ n \in \mathbb{N} \}$ where $y_{i,n}$ is a binomial outcome governed by $p^n(x_{i,n}; \gamma_{0,n})$.  To conserve on notation let $q_{in}=p\left(x_{in}'\gamma_{0n}\right)$ and put $J_{n}=\sum_{i\leq n}q_{in}\left(1-q_{in}\right)x_{in}x_{in}'$. Under the maintained assumptions it will be the case that $\frac{n^h}{n}J_{n}\cvgto J$.

\begin{lem} \label{lem:hjort-pollard}Assume that 
 $\left\Vert \gamma_{0,n} \right\Vert_1 \cdot \sup_{x \in \mathcal{X}}\left\Vert x \right\Vert_\infty \lesssim h \cdot \log n$ for $0\leq h < 1$. Then,
$$J_{n}^{1/2}\left(\widehat{\gamma}_{n}-\gamma_{0n}\right)\wkto \mathcal{N}(0,I_d).$$
\end{lem}

Equivalently, the result implies that $\sqrt{n^{1-h}}\left(\widehat{\gamma}_{n}-\gamma_{0n}\right)\wkto \mathcal{N}(0,J^{-1})$.  This shows the sub-$\sqrt{n}$ rate of convergence.

Observe that in the example where $q_{in} \propto \exp (\alpha_{0n}+\beta_{0}w_{in})$, then this corresponds to $\alpha_{0n} = \log (C\cdot n^{-h})$ where $0\leq h < 1$ and some constant $C>0$.  More generally, the requirement ensures that the parameter (times covariate value) does not diverge too rapidly so that a central limit theorem can be applied.

\

\begin{proof}[{\bf Proof of Lemma \ref{lem:hjort-pollard}}]
The result is an extension of/corollary to Theorem 5.2 of Hjort and Pollard (1993). The convexity-based argument allows consistency and asymptotic normality to be argued in one step. Consider the random convex function
\[
A_{n}\left(s\right)=\sum_{i\le n}\log f_{i}\left(y_{in},\gamma_{0n}+J_{n}^{-1/2}s\right)-\log f_{i}\left(y_{in},\gamma_{0n}\right),
\]
where $f_i$ is the logistic function. This is minimized by $s=J_{n}^{1/2}\left(\widehat{\gamma}_{n}-\gamma_{0n}\right)$.

This can be expressed as
\[
A_{n}\left(s\right)=U_{n}'s-\frac{1}{2}s's-r_{n}\left(s\right)
\]
where\footnote{Observe that $J_{n}^{-1/2}=\sqrt{n^{1-h}}\left(\frac{n^{h}}{n}\sum_{i\leq n}q_{in}\left(1-q_{in}\right)x_{in}x_{in}'\right)^{-1/2}$.%
}
\[
U_{n}=J_{n}^{-1/2}\sum_{i\leq n}\left(y_{in}-q_{in}\right)x_{in}\rightsquigarrow\mathcal{N}\left(0,I\right),
\]
which applies by a Lindeberg-Feller central limit theorem for triangular arrays, as $\min_x q_i(x)  = \Theta( \max_x q_i(x) ) = \omega(1/n)$ by hypothesis on $\gamma_{0,n}$, $x_\ell$, and the Bernoulli probability.
Meanwhile
\[
r_{n}\left(s\right)=\sum_{i\leq n}\frac{1}{6}q_{i}\left(1-q_{in}\right)\cdot\eta_{i}\left(s'J_{n}^{-1/2}x_{in}\right)\cdot\left(s'J_{n}^{-1/2}x_{in}\right)^{3}.
\]
The proof of Theorem 5.2 of Hjort and Pollard (1993) shows $r_{n}(s)\rightarrow0$.  This exploits that $\lambda_{n}:=\max_{i\leq n}\left|J_{n}^{-1/2}x_{in}\right|\rightarrow0$, which holds by the fact that the covariates live in a compact set (making clear that this isn't a tight assumption).
\end{proof}

\

Because of relative sparsity, incidental generation is rare.  Therefore, for most of the data the preceding result directly applies. However, for a vanishing proportion of $m_\ell$-tuples, the structures are present due to incidental generation. We only need to show that this happens for a vanishing proportion of data and is asymptotically negligible.

To make this argument, out of the $n$ observations, we say that each observation is ``invalid'' (i.e., observed with measurement error such as $y_{in} = 1$ when the true value is 0) with some probability. Our claim can be written in the notation of the preceding lemma by saying that some of our $n$ data points are ``invalid'' and we will show the probability that an observation is invalid is bounded by $z_n \downarrow 0$ at a fast enough rate. Our relative sparsity assumption directly implies that $z_n \downarrow 0$.

\begin{lem}\label{lem:mismeasure} Assume the hypotheses of Lemma \ref{lem:hjort-pollard}. Assume either
\begin{enumerate}
\item every observation becomes
invalid with probability at most $z_n \downarrow 0$, or
\item an $O_p(z_n \cdot n)$ share of observations become invalid, with $z_n \downarrow 0$.
\end{enumerate}
Then the conclusion of Lemma \ref{lem:hjort-pollard} holds.
\end{lem}

\begin{proof}
Clearly the second condition is implied by the first, so we only prove the former.
Without loss of generality let $1,..,n^{*}$
denote the set of valid observations and $n^{*}+1,...,n$ the valid
observations. Note that $n^{*}$ is random and is $O_{p}\left(z_{n}n\right)$.
\[
U_{n}=J_{n}^{-1/2}\sum_{i\leq n}\left(y_{in}-q_{in}\right)x_{in}=J_{n}^{-1/2}\left[\sum_{i\leq n^{*}}\left(y_{in}-q_{in}\right)x_{in}+\sum_{n^{*}<i\leq n}\left(y_{in}-q_{in}\right)x_{in}\right].
\]
Observe that
\[
\frac{n^{h}}{n}\sum_{n^{*}<i\leq n}q_{in}\left(1-q_{in}\right)x_{in}x_{in}'=\frac{n^{h}}{n}z_{n}n=o_{p}\left(1\right).
\]
This implies
\begin{eqnarray*}
J_{n}^{-1/2}\sum_{i\leq n}\left(y_{in}-q_{in}\right)x_{in} & = & \left[\frac{n^{h}}{n}\sum_{i\leq n^{*}}q_{in}\left(1-q_{in}\right)x_{in}x_{in}'+\frac{n^{h}}{n}\sum_{n^{*}<i\leq n}q_{in}\left(1-q_{in}\right)x_{in}x_{in}'\right]^{-1/2}\\
 &  & \times\sqrt{\frac{n^{h}}{n}}\left[\sum_{i\leq n^{*}}\left(y_{in}-q_{in}\right)x_{in}+\sum_{n^{*}<i\leq n}\left(y_{in}-q_{in}\right)x_{in}\right].
\end{eqnarray*}
Thus
\[
\left[\frac{n^{h}}{n}\sum_{i\leq n^{*}}q_{in}\left(1-q_{in}\right)x_{in}x_{in}'+\frac{n^{h}}{n}\sum_{n^{*}<i\leq n}q_{in}\left(1-q_{in}\right)x_{in}x_{in}'\right]^{-1/2}\cvgto J^{-1/2}.
\]
Meanwhile, we have $\sum_{i\leq n^{*}}\left(y_{in}-q_{in}\right)x_{in}=O_{p}\left(\frac{1}{\sqrt{n^{1-h}}}\right)$
and to complete the argument
\begin{eqnarray*}
\frac{1}{\sqrt{z_{n}n}}\sum_{n^{*}<i\leq n}\left(y_{in}-q_{in}\right)x_{in} & = & \frac{1}{\sqrt{n^{1-k}}}\sum_{n^{*}<i\leq n}\left(y_{in}-q_{in}\right)x_{in}=O_{p}\left(1\right),\mbox{ where }k=h+\delta\\
\implies O\left(\frac{1}{\sqrt{n^{1-h}}}\right)\sum_{n^{*}<i\leq n}\left(y_{in}-q_{in}\right)x_{in} & = & O\left(\frac{1}{\sqrt{n^{1-k+\delta}}}\right)\sum_{n^{*}<i\leq n}\left(y_{in}-q_{in}\right)x_{in}\\
 & = & O\left(\frac{1}{n^{\delta/2}}\cdot\frac{1}{\sqrt{n^{1-k}}}\right)\sum_{n^{*}<i\leq n}\left(y_{in}-q_{in}\right)x_{in}\\ &=& O_{p}\left(n^{-\delta/2}\right)=o_{p}\left(1\right)
\end{eqnarray*}
showing the result.
\end{proof}

\

\subsubsection*{Example \ref{ext}.1}[Continued] Recall we have set $\alpha_L^n = \log (1/n^{0.7})$ and $\alpha_T^n = \log (1/n^{1.75})$,  $\beta_L = -2$ and $\beta_T = -3$.   Let $n = 100$.  Then the average degree is 3.75, the average clustering is 0.14, the fraction of nodes in the giant component is 92\% and the maximal eigenvalue of the adjacency matrix is 5.5.  Thus, the resulting graph is comparable in structure to the empirical data.

We then run 200 simulations of this process where we generate a graph and then estimate the model parameters via sequential logistic regressions.  First we regress whether a triple exists on a constant and the triad-level covariate over all $\binom{n}{3}$ observations to get $(\widehat{\alpha}_T^b,\widehat{\beta}_T^b)$, for simulation $b=1,...,100$.  Second, on the unused $ij$ pairs not in triangles we regress whether a link exists on a constant and the pair-level covariate which is a logit on all $\binom{n}{2}$ observations less used pairs.  From this we get $(\widehat{\alpha}_L^b,\widehat{\beta}_L^b)$ for $b=1,...,100$.   The results are displayed in Figure \ref{fig-logit}.

\begin{figure}[h!]
\centering
\subfloat[Links parameter]{
	\includegraphics[trim = 0.65in 0in 0.65in 0in, clip = true, scale = 0.35]{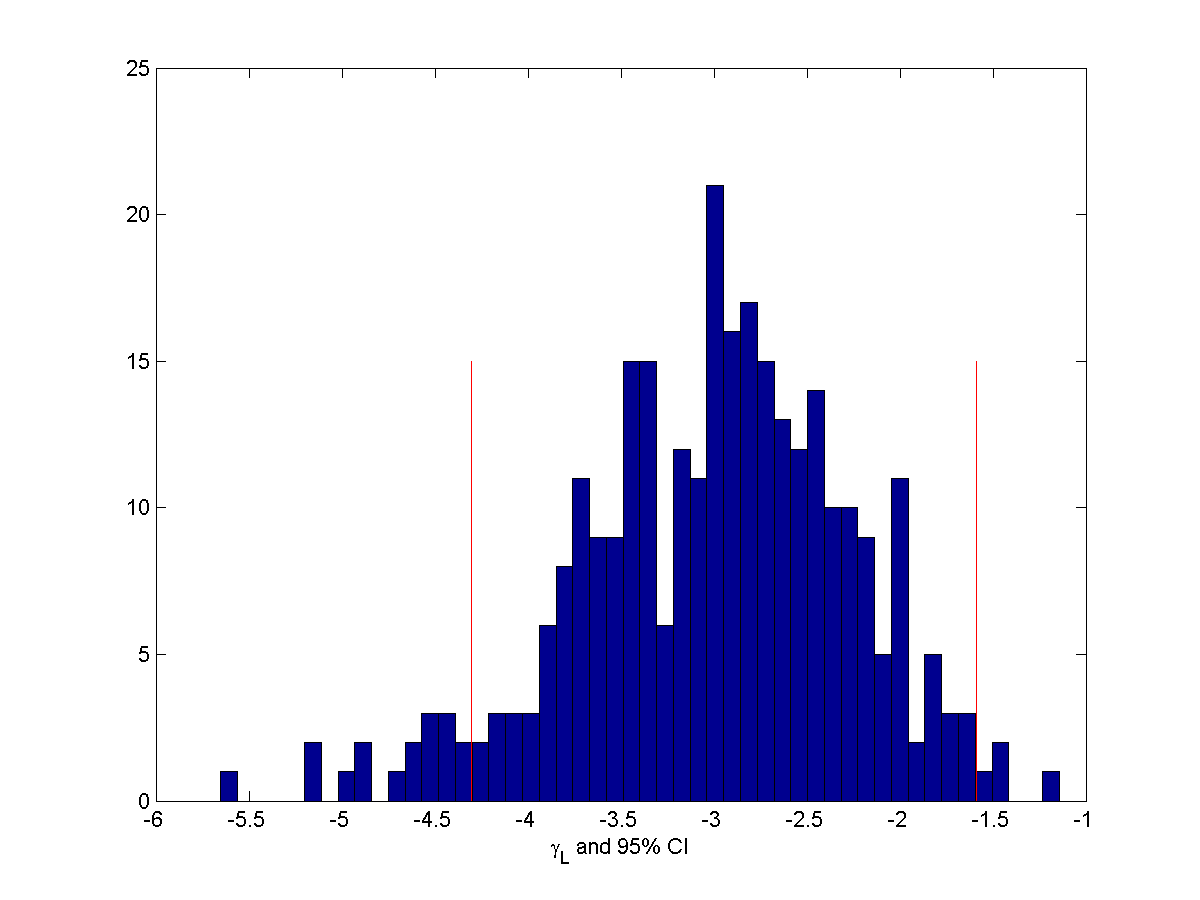}
}
\subfloat[Triads parameter]{
	\includegraphics[trim = 0.65in 0in 0.65in 0in, clip = true, scale = 0.35]{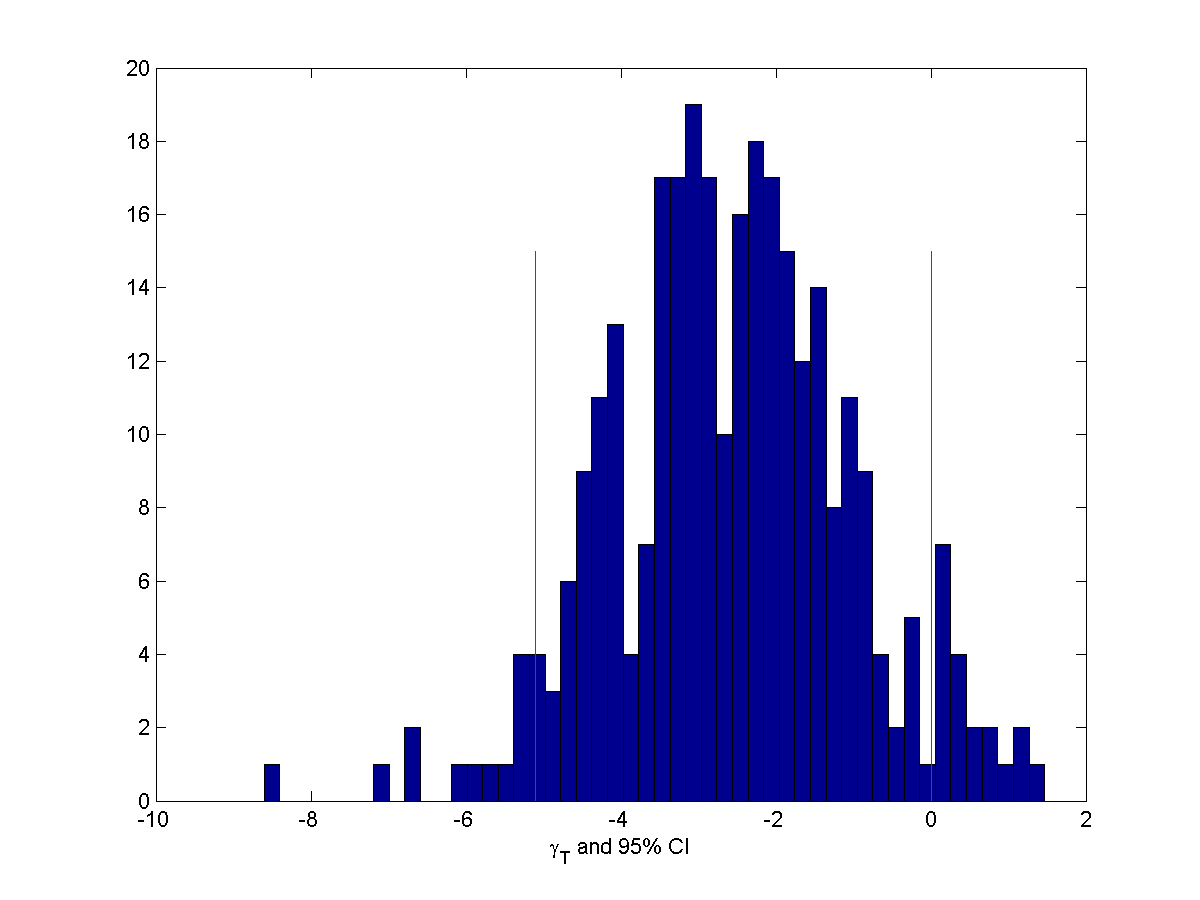}
}
	\caption{Displays the distribution of estimated parameter value as well as the median 95\% confidence interval from a simple logistic regression.}
\label{fig-logit}
\end{figure}

\noindent We show that the parameters are correctly centered and exhibit good coverage properties.

\newpage

\section{Online Appendix: Isolates, Links, Triangles Example}\label{appendILT}

\setcounter{table}{0}
\renewcommand{\thetable}{D.\arabic{table}}

\setcounter{figure}{0}
\renewcommand{\thefigure}{D.\arabic{figure}}

Here we perform some additional diagnostic exercises around the Statnet ERGM estimation from Section \ref{ergm-examplebad}.

First, we first select 17 nodes (one third) to be isolated.  Next, we generate triangles with a probability of .0014 on each possible triangle on the
nodes that are not isolated.  Finally, we generate links with probability  .0415 on the nodes that are not isolated.
Overall, this leads to networks that have on average 20 isolated nodes, 45 links, and 10 triangles (so, $\E[S_I\left(g\right)]=20, \E[S_L\left(g\right)]=45, \E[S_T\left(g\right)]=10)$).
We randomly draw 1000 different networks in this manner.

Again, maximum likelihood estimates are unique for each of the networks, and by continuity should be in a fairly tight range (given
the tight range of the generated networks as shown below).

Using standard ERGM estimation software (\texttt{statnet} via \texttt{R}, \cite{handcock:statnet}) we estimate the parameters of an ERGM with isolates, links and triangles
for each of these randomly drawn networks.  We present the estimates in Figure \ref{fig:ERGM-ILT}.

\begin{figure}[h!]
\centering
\subfloat[Isolate Parameter Estimates]{
\label{fig:isolatesERGM}
\includegraphics[width=0.33\textwidth]{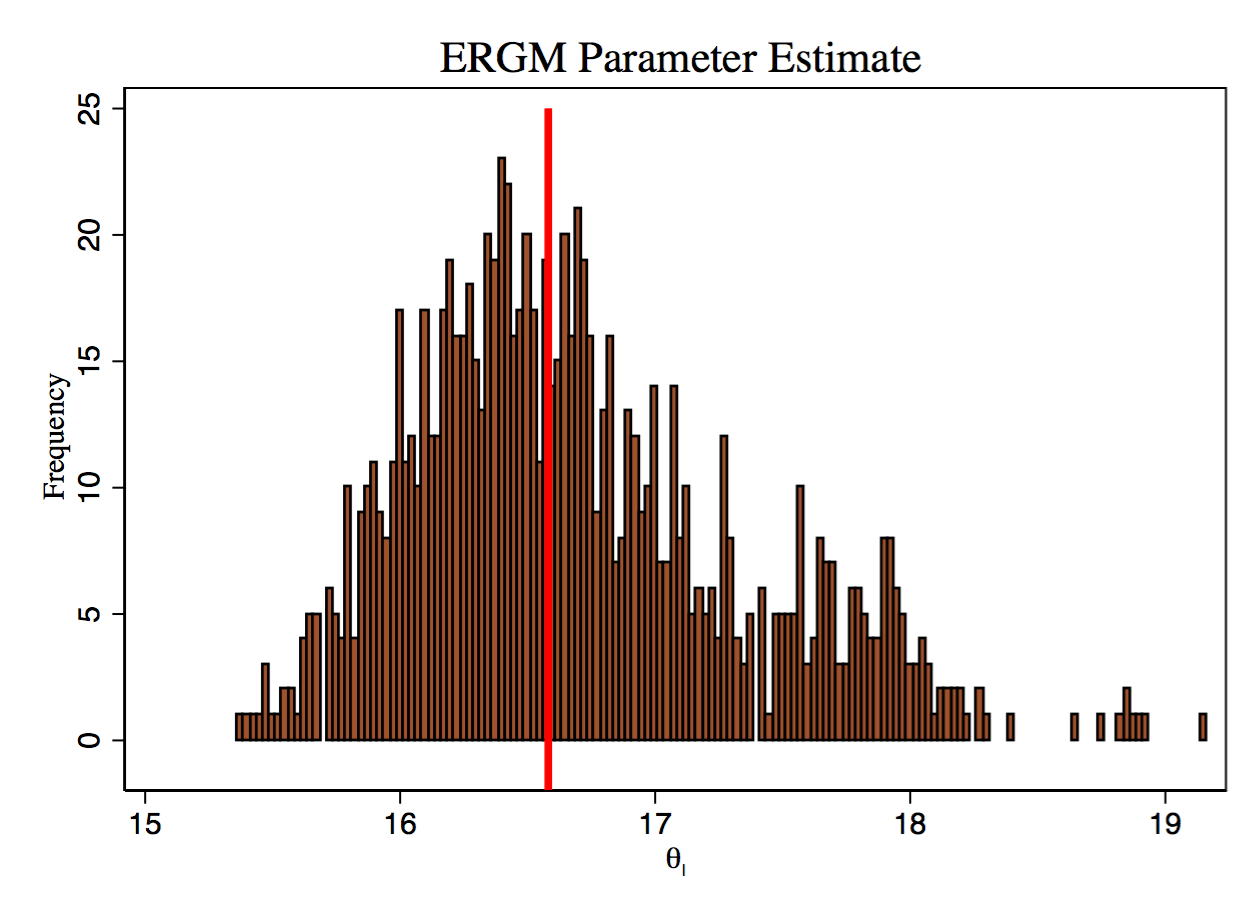}
}
\subfloat[Link Parameter Estimates]{
\label{fig:linksERGM}
\includegraphics[width=0.33\textwidth]{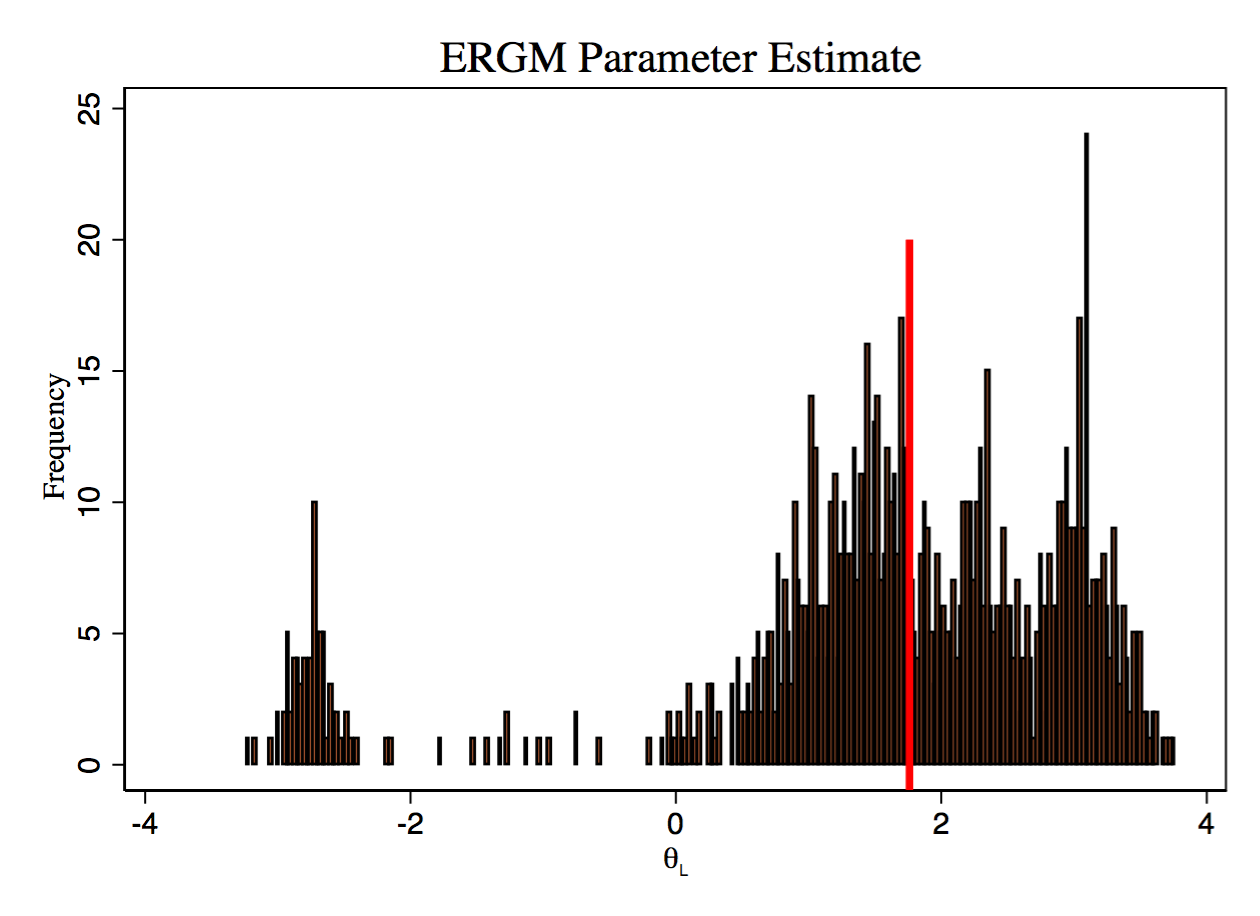}
}
\subfloat[Triangle Parameter Estimates]{
\label{fig:trianglesERGM}
\includegraphics[width=0.33\textwidth]{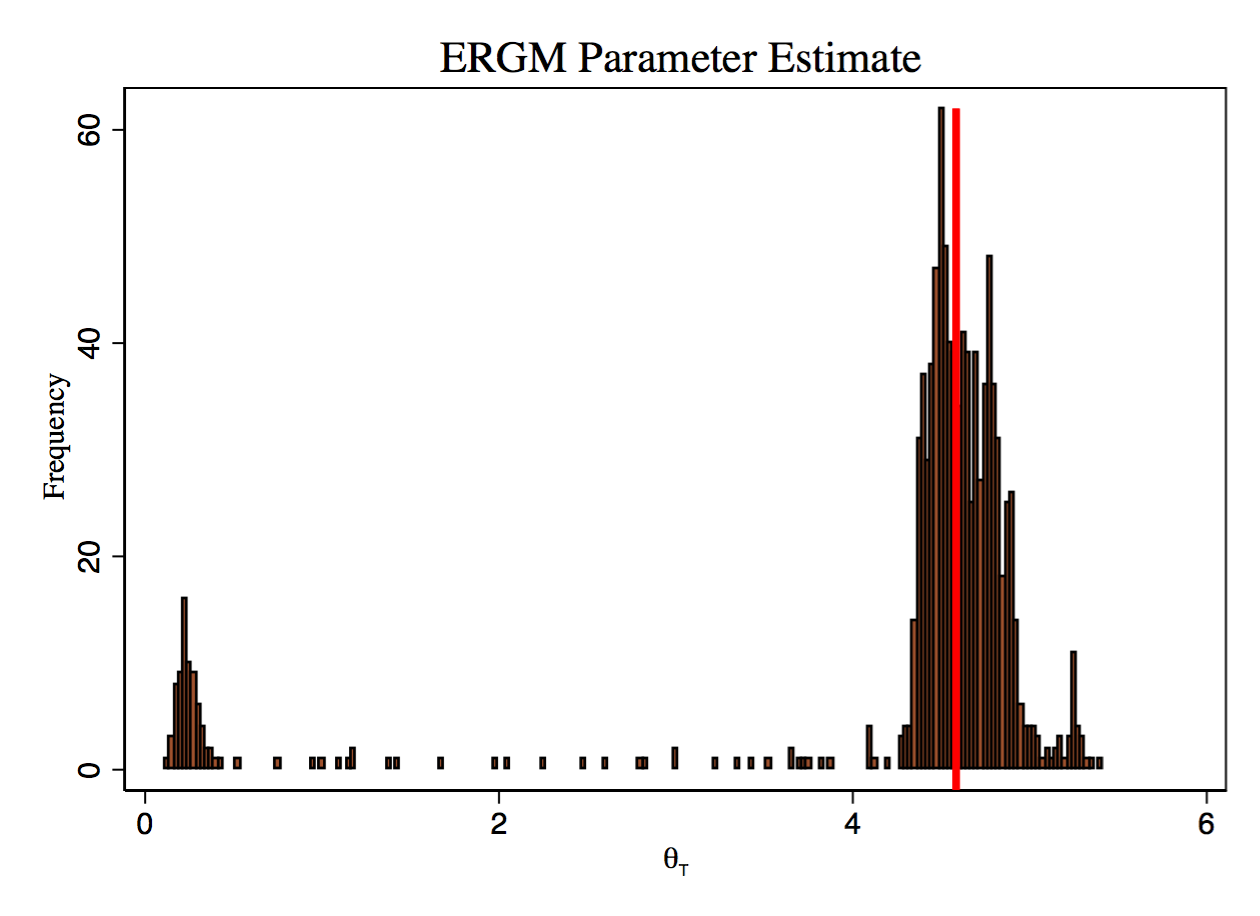}
}
\caption{\label{fig:ERGM-ILT} Standard ERGM estimation software (\texttt{statnet}) output for 1000 draws of networks on 50 nodes, with an
 average of 20 isolated nodes, 45 links, and 10 triangles.
The red lines (on top of each other) are the median left and right 95 percent confidence interval lines (which do not have appropriate coverage).}
\end{figure}

As in Figure \ref{fig:ERGM-ILT2}, but with even more noise, the estimated parameters for links and triangles cover a wide range of values, in fact with the link parameter estimates being both positive and negative and
ranging from below -3 to above 3 (Figure \ref{fig:linksERGM}) and triangles parameter estimates ranging from just above 0 to 5 (Figure \ref{fig:trianglesERGM}).  Only the isolates
parameter estimates are remotely stable (Figure \ref{fig:isolatesERGM}), but even those vary in three different regions with substantial variation.
Second, despite the enormous variation in estimated parameter values from very similar networks, the reported standard errors are quite narrow and almost always report that the parameter estimates are highly significant.  Moreover, the median left and right standard error bars essentially coincide and do not come close to capturing the actual variation.

Second, we report the distribution of the statistics from the simulated networks (Figure \ref{fig:simulstats}) - they are fairly tightly clustered about the mean values.

\begin{figure}[h!]
\centering
\subfloat[Number of Isolates]{
\label{fig:isolatessimulstats}
\includegraphics[width=0.33\textwidth]{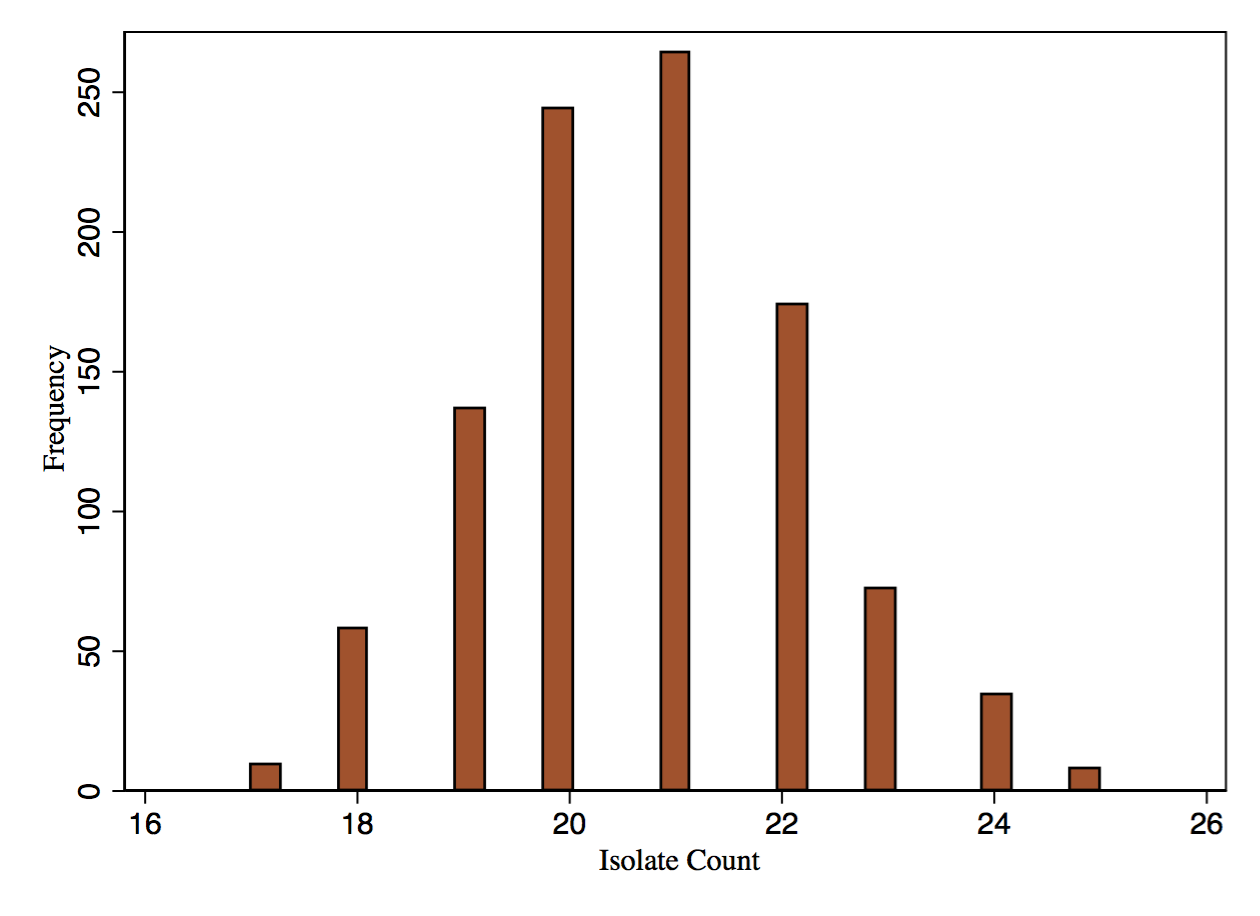}
}
\subfloat[Number of Links]{
\label{fig:linkssimulstats}
\includegraphics[width=0.33\textwidth]{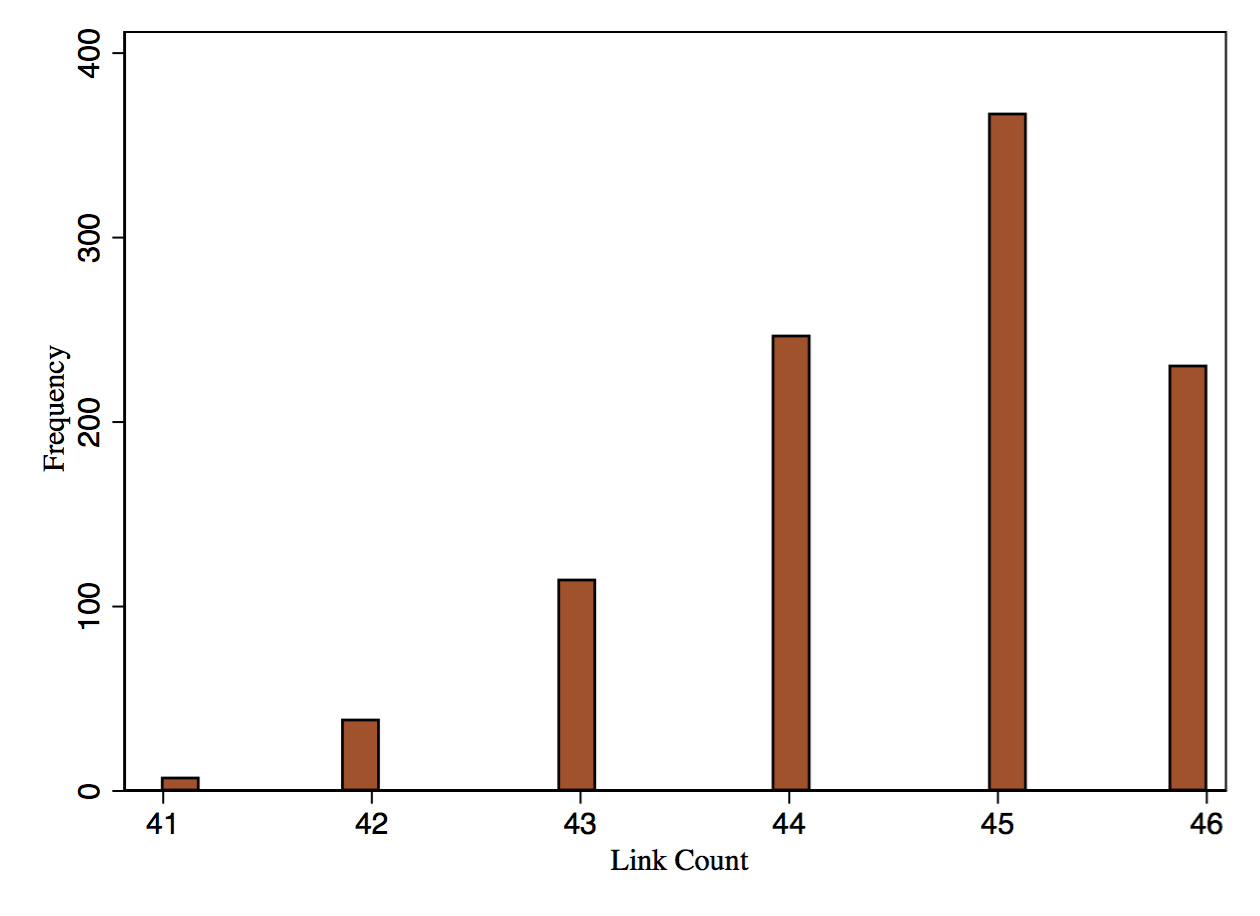}
}
\subfloat[Number of Triangles]{
\label{fig:trianglessimulstats}
\includegraphics[width=0.33\textwidth]{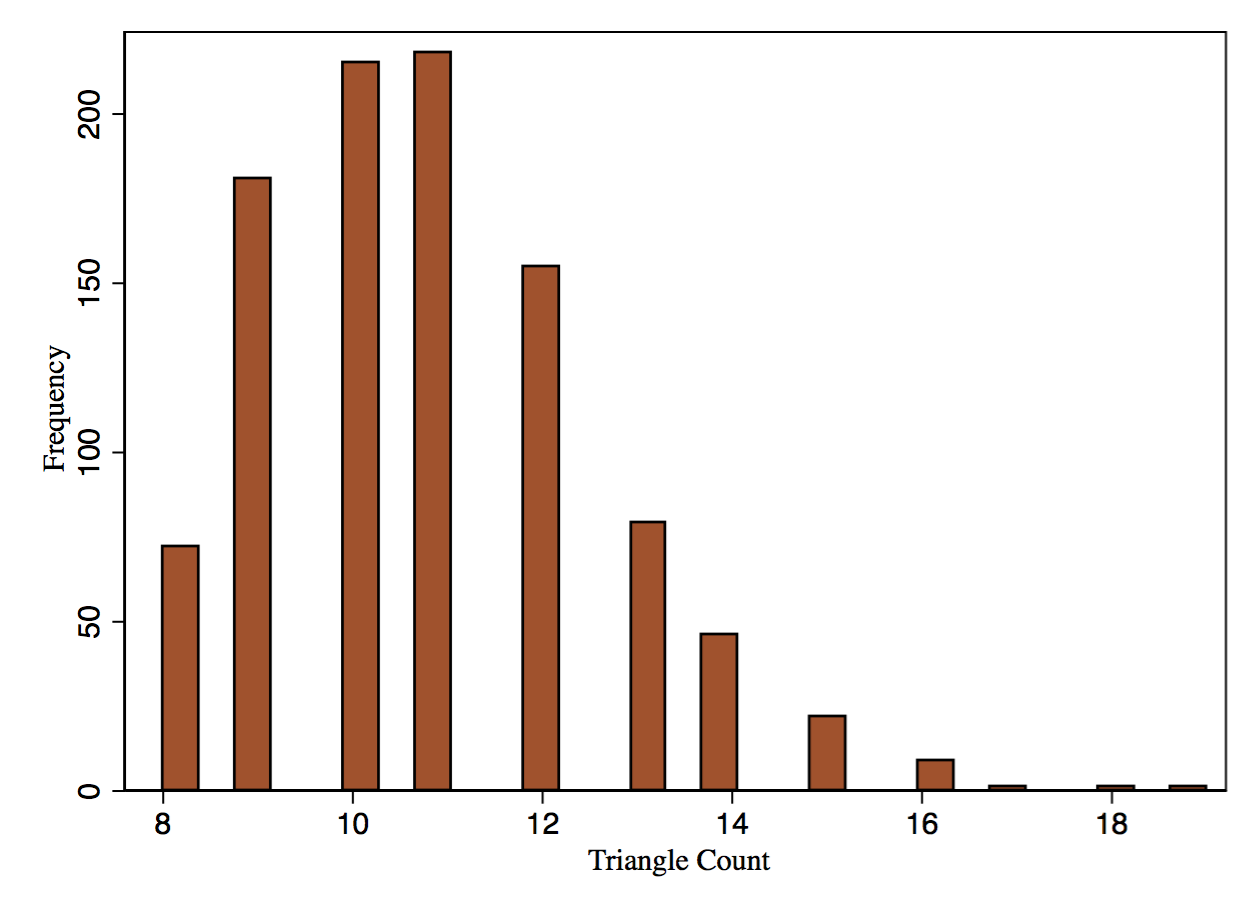}
}
\caption{\label{fig:simulstats} For the 1000 simulated networks we report the
distribution of the number of isolates, links and triangles.}
\end{figure}

Next, for each of the 1000 simulated networks, using the parameter estimates we simulate a network using Statnet's simulation command.  We then check whether the simulated networks come anywhere close to matching the original networks.  Although most of the networks turn out to have nearly 20 isolates, they generally have thousands of links and triangles.   Figure \ref{fig:recreate} looks
nothing like the counts from the original networks (Figure \ref{fig:simulstats}).

\begin{figure}[h!]
\centering
\subfloat[Number of Isolates]{
\label{fig:isolatesRecreate}
\includegraphics[width=0.33\textwidth]{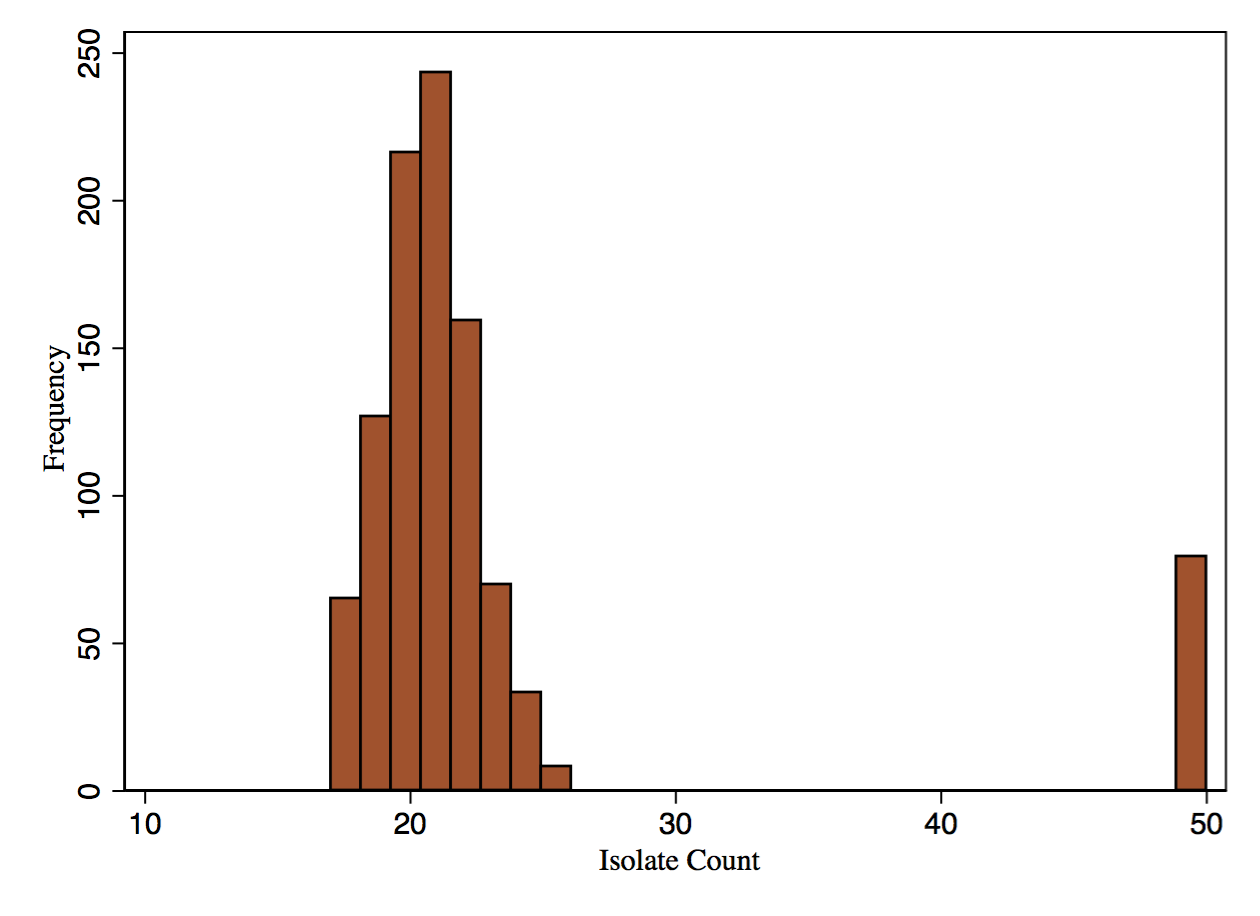}
}
\subfloat[Number of Links]{
\label{fig:linksRecreate}
\includegraphics[width=0.33\textwidth]{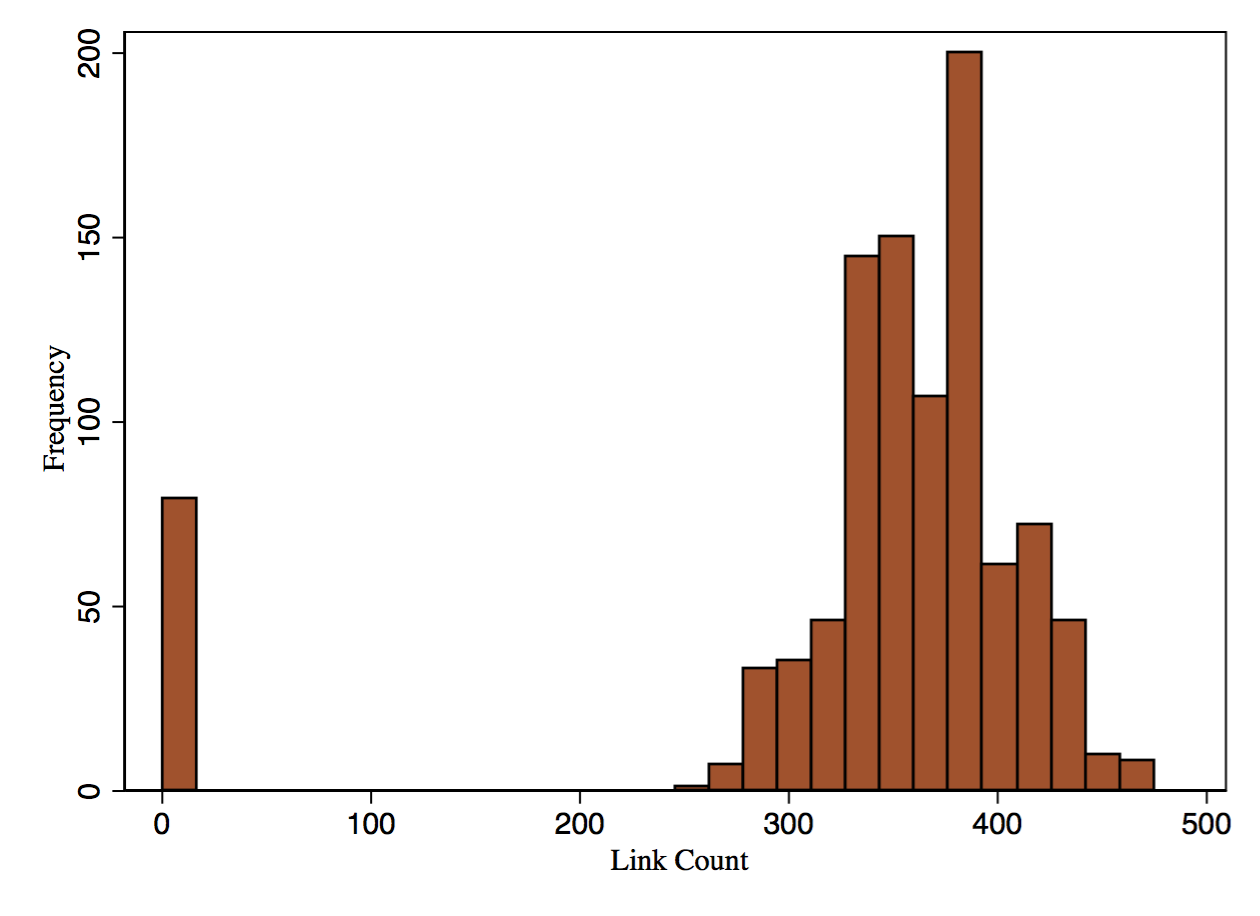}
}
\subfloat[Number of Triangles]{
\label{fig:trianglesRecreate}
\includegraphics[width=0.33\textwidth]{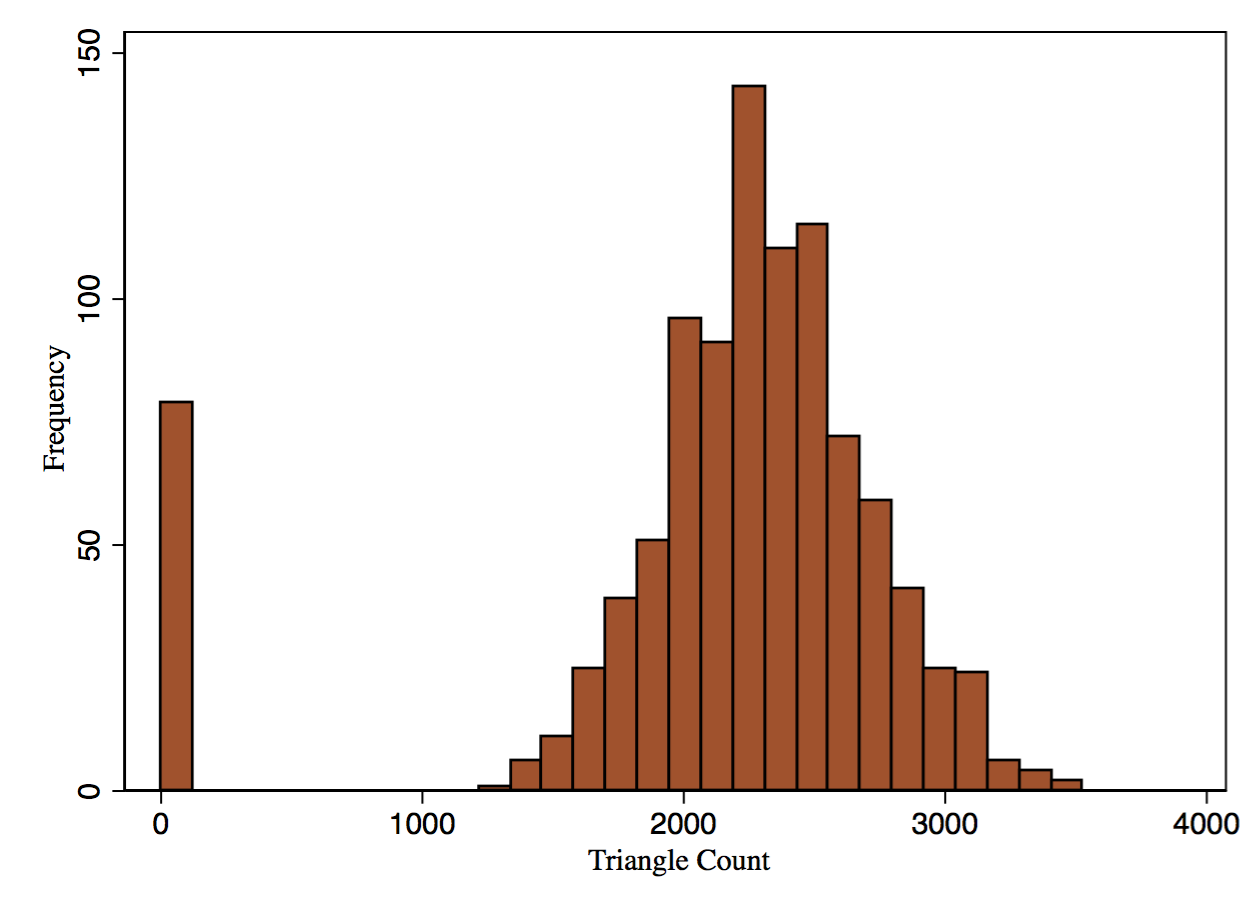}
}
\caption{\label{fig:recreate} For each of the 1000 simulated networks, using the parameter estimates from Statnet we simulate a network using Statnet's simulation command. The resulting
distribution of the number of isolates, links and triangles are pictured here.  They do not match the networks that generated them, which are pictured in Figure \ref{fig:simulstats}.}
\end{figure}

Simulating a network from an ERGM is even a more daunting task than estimating parameters from one, as there is
no obvious network from which to seed the simulation procedure, and again there are far too many from which to calculate likelihoods.  This is another advantage of SUGMs and count SERGMs, which are not only easily estimated but also easily simulated.

\begin{figure}[h!]
\centering
\subfloat[Estimated $p_I$: $\widehat{p}_I$]{
\label{fig:isolatesSUGM}
\includegraphics[width=0.33\textwidth]{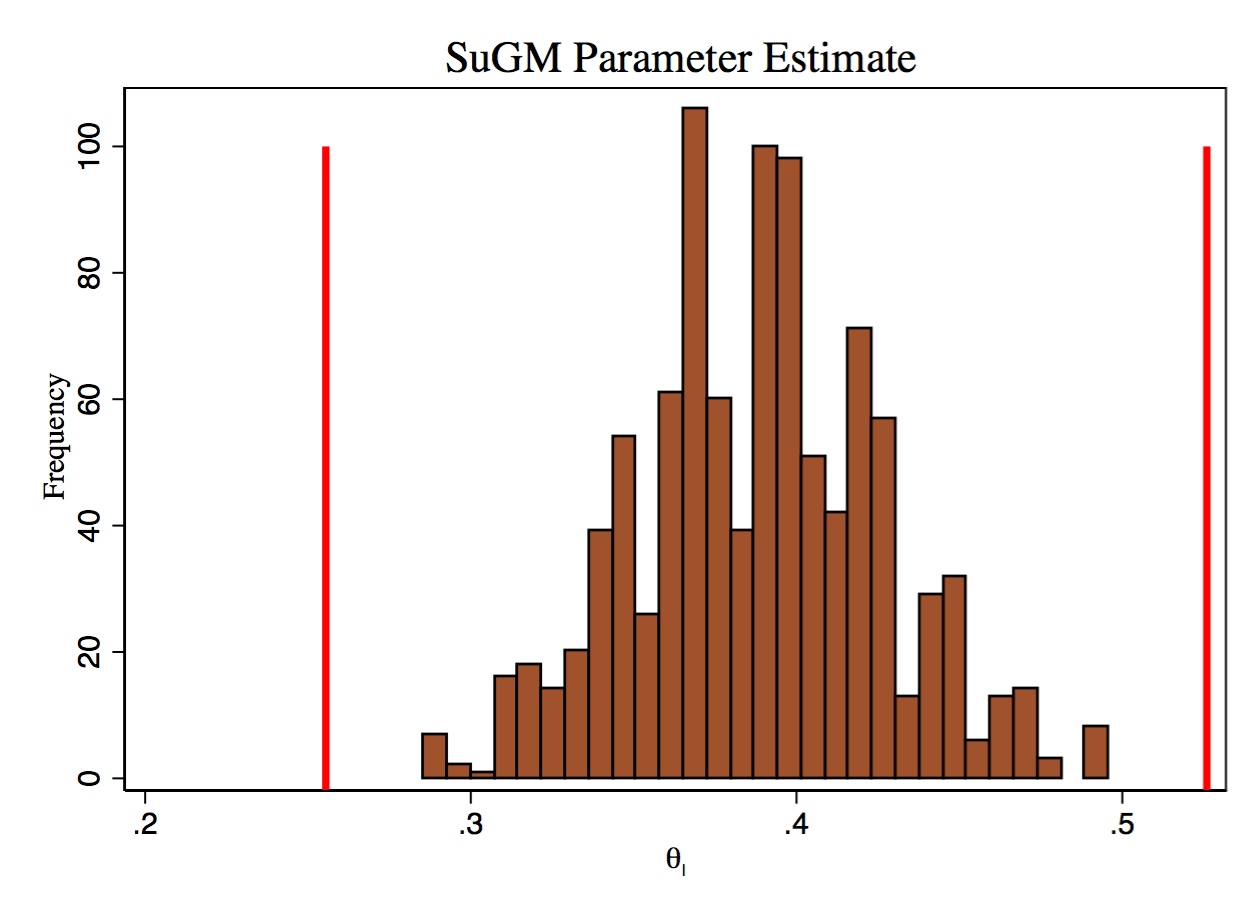}
}
\subfloat[Estimated $p_L$: $\widehat{p}_L$]{
\label{fig:linksSUGM}
\includegraphics[width=0.33\textwidth]{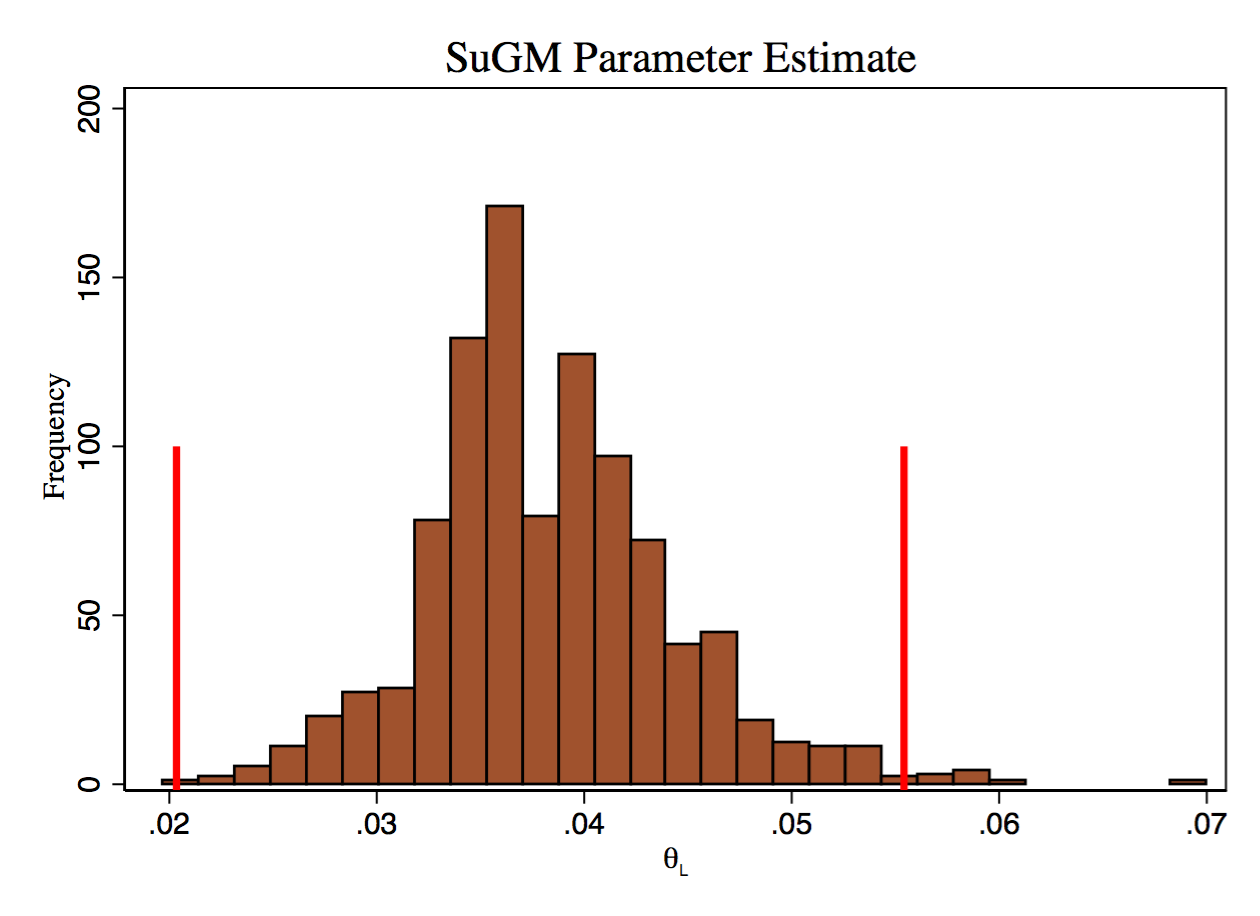}
}
\subfloat[Estimated $p_T$: $\widehat{p}_T$]{
\label{fig:trianglesSUGM}
\includegraphics[width=0.33\textwidth]{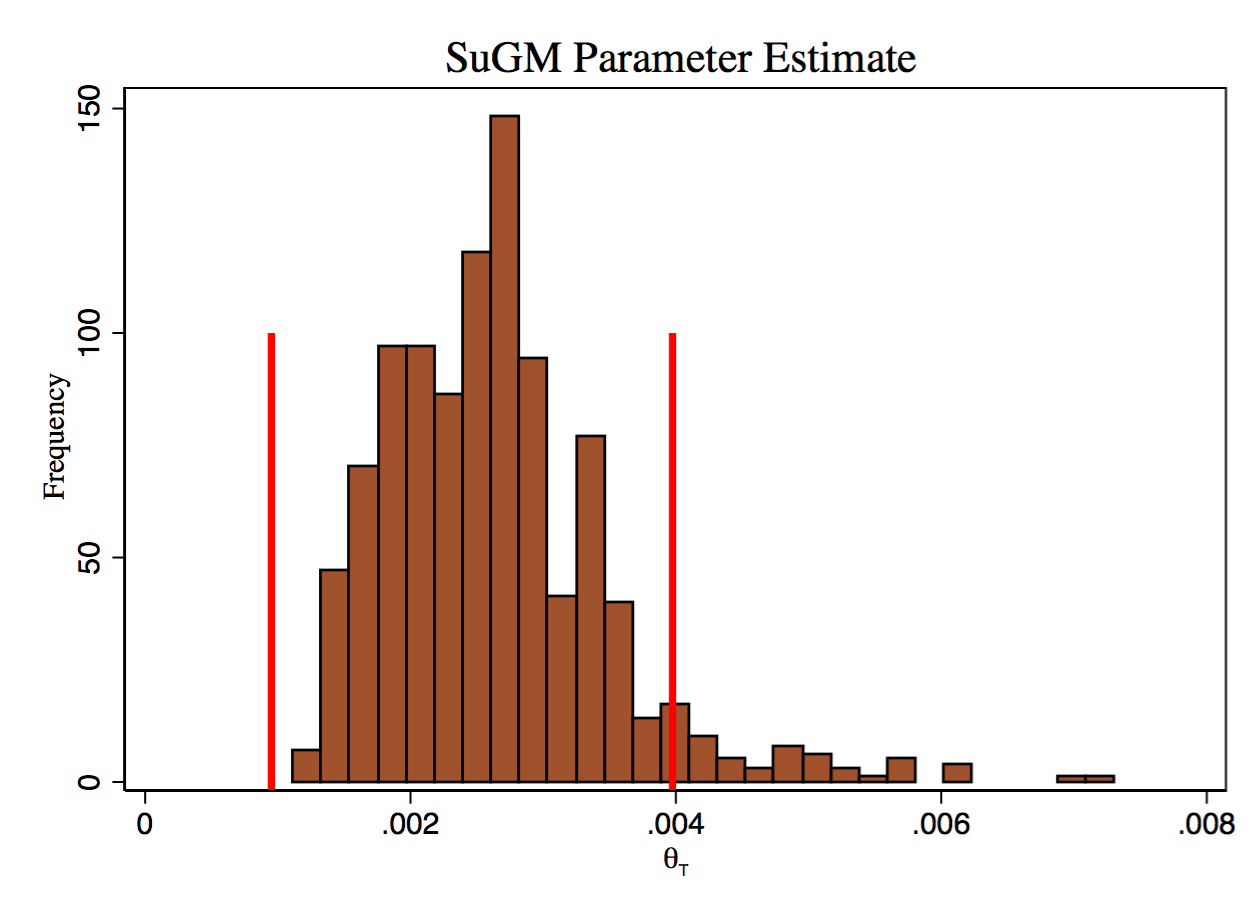}
}
\caption{\label{fig:simulSUGM} For each of the 1000 simulated networks, we solve for estimated probabilities of isolates, links and triangles for a SUGM as defined in (\ref{sugm-ex2}).}
\end{figure}

\begin{figure}[h!]
\centering
\subfloat[Estimated $\theta_I$: $\widehat{\theta}_I$]{
\label{fig:isolatesSERGM}
\includegraphics[width=0.33\textwidth]{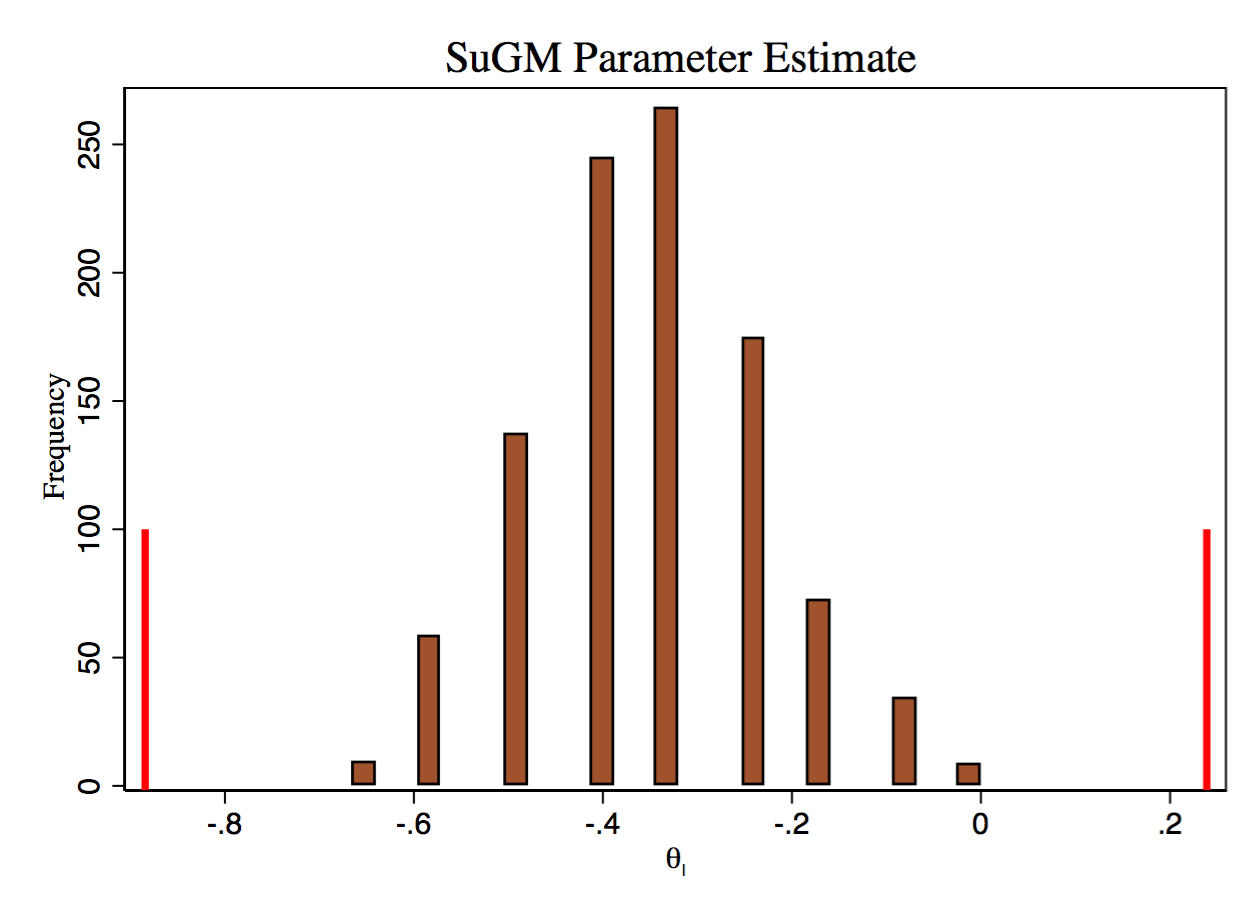}
}
\subfloat[Estimated $\theta_L$: $\widehat{\theta}_L$]{
\label{fig:linksSERGM}
\includegraphics[width=0.33\textwidth]{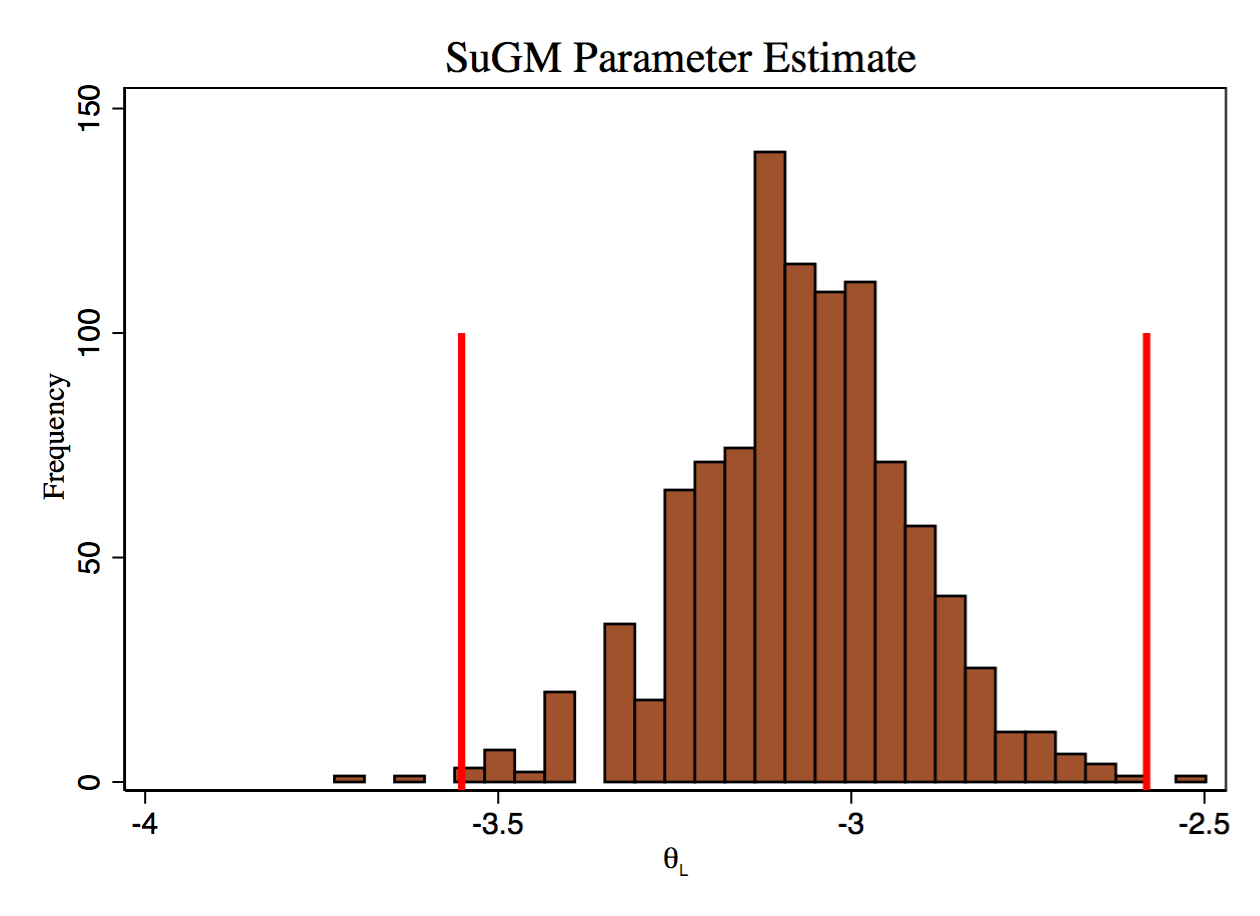}
}
\subfloat[Estimated $\theta_T$: $\widehat{\theta}_T$]{
\label{fig:trianglesSERGM}
\includegraphics[width=0.33\textwidth]{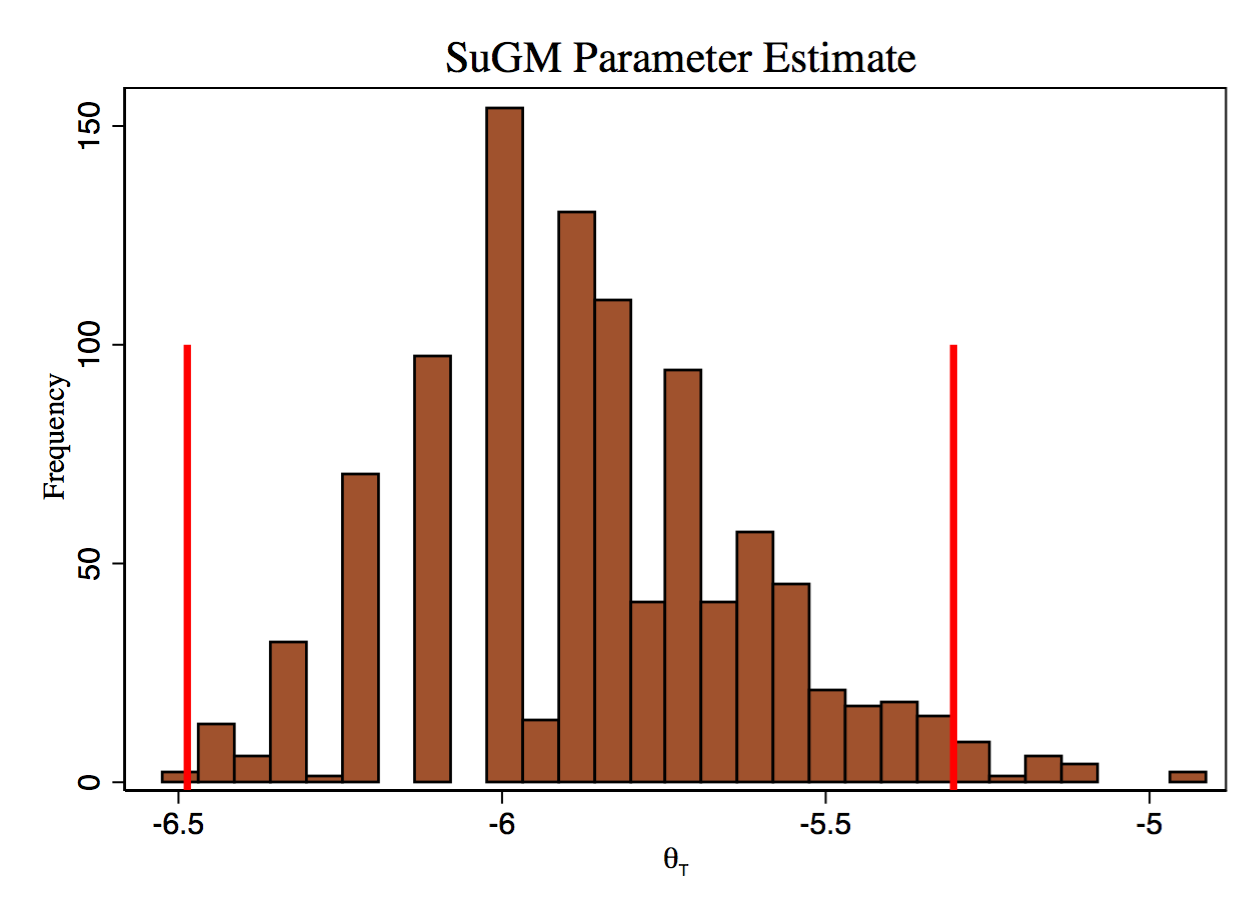}
}
\caption{\label{fig:simulSERGM} For each of the 1000 simulated networks, we solve for estimated SERGM parameters for isolates, links and triangles for as defined in (\ref{sergm-ex2}).}
\end{figure}

\newpage

\
\section{Online Appendix: Additional Consistency Results}\label{addconsistency}

\setcounter{table}{0}
\renewcommand{\thetable}{E.\arabic{table}}

\setcounter{figure}{0}
\renewcommand{\thefigure}{E.\arabic{figure}}

\subsection*{General Consistency Results}

Consider a sequence of SERGMs $(S^n,K^n_{S^n},A_n, \beta^n)$ as defined in (\ref{eq: sergm}).

A sequence of SERGMs  is {\sl expectations-identified} with respect to a
sequence of diagonal matrices $C_n>0$ with positive diagonal entries
if there exists $\gamma>0$ such that
$$ \left\vert C_n\E_{\beta}[S^n]  - C_n\E_{\beta^n}[S^n] \right\vert > \gamma \left\vert \beta-\beta^n \right\vert$$
for all $n$.\footnote{Here the subscript notation $\E_\beta[S^n]$ indicates that the expectation takes the probability
to be specified by (\ref{eq: sergm}) with parameters $(S^n,K^n_{S^n},A_n,\beta)$.}

Expectations identification is an intuitive condition that requires that different parameters distinguish themselves with different means.  It is a sort of minimal
condition since if two different parameter values generate very similar expected statistics, then observing the realized statistic will not allow us to distinguish the parameters.

A sequence of SERGMs  is {\sl concentrated}  with respect to a
sequence of diagonal matrices $C_n>0$ with positive diagonal entries
if
$$C_n ( S^n - \E_{\beta^n}[S^n] ) \cvgto 0\ {\rm for} \  \beta^n \in \mathcal{B},$$
where $\mathcal{B}$ is a set of admissible parameters.

The concentration condition requires that there is some normalization ($C_n$) for which the statistics will concentrate around their means.  As we have not scaled statistics we have to allow for some renormalizations.\footnote{Notice that the exponential term in the associated likelihood can be written $\exp\left(S'\beta\right)=\exp\left(S'C_{n}C_{n}^{-1}\beta\right)$
and we are interested in the associated parameters $\beta$, not $C_{n}^{-1}\beta$ which will typically trail off to infinity at polynomial rates in $n$.}

Note that the choice of $C_n$ in the following theorem links them across the two conditions.  To guarantee concentration there has to exist a sequence of $C_n$ that goes to 0 fast enough, while to guarantee expectations identification they cannot go to 0 too quickly.  So, the $C_n$'s identify the rate at which the statistics approach their means.  The key to verifying consistency is then seeing whether there exists such sequences for which both conditions hold simultaneously.

\smallskip

\begin{prop}
\label{thm-direct}
If a sequence of SERGMs $(S^n,K^n_{S^n},A_n,\beta^n)$  is expectations-identified  and concentrated
with respect to some $C_n$,
then $ | \widehat{\beta}^n (S^n) -\beta^n | \cvgto 0.$
\end{prop}

\smallskip

\noindent The proof of Proposition \ref{thm-direct} is relatively routine.

It is also useful to state a ratio version of Proposition \ref{thm-direct} to address cases where parameters are, for instance, close to 0.  Thus, we wish to have a stronger notion of consistency, not only requiring that the estimator approaches the true parameter, but that it does so in terms of a ratio.  This requires corresponding definitions of concentration and identification that are ratio based.  We do this in Proposition \ref{thm-ratio}.

The result is fairly tight in that a version of a converse holds as well.  In particular, the following result holds.\footnote{We state a version for Proposition \ref{thm-direct}, and the analog
for ratio-consistency of the type in Proposition \ref{thm-ratio} in the appendix is left to the reader.}

A sequence of SERGMs $(S^n,K^n_{S^n},A_n,\beta^n)$,  is {\sl rate-expectations-identified} with rates given by a sequence of diagonal matrices $C_n>0$ with positive diagonal entries
if there exist $\gamma_H>\gamma_L>0$ such that
$$\gamma_H |\beta-\beta^n| >  | C_n\E_{\beta}[S^n]  - C_n\E_{\beta^n}[S^n] | > \gamma_L |\beta-\beta^n|$$
for all $n$.

Rate-expectations-identification is a condition that says that the sequence of diagonal matrices $C_n$ accurately captures the rate at which the expected statistics $ \E_{\beta}[S^n]$ nears $\E_{\beta^n}[S^n] | $ as we let the vector of parameters $\beta$ approach $\beta^n$.

\smallskip

\begin{prop}
\label{thm-direct2}
If a sequence of SERGMs $(S^n,K^n_{S^n},A_n,\beta^n)$  is rate-expectations-identified
with rates $C_n$,
then the random vectors are consistent ($|\widehat{\beta}^n (S^n) -\beta^n|\cvgto 0$) \emph{if and only if}
the sequence is concentrated with respect to $\beta^n$ and the same sequence $C_n$.
\end{prop}

\bigskip

Below, we also discuss a characterization of consistency in terms of variance of the statistics in greater detail.  The argument is similar to those for standard estimators, e.g. Amemiya (1973): For consistency, we need enough variation so that the system accumulates information and concentrates around its mean.  This corresponds to needing the norm of the variance matrix tending to infinity, just as we would in typical regression-like applications.

\

\begin{proof}[{\bf Proof of Proposition \ref{thm-direct}}]
Recall that the MLE  $\widehat{\beta}^n (s)$ is the $\beta$ that solves
$$ s  = \E_{\beta}[S^n]  .$$
Thus, since concentration implies that
$$  C^n( S^n - \E_{\beta^n}[S^n] ) \cvgto 0,$$
it follows that
$$ C^n( \E_{\widehat{\beta}^n (S^n)}[S^n] - \E_{\beta^n}[S^n] ) \cvgto 0.$$
Given that expectations identification implies that
$$  | C^n(\E_{\beta}[S^n]  - \E_{\beta^n}[S^n]) | > \gamma |\beta-\beta^n|$$ for all $n$,
it follows that
\[
|\widehat{\beta}^n (S^n) -\beta^n| \cvgto 0
\]
as claimed.
\end{proof}

\

These results can be rephrased in terms of standard properties of extremum estimators and identifiable uniqueness.\footnote{Following, e.g., Gallant and White (1988), we say the sequence $\beta^n$ is identifiably unique on $\mathcal{B}$ if
$$\lim \inf_{n \rightarrow \infty}\inf_{\beta\in B_{\beta^n}(\epsilon)} Q_{0,n}(\beta^n)-Q_{0,n}(\beta) > 0.$$}

\begin{prop}\label{thm-std}
If a sequence of SERGMs $(S^n,K^n_{S^n},A_n,\beta^n)$  satisfies concentration,  rate-expectations identification, and the $\beta^n $ all lie in a compact set, then
$|\widehat{\beta}^n (S^n) -\beta^n|\cvgto 0.$
\end{prop}

\smallskip

\begin{proof}[{\bf Proof of Proposition \ref{thm-std}}]
Let
\[
m_{n}\left(\beta\right):=C_{n}\left(S^{n}-{\rm E}_{\beta}\left[S\right]\right).
\]
\noindent The objective function is $Q_n(\beta) := m_n(\beta)'m_n(\beta)$.

First, we want to show that the moment function satisfies a uniform
law of large numbers: $\sup_{\beta\in\mathcal{B}}\left\Vert m_{n}\left(\beta\right)\right\Vert =o_{p}\left(1\right)$.
By concentration, we have pointwise convergence of $m_{n}\left(\beta\right)$
to zero in probability.

Therefore, we need to only check stochastic equicontinuity: that for every $\eta>0$
there is a $\delta>0$ with
\[
{\rm P}\left(\sup_{\left\Vert \beta - \beta' \right\Vert <\delta}\left|m_{n}\left(\beta\right)-m_{n}\left(\beta'\right)\right|>\eta\right)<\eta.
\]

\noindent A sufficient condition (Andrews, 1994) is if a H\"{o}lder condition is satisfied:
\[
\left|m_{n}\left(\beta\right)-m_{n}\left(\beta'\right)\right|\leq X_{n}\cdot\left\Vert \beta-\beta'\right\Vert
\]
where $X_{n}$ is some $O_{p}\left(1\right)$ random variable. This is directly guaranteed by rate-expectations identification with $X_n = \gamma_H$.

Second, one can check that expectations identifiability guarantees identifiable uniqueness.  Together with compactness of $\mathcal{B}$, this implies the above implies that $\betahat$ is consistent for $\beta^n$.
\end{proof}

\medskip

\begin{proof}[{\bf Proof of Proposition \ref{thm-direct2}}]
Recall that the MLE  $\widehat{\beta}^n (s)$ is the $\beta$ that solves
$$ s  = \E_{\beta}[S^n]  .$$
Given Proposition \ref{thm-direct}, we need only show that if consistency holds then
concentration must also hold.

Given that rate-expectations identification implies that
$$ a^n | \E_{\beta}[S^n]  - \E_{\beta^n}[S^n] | < \gamma_H |\beta-\beta_0|$$ for all $n$,
it follows that if if consistency holds so that
\[
|\widehat{\beta}^n (S^n) -\beta^n| \cvgto 0,
\]
then it must be that
$$ a^n | \E_{\widehat{\beta}^n}[S^n]  - \E_{\beta^n}[S^n] |
 \cvgto 0. $$
 This implies that
$$ a^n |S^n  - \E_{\beta^n}[S^n] |
 \cvgto 0 ,$$
which implies concentration.
\end{proof}

\subsection*{Ratio Convergence Results}

A sequence of SERGMs $(S^n,K^n_{S^n},A_n,\beta^n)$  is {\sl ratio-expectations-identified}
with respect to a sequence of diagonal $C^n$ with positive diagonal entries
if there exists $\gamma>0$ such that
$$ \left\vert \frac{C^n_{hh}\E_{\beta}[S^n_h]}{C^n_{hh}\E_{\beta^n}[S^n_h]}  - 1 \right\vert > \gamma \left\vert \frac{\beta_h}{(\beta^n)_h}-1\right\vert$$
for all $n$ and $h$.

A sequence of SERGMs $(S^n,K^n_{S^n},A_n,\beta^n)$  is {\sl ratio-concentrated}
with respect to a sequence of diagonal $C^n$ with positive diagonal entries
if
$$  \frac{C^n_{hh} S^n_h}{C^n_{hh}\E_{\beta_0}[S^n_h]}-1  \cvgto 0$$
for each $h$.

\begin{prop}
\label{thm-ratio}
If a sequence of SERGMs $(S^n,K^n_{S^n},A_n,\beta^n)$   is { ratio-expectations-identified} and ratio-concentrated
with respect to a sequence of diagonal $C^n$ with positive diagonal entries,
then
$\frac{\widehat{\beta}^n (S^n)_h}{  (\beta^n)_h}\cvgto 1$
for each $h$.
\end{prop}

\bigskip

\begin{proof}[{\bf Proof of Proposition \ref{thm-ratio}}]
Again, recalling that the MLE  $\widehat{\beta}^n (s)$ is the $\beta$ that solves
$$ s  = \E_{\beta}[S^n].$$
Thus, since ratio-concentration implies that
$$  \frac{C^n_{hh}S^n_h}{\E_{C^n_{hh}\beta^n}[S^n_h]}-1  \cvgto 0$$
it follows that
$$   \frac{C^n_{hh}\E_{\widehat{\beta}^n (S^n)}[S^n_h]}{C^n_{hh}\E_{\beta^n}[S^n_h]}  - 1 \cvgto 0.$$
Given that ratio expectations identification implies that
$$ | \frac{C^n_{hh}\E_{\beta}[S^n_h]}{C^n_{hh}\E_{\beta^n}[S^n_h]}  - 1 | > \gamma |\frac{\beta_h}{(\beta^n)_h}-1|$$ for all $n,h$,
it follows that
\[
 \frac{\widehat{\beta}^n (S^n)_h}{(\beta^n)_h}-1  \cvgto 0
\]
or
\[
\frac{\widehat{\beta}^n (S^n)_h}{  (\beta^n)_h}\cvgto 1,
\]
as claimed.\end{proof}

\newpage
\section{Online Appendix: Extension of Table \ref{tab-emprops}}\label{simsappend}

\setcounter{table}{0}
\renewcommand{\thetable}{F.\arabic{table}}

\setcounter{figure}{0}
\renewcommand{\thefigure}{F.\arabic{figure}}

Here we present an extension of the analysis in Table \ref{tab-emprops}.
Instead of simply controlling for ``close'' versus ``far'' links on the dimensions of caste and GPS, we allow for a considerably richer specification. The goal here is to show that even when we control, flexibly, for a rich set of covariates, a link-based model exploiting the observable homophily is unable to replicate key features of observed networks. To do this, we estimate a link-based model within each village using the following vector of controls:
\begin{itemize}
	\item Geographic distance between households,
	\item Square of geographic distance between households,
	\item Households are of different caste,
	\item Difference in number of rooms household has,
	\item Square of difference in number of rooms,
	\item Difference in number of beds,
	\item Square of difference in number of beds,
	\item Difference in quality of electricity,
	\item Square of difference in quality of electricity,
	\item Difference in latrine quality,
	\item Square of difference in latrine quality,
	\item Whether or not both households have the same status in terms of owning or renting their house.
\end{itemize}
\noindent We use a logistic regression for this estimation.

The estimated a vector of regression coefficients for each village capture how characteristics of a dyad correspond to linking probabilities.
This gives a predicted probability that each household is linked to each of the other households in the village.  We use these predicted probabilities to generate 100 simulated networks per village and study the characteristics of the resulting networks.
These are presented in column [3] of Table \ref{tab-append}.

\begin{table}[h]
\caption{Estimation of Additional Models: Extension of Table \ref{tab-emprops} }\label{tab-append}
\begin{center}
\includegraphics[trim = 15mm 195mm 0mm 26mm, clip = true, scale = 0.8]{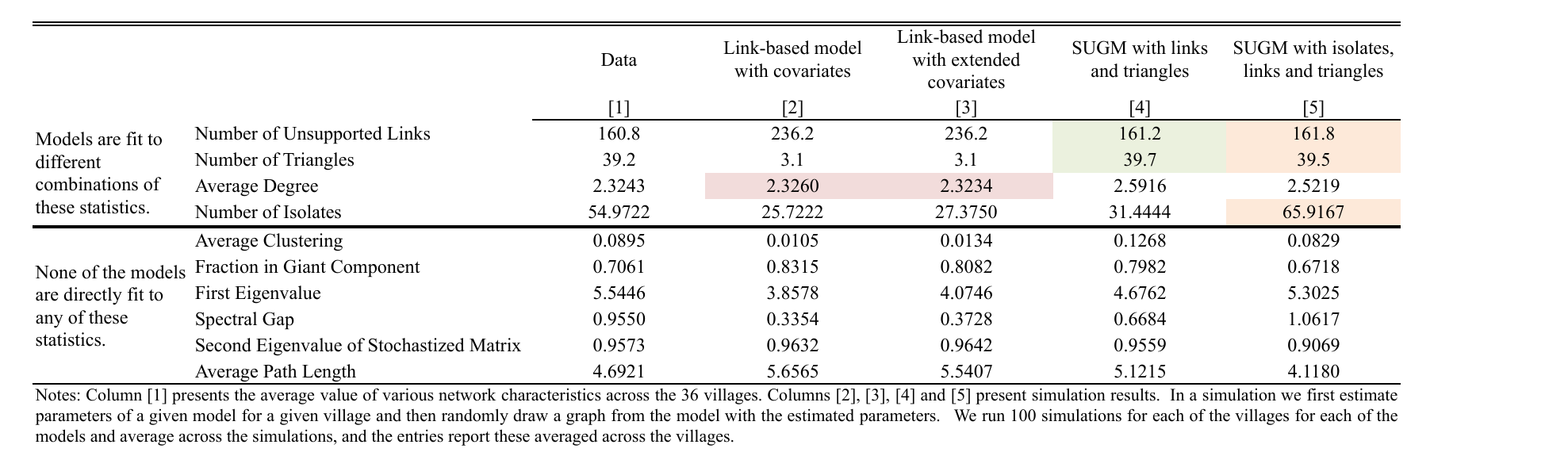}
\par\end{center}
\end{table}

Column [3] contains the statistics from the enriched link-based model, while the remainder of the table is exactly the same as what is presented in the body of the paper.  Adding over 12 parameters to flexibly control for demographic attributes
makes almost no difference in generating network characteristics that match the observed data, providing
very small improvements, and still not coming close to doing as well as the simple SUGMs.  Moreover, since the specification developed here makes use of considerably richer data than those used in the two candidate SUGM models, it suggests that by decomposing a network into a tapestry of random structures (triangles, links and even isolates), considerable value is added in modeling higher order features of networks in a parsimonious way.

In Table \ref{tab-emprops-appendix}, we show the results of Table \ref{tab-emprops} adding standard errors, to show that the SUGM models better replicate patterns in the data.

\begin{table}[h]
\caption{Network Properties: Extended}\label{tab-emprops-appendix}
\begin{center}
\includegraphics[trim = 15mm 150mm 0mm 29mm, clip = true, scale = 0.9]{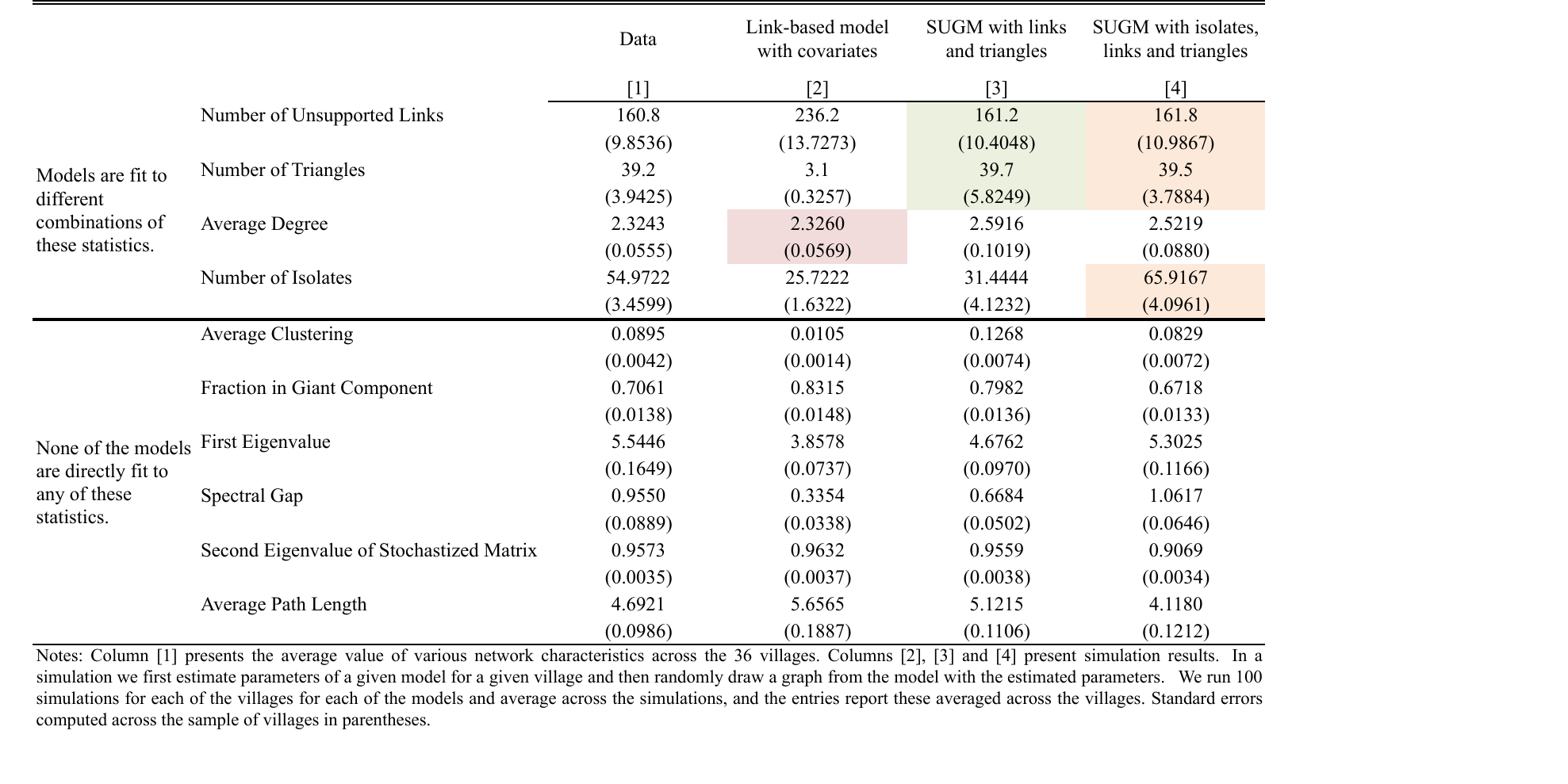}
\par\end{center}
\end{table}

\end{document}